\documentclass[pra,aps,superscriptaddress,showpacs, twocolumn, notitlepage,nofootinbib,10pt]{revtex4-1}

\usepackage{graphicx,color}
\usepackage{amsmath,amsfonts,enumerate,amsthm,amssymb,bbm}
\usepackage{hyperref}
\usepackage{multirow}
\usepackage{palatino}
\usepackage{bbold}
\usepackage{braket}
\usepackage{tikz}
\usepackage[utf8]{inputenc}
\usepackage[T1]{fontenc}
\usepackage{subfig}
\usepackage{xcolor}

\hypersetup{
	colorlinks=true,         % false: boxed links; true: colored links, false is default
	linkcolor=blue,          % color of internal links, red is default
	citecolor=blue,          % color of links to bibliography
	urlcolor=blue          % color of external links, cyan is default
}

\def\h{\mathcal H}
\def\uu{\mathfrak{U}}

\def\iden{\mathbb{1}}
\def\Tr{\mbox{tr}}
\def\hc{^\dagger}

\newcommand{\vcb}[2][]{\mathinner{\langle vec(#2)\rvert}_{#1}}
\newcommand{\vc}[2][]{\mathinner{\lvert vec(#2)\rangle}_{\hspace{-0.1em}#1}}

\newcommand{\clebsch}[6]{\langle #1,#2;#3,#4\hspace{0.5px} | \hspace{0.5px}#5,#6\rangle}

\def\>{\rangle}
\def\<{\langle}

\def\E{ {\mathcal E} }
\def\P{ {\mathcal P} }
\def\H{ {\mathcal H} }
\def\M{ {\mathcal M} }
\def\U {{\mathcal U}}
\def\V {{\mathcal V}}
\def\R {{\mathcal R}}
\def\G {{\mathcal G}}

\def\F{ {\mathcal F} }
\def\A{ {\mathcal A} }
\def\B{ {\mathcal B} }
\def\O{ {\mathcal O} }
\def\P{ {\mathcal P} }
\def\D{ {\mathcal D} }
\def\T{ {\mathcal T} }
\def\I{ \mathbbm{1} }
\def\tr{ \mbox{tr} }

\def\x{\boldsymbol{x}}
\def\y{\boldsymbol{y}}

\newtheorem{theorem}{Theorem}[section]

\newtheorem{lemma}[theorem]{Lemma}
\newtheorem{definition}[theorem]{Definition}

\newtheorem{corollary}[theorem]{Corollary}

\begin{document}
\title{Global and local gauge symmetries beyond Lagrangian formulations}
\author{Cristina C\^{i}rstoiu}
\affiliation{Controlled Quantum Dynamics Theory Group, Imperial College London, Prince Consort Road, London, SW7 2BW, United Kingdom}
\author{David Jennings}
\affiliation{Department of Physics, University of Oxford, Oxford, OX1 3PU, United Kingdom}
\date{\today}

\date{\today}

\begin{abstract}
What is the structure of general quantum processes on composite systems that
respect a global or local symmetry principle? How does the irreversible use of quantum resources behave under such
symmetry principles? Here we employ an information-theoretic framework to address these questions and show that
every symmetric quantum process on a system has a highly rigid decomposition in
terms of the flow of symmetry-breaking degrees of freedom between each subsystem and its environment. The decomposition has a natural causal structure that can be represented diagrammatically and makes explicit gauge degrees of freedom between subsystems. The framework also provides a novel quantum information perspective on lattice gauge theories and a method to gauge general quantum processes beyond Lagrangian formulations. This procedure admits a simple resource-theoretic interpretation, and thus offers a natural context in which features such as information flow and entanglement in gauge theories and quantum thermodynamics could be studied. The framework also provides a flexible toolkit with which to analyse the structure of general quantum processes. As an application, we make use of a `polar decomposition' for quantum processes to discuss the repeatable use of quantum resources and to provide a novel perspective in terms of the coordinates induced on the orbit of a local process under a symmetry action.
\end{abstract}

\maketitle

\section{Introduction}
	
Symmetry principles are typically associated with reversible dynamics, where they are fundamentally linked with conservation laws. However, they also arise in situations in which there is some form of \emph{irreversibility} present \cite{LK,korzekwa2016extraction,Noether,Marvianthesis,MatteoThesis,KamilThesis}. In such regimes, it has been shown that there is a break-down between symmetry principles and conservation laws \cite{Noether}, and novel information-theoretic measures come into play \cite{PhysRevA.77.032114,gour2008resource, PhysRevA.93.042107,PhysRevLett.116.150502,L,1367-2630-14-7-073022}. 

Can we understand broad concepts such as gauge symmetries and irreversibility under a unifying framework? There are increasing motivations to extend these concepts beyond Lagrangian and state formulations into a more general setting \cite{PhysRevLett.119.020501,Freivogel2016,PhysRevLett.115.151601}. This is not only for the sake of greater abstraction and unity, but also to connect with the large array of results that have been developed recently in quantum information theory, which are framed in the more general terms of completely-positive trace-preserving (CPTP) operations \cite{Nielsen,preskill1998lecture,Wilde:2013aa}. The present work seeks to contribute to this goal.

The central question we take as a starting point in this work is: 
\begin{center}
\emph{What are the consequences of global or local gauge symmetry on the structure of many-body quantum processes?}
\end{center}
  We tackle this within the context of quantum information theory, and develop a ``diagrammatic process mode'' formalism for general quantum processes. In particular in Section \ref{Section-bipartite} we analyse how the dynamics of a quantum system with global symmetry constraints arises from local exchange of symmetry-breaking resources across any bipartite split. Previous work \cite{asymmetry1, asymmetry2} mainly focused on resource states that break a symmetry, and the resulting framework has provided a number of significant applications \cite{L,LK,PhysRevA.93.052331,lostaglio2017markovian,hebdige2018classification}. In addition, in \cite{asymmetry2} a harmonic decomposition of quantum processes was introduced and discussed, and which we build on in this work. In particular we deal with localized symmetry-breaking degrees of freedom, and develop an intuitive diagrammatic analysis for general quantum processes, that leads to a range of extensions and applications.

We also note that traditional quantum reference frame analysis usually starts with some target quantum operation $\E$ and 
aims to construct candidate \emph{models} involving an external reference frame and a choice of interactions with the reference frame and system in order to approximate $\E$ as closely as possible \cite{RevModPhys.79.555, popescu2018quantum}. In contrast, the analysis we present here has the distinct advantage that it is ``model independent''. It specifies explicitly the minimal resources needed to realise $\E$, without having to commit to a particular resource state or interaction.
 
We also show that this analysis of quantum processes has a natural gauge degree of freedom. In Section \ref{sec-orbit-gauge-freedom} we show that this freedom has a simple interpretation in terms of a local `process orbit', while in Section \ref{sec-irreversibility} we use this process orbit setting to consider potential incompatibility in the use of symmetry-breaking quantum resources for local information-theoretic tasks. This provides a clear physical explanation of recent results on quantum coherence \cite{aberg,aberg2016fully, korzekwa2016extraction,woods2016autonomous,erker2017autonomous} framed in simple geometric terms. 

Finally in Section \ref{sec-gauge} we apply our diagrammatic process mode formalism to the problem of gauging a global symmetry principle for a general quantum process to a local one. We provide an information-theoretic perspective on the gauging procedure in terms of concepts from the field of quantum reference frames, and so enables the application of ideas from one area into the other. Our gauging procedure for quantum processes neither assumes a Lagrangian formulation, nor places restrictions on the existence of `classical regimes' in the form of macroscopic reference frames. To demonstrate consistency with traditional gauge theories we describe how our procedure coincides with the gauging of unitary dynamics on a lattice model. We also describe how this approach provides a simple interpretation of Gauss' law and gauge dynamics from a resource-theoretic perspective, and discuss future directions to be explored.
\section{Diagrammatic decomposition of quantum processes under a symmetry group}
\label{diagram-decomposition-process}
Symmetries may originate from various physical considerations -- conservation laws, geometry of a specific physical set-up, lack of shared reference frames, fundamental laws in particle physics etc. However, in this analysis we will not focus on a particular model but rather consider a general framework that can describe symmetry principles that do not have associated conservation laws. This goal can be viewed as trying to extend constructions traditionally used in Lagrangian dynamics to general completely positive trace-preserving maps. 

Our focus is on the symmetry properties of general quantum processes $\E$ that take states of a quantum system $A$ with Hilbert space $\H_A$ into states of a quantum system $A'$ with Hilbert space $\H_{A'}$. The symmetry group $G$ (assumed to be a discrete or compact Lie group) acts on both the input and output systems $\h_{A}$ and $\h_{A'}$ through unitary representations $U_{A}$ and $U_{A}'$. Such a unitary representation $U_A: G \rightarrow \B(\H_A)$ maps any group element $g\in G$ to a unitary operator in $\B(\h_{A})$  in such a way as to respect the group composition law. This group action on $\H_A$ lifts naturally to the adjoint action on $\B(\H_A)$, the space of linear operators on $\H_A$, which we denote by $\mathcal{U}_{g}(\cdot):=U_{A}(g)(\cdot)U_{A}(g)\hc$. The space $\mathcal{S}(\h_{A},\h_A')$ of all linear superoperators from $\B(\h_A)$ into $\B(\h_{A}')$ also carries a natural group action: $\E\mapsto \U_{g}'\circ\E\circ\U_{g}\hc$ for any $\E \in \mathcal{S}(\H_A, \H_{A'})$. A \emph{symmetric} process $\E\in \mathcal{S}(\H_A, \H_{A'})$ is then a completely positive trace-preserving element of $\mathcal{S}(\H_A, \H_{A'})$ that is left invariant under this group action. General processes will not be symmetric, and instead contain a symmetry-breaking component that we want to describe quantitatively. 

A detailed analysis of the consumption of symmetry-breaking resources at the level of quantum states was provided in \cite{asymmetry2,asymmetry1} based on \emph{modes of asymmetry}, in which a state $\rho$ is decomposed in terms of irreducible components. Specifically, one can write $\rho = \sum_{\lambda,m,k} \rho^{\lambda,m}_kT^{\lambda,m}_k$, where $\rho^{\lambda,m}_k \in \mathbb{C}$ and the operators $T^{\lambda,m}_k \in \B(\H_A)$ form a basis of \emph{irreducible tensor operators (ITO)} \cite{ITO} that transform under the group as  $\U_g (T^{\lambda,m}_k) = \sum_j v_{kj}(g) T^{\lambda,m}_k$ for any $g \in G$, with $v^\lambda(g)_{kj}$ being the matrix components of $\lambda$-irrep of the group $G$, $\lambda$ labelling the irrep of $G$, $k$ the basis vector of the irrep and $m$ an irrep multiplicity label.

The starting point of our work is the generalization of this approach to the level of quantum processes with a natural extension of ITOs.  We define \emph{process modes} as a set of superoperators $\{ \T^\lambda_k \in \mathcal{S}(\H_A, \H_{A'})\}_k$ with the property that
\begin{equation}
\U_{g}'\circ\T^{\lambda}_{k}\circ \U_{g}\hc=\sum_{j=1}^{\rm{dim}(\lambda)}v^{\lambda}(g)_{kj}\T^{\lambda}_{j},
\end{equation}
where $\lambda$ labels an irrep of $G$ and the indices $k,j$ range from $1$ to $\rm{dim}(\lambda)$, the dimension of the irrep. In general, these superoperators $ \T_k^\lambda$ are not completely positive or tracing-preserving maps. The label $\lambda$ may range over the set we denote by $\rm{Irrep}(A,A')$ consisting of all irreducible representations that arise in the decomposition of $\mathcal{U}'\otimes\mathcal{U}$ (or equivalently $U'\otimes U'^{*}\otimes U^{*}\otimes U^{*}$). Therefore, the process modes provide a symmetry-adapted basis for the set of superoperators $\mathcal{S}(\h_{A},\h_{A'})$.

\begin{figure}[h]
	\begin{center}
		\includegraphics[width=3.5cm]{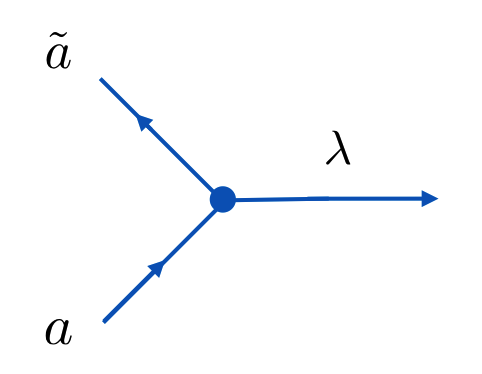}\\
		\caption{\textbf{The structure of process modes}.\\  A diagrammatic representation of a process mode $\{\Phi^{\boldsymbol{\lambda}}_k\}$ transforming between state mode $(a,p)$ in the input system $A$ and state mode $(\tilde{a},q)$ in the output system $A'$. Time runs up the page and we have suppressed multiplicity labels. The horizontal leg is labelled by $(\lambda,m)$ and corresponds to symmetry-breaking degrees of freedom required for the process to be realised. For the example of state preparation process the input system is the trivial system $\mathbb{C}$ and so $a$ can only be the trivial irrep of $G$. This implies that $(\tilde{a},q) = (\lambda,m)$ only, which recovers the modes of asymmetry decomposition. }
		\label{3-leg}
	\end{center}
\end{figure}

For any pair of input and output spaces $(\H_A, \H_{A'})$ we define the set of \emph{canonical process modes} $\{\Phi^{\boldsymbol{\lambda}}_k\}$, with $\lambda$ an irrep in $\rm{Irrep}(A,A')$ and $m$ a multiplicity label packaged together into $\boldsymbol{\lambda} = (\lambda, m)$.
These are built out of coupling incoming state-modes $\{T^{a,p}_k \in \B(\H_A)\}_k$ in the input system with outgoing state modes $\{S^{\tilde{a},q}_j \in \B(\H_{A'})\}_k$ in the output system as described in Supplementary Material Section \ref{S-canonical-ITS} to form a superoperator transforming as a $\lambda$-irrep. These are the basic building-blocks of the formalism, and can be represented as in Fig. \ref{3-leg} by three-legged objects labelled with an ``in-going'' mode $(a,p)$ that evolves into an ``out-going'' mode $(\tilde{a},q)$ by way of an interaction with an external degree of freedom $(\lambda,m)$. This decomposition has a natural causal structure to it that describes the flow of symmetry-breaking resources.

The space $\mathcal{S}(\H_A,\H_{A'})$ decomposes into irrep subspaces spanned by $\{\Phi^{\boldsymbol{\lambda}}_k\}_k$ for each $\boldsymbol{\lambda}=(\lambda,m)$ in $\rm{Irrep}(A,A')$, and thus the canonical process modes can be viewed as the elementary units of any quantum process with respect to a symmetry group $G$.
Each of them has an associated diagram that gives information on the state mode on which it acts non-trivially, how it transforms under the group action and the state mode it can output.
For a fixed choice of basis for the input and output spaces, the diagram encodes the multiplicity label and uniquely defines a process mode. Because of this the label $\boldsymbol{\lambda}$ is basis specific, and hides the multiplicity label so as to make the exposition clear without losing any relevant information.

Given this notation, any $\E\in \mathcal{S}(\H_A,\H_{A'})$ may be uniquely decomposed as
\begin{align}
\E=\sum_{\boldsymbol{\lambda},k}\alpha_{\boldsymbol{\lambda},k}\Phi^{\boldsymbol{\lambda} }_{k}
\end{align}
for some complex coefficients $\alpha_{\boldsymbol{\lambda},k} \in \mathbb{C}$.

Simple examples of process modes are easily constructed in the case of the rotational group on a single spin-1/2 system. For this the irrep label $\lambda$ is an angular momentum label, and the set of quantum processes involve only spin-0, spin-1 and spin-2 contributions. More details on this can be found in Supplementary Material Section \ref{S-single-qubit-process-modes}. 

\subsection{Local coordinates for the orbit of a process}
\label{sec-orbit}
Given a symmetry principle, a core question is how quantum processes local to some region $A$ can arise dynamically through interactions with an ambient environment $B$. If these interactions are constrained by underlying symmetry principles then the ambient environment must function so as to generate a set of local ``coordinates'' $\{x_i\}$ with respect to which a quantum process $\E$ at $A$ is induced. For example, a time coordinate $\{x_i\} = \{t\}$ is necessary when using $B$ as a quantum clock with which to perform timed operations on $A$, or angular data $\{(\theta,\phi)\}$ arises when we want to use a quantum system to break rotational symmetry on $A$. The coordinates $\{x_i\}$ required depend on both $G$ and the quantum process $\E$, and are described by \emph{process orbit} $\M(G,\E)$. More precisely, denote by
\begin{equation}
 \M(G,\E) := \{ \U_g' \circ \E\circ \U_g^\dagger : g \in G\}
\end{equation}
the orbit of $\E$ within the space of superoperators, under the symmetry action. The motivations for introducing $\M(G,\E)$ are
\begin{enumerate}[i]
	\item $\M(G,\E)$ specifies the minimal set of classical coordinates for $\E$ under the symmetry constraint.
	\item Choosing an origin for $\M(G,\E)$ corresponds to a gauge freedom in the description of the physics occurring at $A$, and this perspective is significant when we discuss the gauging of multipartite quantum processes in Section \ref{sec-gauge} 
	\item $\M(G,\E)$ has a natural geometry to it, which is determined by the asymptotic regime of classical reference frames.
\end{enumerate}
See Suplementary Material Section \ref{S-use-of-process-orbits} for more discussion.

\begin{figure}[t]
\includegraphics[width=9cm]{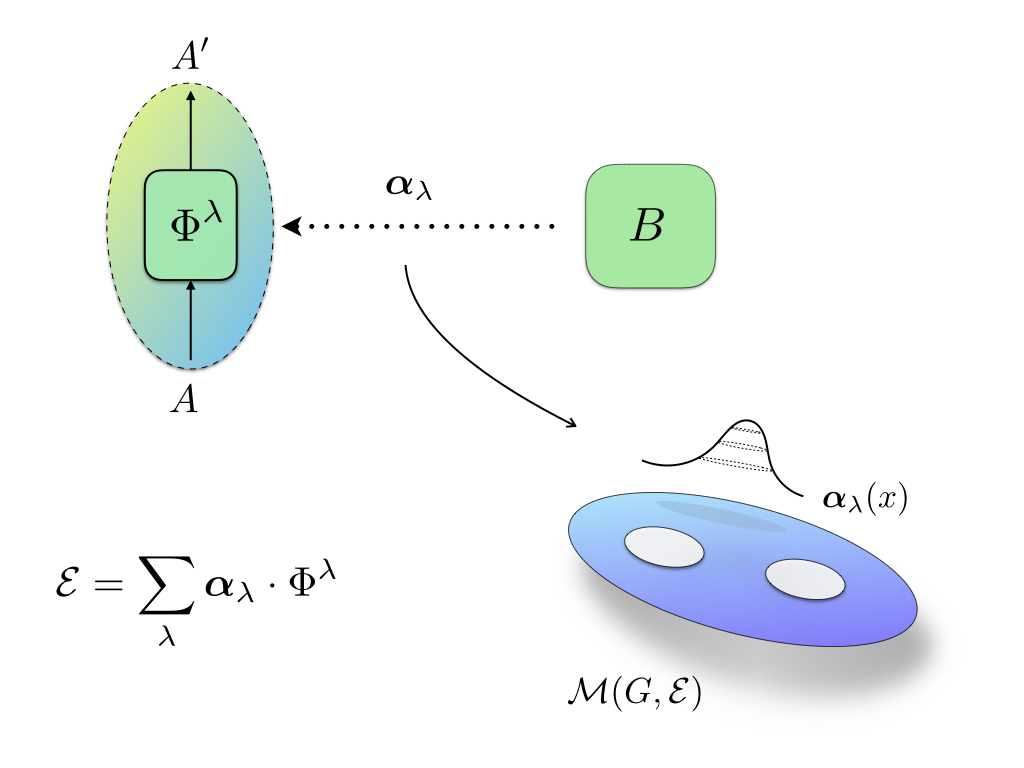}
\caption{\textbf{Polar-decomposition of a general quantum process.} Given a symmetry $G$ the decomposition of quantum processes gives rise to process modes $\{\Phi^{\boldsymbol{\lambda}}_k\}$. The local simulation of a symmetry-breaking process $\E$ on a quantum system $A$ requires specific resources in the environment $B$. These are encoded by data $\{\alpha_{{\boldsymbol{\lambda}},k}(x)\}$, which are un-normalised harmonic wavefunctions on a process orbit $\M(G,\E)$. The local process is given by $\E = \sum_\lambda \boldsymbol{\alpha}_{\boldsymbol{\lambda}} \cdot \boldsymbol{\Phi}^{\boldsymbol{\lambda}} \equiv  \sum_{\lambda,k} \alpha_{\boldsymbol{\lambda} ,k}\cdot \Phi^{\boldsymbol{\lambda}}_k$. For the case of $\E$ being a symmetric process the process orbit $\M(G,\E)$ collapses to being a single point and so has no structure.}
\end{figure}

\tikzstyle{mybox} = [draw=blue, fill=gray!8, very thick,
rectangle, rounded corners, inner sep=-2pt, inner ysep=9pt]
\tikzstyle{fancytitle} =[draw=blue, fill=white!8, text=black, thick, rectangle, rounded corners]
\begin{figure*}[t!]
	\begin{tikzpicture}
	\node [mybox] (box) {%
		\begin{minipage}[t!]{\textwidth}
		\setlength{\tabcolsep}{0.5em}
		{\renewcommand{\arraystretch}{1.85}
			\begin{tabular}{p{100mm}c}
			{{\textbf{Quantum process}}} &\large{$\left(a_{0},{a'_1},{a_1},{a_2}\right)$}\\ \hline  \hline  {\bf{ Dephasing:}}$ \ \ \E(\rho)=p\rho+(1-p)\sum_k \Tr(\Pi_k \rho) \Pi_k  $&{$\left(\frac{2p-1}{\sqrt{3}},0,1-p,0\right)$}\\ \hline
			{\bf{Projective measurements:}}$\ \ \E(\rho)=\sum_k \Pi_k \tr(\Pi_k\rho)$&$\left(-\frac{1}{\sqrt{3}},0,0,1,0\right)$ \\  \hline
			{\bf{Rotation about an axis:}}$ \ \ \E(\rho)=e^{i \frac{\phi}{2}\mathbf{\hat{n}}\cdot\boldsymbol{\sigma}}\rho e^{-i\frac{\phi}{2}\mathbf{\hat{n}}\cdot\boldsymbol{\sigma}} $ & $\left(-\frac{1}{\sqrt{3}}(1+2\cos{\phi}),0,-i\sqrt{2}\sin{2\phi},2\sin^{2}{\phi}\right)$\\ \hline
			{\bf{State preparation}}$\ \ \E(\rho)=\frac{1}{2} (\I + p \mathbf{\hat{n}}\cdot\boldsymbol{\sigma}) $ & $\left(0,p,0,0\right)$\\ \hline
			\bf{ Depolarising process:}$\ \ \E(\rho)=p\rho+(1-p)\frac{1}{2}\iden$ &$\left(\frac{1-4p}{\sqrt{3}},0,0,0\right)$\\ \hline
			\end{tabular}
		}
		\end{minipage}
	};
	\end{tikzpicture}
	\caption{\textbf{Axial processes and resource demands}. Any process $\E$ admits a natural process mode decomposition under a symmetry group $G$, however for axial processes this decomposition takes on a simple and intuitive form as in Theorem 1. The table shows the decomposition of axial processes on a single qubit with an SU(2) symmetry principle into invariant resource demands. Here $\Pi_0 = |\hat{n}\>\<\hat{n}|$ is the projector onto the $\hat{n}$ direction, while $\Pi_1 = \I - \Pi_0$ is the projection onto the $-\hat{n}$ direction. The data for each process are invariants for the group orbit of that process, and together with the choice of relative alignment to external references (in terms of the location on $\M(G,\E)$ i.e $\mathbf{n}$) fully specify the particular process. The coefficients correspond to spin-0 ($a_0$), spin-1 ($a_1$ and $a'_1$) and spin-2 ($a_2$) contributions -- no higher orders are needed for qubit to qubit processes.}
	\label{processEx}
\end{figure*}
\subsection{Process data $\{\alpha_{\boldsymbol{\lambda},k}\}$ as wavefunctions on the space $\M(G,\E)$.}
\label{sec-orbit-gauge-freedom}

There is a very clear link between the process orbit and process modes. While this is best motivated by looking at \emph{axial processes} one can make more general statements for arbitrary groups and processes. The set of axial quantum processes comprises of CPTP maps that break the full rotational symmetry group, but still have a residual symmetry in some direction. Such maps are abundant throughout quantum physics -- for example: dephasing a qubit about an axis, preparation of a pure, polarized spin state, measurements along a particular axis, unitary rotations that leave a fixed axis invariant -- and therefore form a convenient set of quantum processes to illustrate structures. Specifically when we consider the global symmetry action for $G=SU(2)$, if the group elements $h\in G$ that leave $\E\in\T(A,A')$ invariant i.e \, $\U_{h}\circ \E\circ\U_{h}\hc=\E$ form a $U(1)$ subgroup of $SU(2)$ then $\E$ is said to be an axial process.

In this case the process orbit $\M(G,\E)$ is a sphere $S^{2}\cong SU(2)/U(1)$ and there is a distinguished unit vector $\mathbf{\hat{n}}$ on $S^{2}$ associated to $\E$ such that $\E$ remains invariant under rotations around the axis defined by $\mathbf{\hat{n}}$. Then the coefficients $\alpha_{\lambda,k}$ in the process modes expansion of $\E$ take a particular simple structure as un-normalised wave functions on the sphere. Concretely, in this case they are proportional to spherical harmonics:
\begin{align}
\alpha_{\boldsymbol{\lambda},k}=a_{\boldsymbol{\lambda}}Y_{\lambda k}(\theta,\phi)
\label{mainspherical}
\end{align}
where $(\theta,\phi)$ are the angular coordinates of the point $\mathbf{\hat{n}}$ on the sphere. The coefficients $a_{\boldsymbol{\lambda}}\in\mathbb{C}$ are independent of the vector component $k$ and constant for all processes in the orbit of $\E$. 

The core point of this result is that it separates the process resource requirements $\{\alpha_{\boldsymbol{\lambda},k}\}$ into local demands, given by a set of invariant resource demands $(a_{\boldsymbol{\lambda}_1}, a_{\boldsymbol{\lambda}_2}, \dots)$, from the purely relational information on how $B$ is aligned relative to $A$. 
More explicitly, any axial process $\E$ is fully specified by the numbers $\{\alpha_{\boldsymbol{\lambda},k}\}$.
These can be further decomposed into quantities $(a_{\boldsymbol{\lambda}})$ that are independent of the relative alignment of $A$ and its environment, together with a choice of coordinates $\mathbf{n}=(\theta,\phi)$ on $\M(G,\E)$ that specify the relative alignment of $A$ and $B$.

While axial processes are natural and intuitive, the above construction can be extended easily to a general statement for any quantum process $\E$ that has a particular symmetry sub-group $H \subset G$ with process orbit $\M(G,\E)\cong G/H$. We summarise the above results with the following general theorem and refer the reader to the Supplementary Material Section \ref{S-axial-operations} for the rigorous statements and proofs.

\emph{Theorem 1:  Under a symmetry principle for a (compact) group $G$, for any process mode decomposition of a quantum process $\E\in \mathcal{S}(\h_{A},\h_{A'})$ into $\E=\sum_{\lambda,k} \alpha_{\boldsymbol{\lambda},k}\Phi^{\boldsymbol{\lambda}}_{k}$, the complex coefficients $\alpha_{\lambda,k}$ are un-normalised spherical harmonic wavefunctions on the process orbit $\M(G,\E)$:
\begin{equation}
\alpha_{\boldsymbol{\lambda},k}=a_{\boldsymbol{\lambda}} Y_{\lambda,k}(\mathbf{x})
\end{equation}
with $\mathbf{x}\in \M(G,\E)$.
}

This can be viewed as a form of \emph{polar-decomposition} for the process $\E$ into parts independent of laboratory alignments and those parts that specify these alignments. It can be therefore phrased schematically as
\begin{equation}
\E \approx (\mbox{Invariant resources}) \times (\mbox{Choice of  gauge}).\nonumber
\end{equation}
For example, if $\E$ is a symmetric process then the symmetry subgroup is the full group $G$ and the process orbit is a single point, so it lacks structure.
In this case the resource demands for $\E$ do not require any reference frame synchronisation with the environment.

\subsection{Globally symmetric quantum processes}\label{Section-bipartite}

The previous analysis explains the physical significance of the process mode decomposition, and provides a compact perspective on the role of quantum reference systems for the implementation of a quantum process on a system. However it does not tell us how these resources and global processes are constrained under a global symmetry.  So far we have only described how local quantum processes on a subsystem $A$ decompose in the demands they place on $B$, which serves to encode reference data $\M(G,\E)$. As mentioned, the choice of origin on $\M(G,\E)$ is a gauge freedom corresponding to how $A$ and $B$ are jointly described. We now build on this and specify the structure of global quantum processes that respect the symmetry principle. 

To begin with, we consider a bipartite split of the full quantum system into $A$ and $B$. Moreover, given an irrep $\lambda$ for a group $G$, we denote the dual irrep as $\lambda^*$, where the dual representation $R^*$ to a matrix representation $R$ of $G$ is defined via $R^*(g) = R(g^{-1})^T$ for all $g$ in $G$. 
The input space $\h_{in}=\h_{A}\otimes \h_{B}$ and output space $\h_{out}=\h_{A}'\otimes\h_{B}'$ carry the tensor product representations $U_{A}\otimes U_{B}$ and $U_{A}'\otimes U_{B}'$ respectively.

\emph{Theorem 2: \label{theorem2} Every symmetric quantum process $\E_{AB}\in \mathcal{S}(AB,A'B')$ has a decomposition into symmetric superoperators:
 \begin{align}
\E_{AB} &= \sum_{\boldsymbol{\lambda_{AB}}}c_{\boldsymbol{\lambda_{AB}}}\chi^{\boldsymbol{\lambda_{AB}}} \nonumber \\
\chi^{\boldsymbol{\lambda_{AB}}}&:=\sum_{k=1}^{\rm{dim} \lambda} \mathcal{C}^{\boldsymbol{\lambda_{A}}}_k\otimes \mathcal{C}^{\boldsymbol{\lambda_{B}^{*}}}_{k}
\end{align}
where $c_{\lambda,\theta}\in\mathbb{C}$ and $\boldsymbol{\lambda_{A}}=(\lambda, m_{A})$ and $\boldsymbol{\lambda_{B}^{*}}=(\lambda^{*},m_{B})$), for any choice of multiplicity labels $m_A, m_B$, and where $\{\mathcal{C}^{\boldsymbol{\lambda_A}}_{k}\}$ (respectively $\{\mathcal{C}^{\boldsymbol{\lambda_B}}_{k}\}$) is any complete set of process modes for $\mathcal{S}(\H_{A},\H_{A'})$ (respectively $\mathcal{S}(\H_{B},\H_{B'})$). The summation ranges over all irreps  $\lambda \in \rm{Irrep}(A,A')$ for which there is $\lambda^{*} \in \rm{Irrep} (B, B')$ and their associated multiplicities $m_{A}$ and $m_{B}$ are labelled collectively by $\boldsymbol{\lambda}_{AB}=(\lambda, m_A, m_B)$.}

The full proof is provided in the Supplementary Material \ref{S-bipartite-theorem}. The result highlights the rigid structure of symmetric quantum processes, and physically states that the bipartite process is composed of invariant process modes, which involve internal exchange of asymmetry between $A$ and $B$ in a balanced way. 

It also allows a diagrammatic representation of the components of such a quantum process $\E_{AB}$ whenever we consider the \emph{canonical process modes} $\mathcal{C}^{\boldsymbol{\lambda_{A}}}_{k}=\Phi^{\boldsymbol{\lambda_{A}}}_{k}$ and $\mathcal{C}^{\boldsymbol{\lambda_{B}}}_{k}=\Phi^{\boldsymbol{\lambda_{B}}}_{k}$ for the local systems $A$ and $B$. We have seen in Section \ref{diagram-decomposition-process} that each such local process mode say at A, generically $\Phi^{\boldsymbol{\lambda_{A}}}$ corresponds to a diagram $(a,\tilde{a})\stackrel{\lambda}{\longrightarrow}$ with incoming and outgoing modes on which the process mode acts non-trivially and similarly at $B$, $\Phi^{\boldsymbol{\lambda_{B}^{*}}}$ corresponds to $(b,\tilde{b})\stackrel{\lambda^{*}}{\longrightarrow}$. In this context, each symmetric process $\chi^{\boldsymbol{\lambda_{AB}}}$ acts non-trivially on the tensor product of incoming modes at $A$ and $B$ and transforms them into tensor product of outgoing modes. We can bundle this action on mode data in terms of a
\emph{diagram label} $\theta = [(a,\tilde{a})\stackrel{\lambda}{\longrightarrow}(b,\tilde{b})]$. Since to each multiplicity $m_{A}$ and $m_B$ there is an associated local diagram at $A$ and similarly at $B$, then the diagram label $\theta$ packages the multiplicities $(m_A,m_B)$. As such, in terms of the local canonical process modes, to every symmetric superoperator $\chi^{(\lambda,\theta)}:=\chi^{\boldsymbol{\lambda_{AB}}}$ there is the associated $\theta$-diagram that has the representation shown in Fig. 4.
\begin{figure}[h!]
\includegraphics[width=5.5cm]{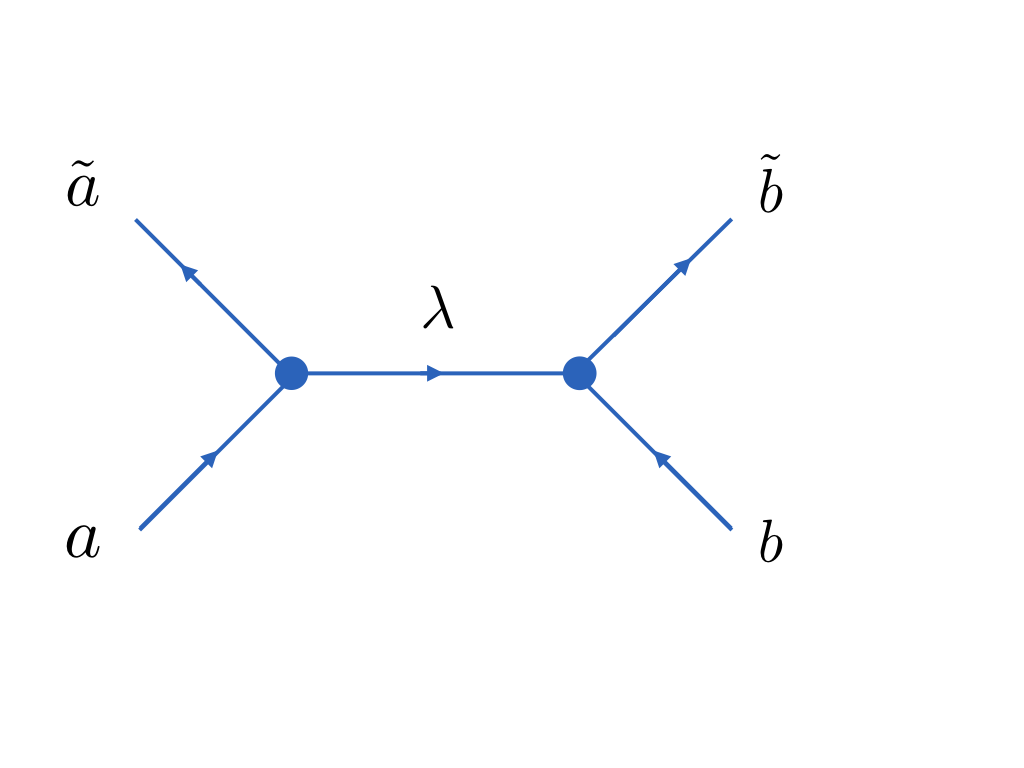}
\caption{\textbf{A  generic diagram for symmetric bipartite processes.} Basis of superoperators for the space of symmetric, bipartite quantum processes $\T(\H_{\rm \tiny in}, \H_{\rm \tiny out})$. The $\lambda$-irrep arrow is associated to a (directed) flow of quantum information. For an abelian group $G$, this is a one-dimensional degree of freedom and so corresponds to classical data (e.g. can be broadcasted, as in the case of quantum coherence). The diagram is represented mathematically in terms of incoming and outgoing asymmetry modes $\theta = [ (a, \tilde{a}) \stackrel{\lambda}{\longrightarrow} (b,\tilde{b})]$.}
\end{figure}\label{bipart-fig}

The diagrammatic decomposition, the polar decomposition in Theorem 1 and the Theorem 2 for bipartite symmetric processes are the main technical results of this section and provide us with the basic tools to analyse concrete model-independent scenarios. We next turn to applications of these results and find that a range of non-trivial insights follow.

\section{Application: Limitations on the efficient use of quantum states under symmetric dynamics}
\label{sec-irreversibility}
For general symmetric quantum processes quantum incompatibility \cite{heinosaari2016invitation} is expected to give rise to irreversibility in the symmetry-breaking degrees of freedom of a quantum system. For example, a quantum system that acts as a clock functions to break time-translation symmetry. However its use in say quantum thermodynamics \cite{janzing2000thermodynamic, LK} may result in a back-action that distorts its subsequent ability to function as a clock \cite{Aberg:2013aa, kwon2018clock, erker2017autonomous, woods2016autonomous}.

One might generally expect globally symmetric quantum processes $\rho_A\otimes \sigma_B \mapsto \E_{AB} (\rho_A \otimes \sigma_B) $ such that $\sigma_B  \mapsto \sigma'_B=\E_B(\sigma_B):= \tr_A  \left [ \E_{AB} (\rho_A \otimes \sigma_B) \right ]$, such that the state $\sigma'_B$ breaks the symmetry in a much weaker form than the original state $\sigma_B$ and is therefore less useful as a result. This constitutes an irreversibility under the symmetry constraint, however it could arise due to the particular interactions used -- might it be possible to use the state more wisely and not suffer such irreversibility?

In the simplest case an isolated symmetric, unitary evolution preserves all symmetry-breaking properties and conserves charges -- but there are many non-trivial fruitful scenarios that illustrate the boundary between reversibility and irreversibility.

In light of this, we can consider the repeatable use of resource states of a reference frame $B$ with a protocol $\P$ whose aim is to implement a simulation of a quantum process $\E$ locally at $A$ via interactions governed by a symmetry principle. The protocol, given a single use of resource state $\sigma_{B}$ on $B$ implements $\E(\rho)=\tr_{B'}(\V_{AB}(\rho_{A}\otimes \sigma_{B}))$ where $\V_{AB}$ is a globally symmetric isometry on $AB$ determined by the target process $\E_{\rm \tiny target}$ that we wish to simulate on $A$.
To address irreversibility features, we consider a repeated application of the protocol using the reduced state in the reference $\sigma_{B}'=\tr_{A}(\V_{AB}(\rho_A\otimes\sigma_{B}))$. 

We say a protocol $\P$ is \emph{arbitrarily repeatable} if for all finite $n$ and every reference frame state $\sigma_{B}$, the local simulation on each system $A_{i}\cong A$ is some fixed process $\E$, and where $\E(\rho_{A_i})=\Tr_{i}(\mathcal{V}_{A_1,A_2,...A_n,B}(\rho_{A_1}\otimes...\otimes\rho_{A_n}\otimes \sigma_{B}))$, where $\mathcal{V}_{A_1,A_2,...A_n,B}$ is a product of symmetric isometries each acting pairwise on $B$ and each system $A_i$ in some ordering.  
\begin{figure}
	\begin{center}
		\includegraphics[width=6cm]{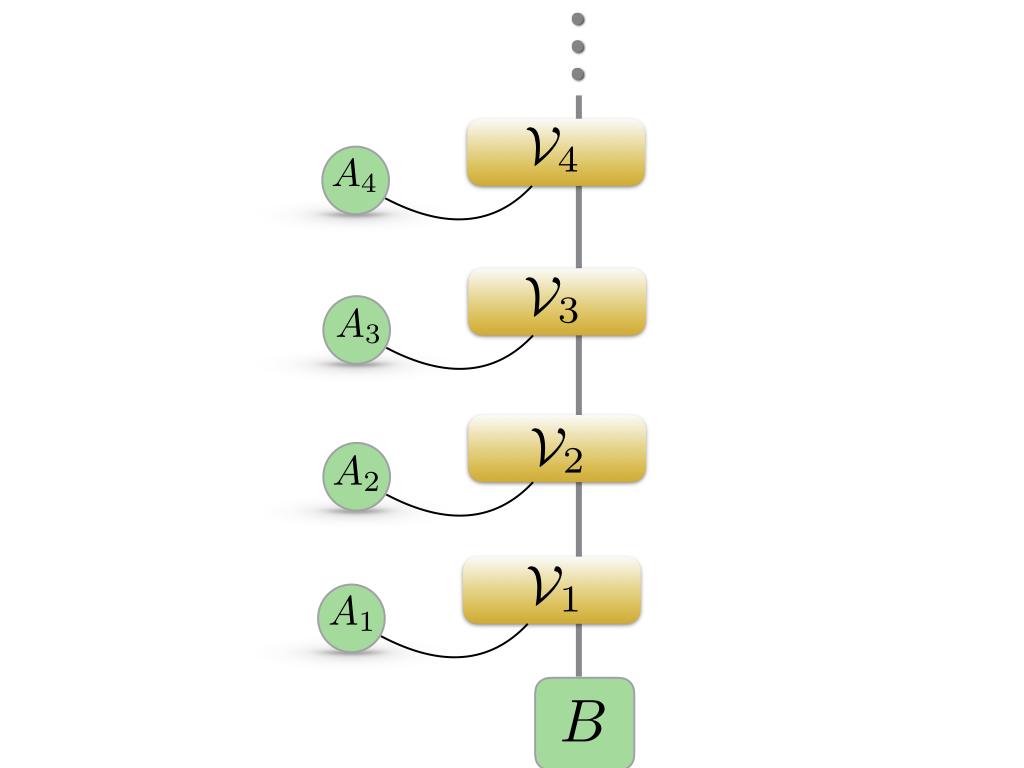}
		\caption{\textbf{Repeatable use of symmetry-breaking system $B$}. A system $B$ is used sequentially to induce otherwise an inaccessible process on systems $A_1$, $A_2$, $A_3 \dots A_n$. If the induced processes are identical for all $n\in\mathbb{N}$ then the system $B$ is used with arbitrary repeatability, even if the state of $B$ changes in time.}
		\label{circuit}
	\end{center}
\end{figure}

This definition captures the ability of the reference frame $B$ to be used in such a way that its performance on each individual quantum system $A_k$ is identical, regardless of the number of systems involved, and so there is necessarily some reference frame property of $B$ that never degrades. One motivation for considering this is given by the prominent work \cite{aberg} in which a feature called catalytic coherence was studied in which quantum coherence can be re-used in such a way that the state of the resource constantly changes, however its ability as a resource for inducing processes on multiple independent systems remains unchanged. In \cite{aberg} the arbitrarily repeatable protocol is subject to a global U(1) symmetry, and is given in terms of a set of unitary interactions $V(U)$ that act on system A and the reference system $B$ consisting of a ladder system with Hilbert space $\h_{\rm{ladder}}$ spanned by eigenstates $\{\ket{n}\}_{n\in \mathbf{Z}}$ of the number operator $N$. Interactions take a particular form
\begin{equation}
V(U)=\sum U_{mn}\ket{\phi_{m}}\bra{\phi_{n}}\otimes \Delta^{n-m}
\end{equation}
where $\{\ket{\phi_{m}}\}_{m=1}^{\rm{dim}(A)}$ forms an orthonormal basis for system $A$ such that it transforms under the U(1) action as $U(\theta)\ket{\phi_{m}}=e^{im\phi}\ket{\phi_{m}}$, the operators $\Delta^{n-m}$ are displacement operators on the ladder system $\Delta^{n}=\sum_{j\in\mathbb{Z}} \ket{j+n}\bra{j}$ and $U_{mn}$ denotes the matrix entries of some arbitrary target unitary $U$ with dimension $\rm{dim}(A)$ that we wish to induce on $A$. Crucially this interaction implements a local simulation on $A$ that depends on $B$ only via the expectation values of $\Delta^{n}$ and takes the form: 
\begin{equation}
\E(\rho) =  \sum_{n,i}  \tr(\Delta^n \sigma ) K_{n,i} \rho K_{n,i}^\dagger
\label{catalytic}
\end{equation}
for operators $\{K_{n,i}\}$ on $A$. In what follows, we shall call any protocol that simulates $\E$ in this form using a ladder system as simply \emph{a catalytic coherence protocol}, without any further qualifications. 

The system $B$ can be reused arbitrarily many times, and its reduced state will change continually under the protocol. Despite this, its ability to function as a coherence reference remains the same. One might think that the protocol in \cite{aberg} functions by doing a projective measurement on the reference system via the covariant measurement $\{|\theta\>\<\theta|\}$ for the $U(1)$ phase of the system and then making use of the this phase angle at $A$ to perform the target map. This would certainly allow for the repeated use of the reference as claimed, however the protocol in \cite{aberg} is not doing this, which can be seen from the fact that the back-action on the reference $B$ under the catalytic coherence protocol can be very slight, and moreover depends explicitly on the type of target unitary $U$. In contrast, the projective measurement on $B$ is independent of $U$ and collapses $B$ to a uniform superposition over the states $\{|n\>\}$. We therefore seek a deeper understanding of what is going on within catalytic coherence protocols and how it relates to the broader notion of repeatability.

Our analysis of an arbitrarily repeatable protocol (irrespective of symmetry constraints) begins with the observation that the effective process from the reference frame into the simulation $\Lambda_{\rho}:\sigma_{B}\rightarrow \E(\rho)=\Tr_{B'}(\V_{AB}(\rho\otimes\sigma_{B}))$ is $n$-extendible \cite{pankowski2011entanglement} for all fixed $n$ and every $\rho$. In general, a process $\Lambda\in \T(\h_{B},\h_{A})$ is $n$-extendible if there is a process $\Lambda_{n}\in \T(\h_{B},\h_{A}^{\otimes n})$ symmetric under permutations of the output spaces and with equal marginals $\Lambda(\cdot) =\Tr_{i}(\Lambda_{n}(\cdot))$ for all $i$, where we trace out over all but the $i$-th system. This observation on extendibility then leads to a simple statement on what types of simulations can be achieved by an arbitrarily repeatable protocol. The proof is provided in Supplementary Material Section \ref{S-protocol}.

\emph{Theorem 3: Given that $\E$ is a process on $A$ simulated by a reference frame state $\sigma_{B}$ via an arbitrarily repeatable protocol $\P$, then there exists a POVM set $\{M_{a}\}$ on system $B$ and CP maps $\Phi_{a}$ on $A$ such that:
\begin{equation}
\E(\rho)=\sum_{a} \Tr(M_{a}\sigma_{B})\Phi_a(\rho).
\label{arbitrary-repeatability}
\end{equation}}

We can now combine this result with our previous analysis to deduce that the outcome measurement probabilities $\Tr(M_{a}\sigma_{B})$ have a natural interpretation in terms of the process orbit. Explicitly, we apply the process mode decomposition to the case where $B$ is an infinite-dimensional ladder system $\h_{\rm{ladder}}$ and in terms of the eigenstates $\{\ket{n}\}_{n\in\mathbb{Z}}$ we consider the following set of orthonormal `states' that encode any $\theta\in U(1)$ \cite{Busch}:
\begin{equation}
\ket{\theta}:=(2\pi)^{-1/2}\sum_{n\in \mathbb{Z}} e^{-in\theta}\ket{n},
\label{classicalreference}
\end{equation}
which should be understood as being meaningful in a distributional sense as a Dirac delta wavefunction on the unit circle $\{e^{i\theta}\}$. We will refer to the states $\{\ket{\theta}\}_{\theta\in U(1)}$ as \emph{asymptotic reference frames}.
We can thus establish the following theorem. 

\emph{Theorem 4: A protocol $\mathcal{P}$ that is used to simulate a local process $\E_{\rm \tiny target}$ on $A$ via a ladder system $B$ satisfies:
\begin{enumerate}[i]
	\item  Global U(1) symmetry.
	\item  Arbitrary repeatability.
	\item  Asymptotic reference frames on $B$ are not disturbed.
	\item  Asymptotic reference frames on $B$ yield perfect simulations of $\E_{\rm \tiny target}$.
\end{enumerate}
if and only if $\P$ is a catalytic coherence protocol.}

This provides a clear physical interpretation of the repeatable use of quantum coherence in simple physical terms and identifies catalytic coherence protocols to be essentially unique under mild assumptions. Note it does \emph{not} imply that the system $B$ is in some perfectly coherent state, or that the state of $B$ stays the same -- the repeatability holds irrespective of the state on $B$. 
\begin{proof}
	From ii, the protocol $\mathcal{P}$ is arbitrarily repeatable so it follows that the induced map $\E$ takes the form $ \E_{\sigma}(\rho)=\sum_{a}\Tr(M_{a}\sigma)\Phi_{a}(\rho)$ for $\{M_{a}\}$ a POVM and $\Phi_{a}$ set of CP maps, and we include the label $\sigma$ in the induced process $\E_{\sigma}$ to account for the fact that different reference states induce different processes on $A$. However, we can decompose each $\Phi_{a}$ into the complete process modes basis as $\Phi_{a}=\sum_{\lambda}c_{\boldsymbol{\lambda},a}\Phi^{\boldsymbol{\lambda}}$ for constants $c_{\boldsymbol{\lambda},a}$ resulting in:
	\begin{equation}
		\E_{\sigma}(\rho)=\sum_{\lambda}\Tr(\sum_{a}\left(c_{\boldsymbol{\lambda},a}M_{a}\right)\sigma)\Phi^{\boldsymbol{\lambda}}.
	\end{equation}
	We simplify the above equation using the notation $X^{\boldsymbol{\lambda}}:=\sum_{a}c_{\boldsymbol{\lambda},a}M_{a}$ to get the compact form for the induced map:
	\begin{equation}
	\E_{\sigma}(\rho)=\sum_{\boldsymbol{\lambda}}\Tr(X^{\boldsymbol{\lambda}}\sigma)\Phi^{\boldsymbol{\lambda}}.
			\label{ind-2}
	\end{equation}
	As a direct consequence of the global U(1) symmetry the action of the symmetry group on $\sigma$ will generate the orbit of $\E_{\sigma}$. More concretely for any $\sigma\in\B(\h_{B})$:
	\begin{equation}
		\E_{\U_{\theta}(\sigma)}=\U_{\theta}\circ\E_{\sigma}\circ\U_{\theta}\hc.
		\label{cov-2}
	\end{equation}
	
	Now we substitute equation (\ref{ind-2}) into (\ref{cov-2}) to get that:
	\begin{equation}
	\sum_{\boldsymbol{\lambda}}\Tr(X^{\boldsymbol{\lambda}}\U_{\theta}(\sigma))\Phi^{\boldsymbol{\lambda}}=\sum_{\boldsymbol{\lambda}} \Tr(X^{\boldsymbol{\lambda}}\sigma) \U_{\theta}\circ\Phi^{\boldsymbol{\lambda}}\circ\U_{\theta}\hc.
	\end{equation}
	The process modes form a complete orthonormal set and transform as $\U_{\theta}\circ\Phi^{\boldsymbol{\lambda}}\circ\U_{\theta}\hc=e^{i\lambda\theta}\Phi^{\boldsymbol{\lambda}}$. Therefore the coefficients associated to each $\Phi^{\lambda}$ in the above must be equal and we have that for all $\lambda$-irreps and all $\theta\in U(1)$
	\begin{equation}
		\Tr(X^{\boldsymbol{\lambda}}\U_{\theta}(\sigma))=\Tr(X^{\boldsymbol{\lambda}}\sigma)e^{i\lambda \theta}.
		\label{transforms}
	\end{equation} 
	Using cyclicity of the trace in the left-hand side of the above we move the group action $\U_{\theta}$ on to the POVM element. Then we use the fact that equation (\ref{transforms}) holds for all $\sigma\in\B(\h_{ladder})$:
	\begin{equation}
	\U_{\theta}\hc (X^{\boldsymbol{\lambda}})= e^{i\lambda\theta}X^{\boldsymbol{\lambda}}.
	\end{equation}
	
	Assumption iii is equivalent to the statement that the POVM effects $\{M_{a}\}$ must all commute with the self adjoint operator $\hat{\Phi}$ associated with the asymptotic reference frames $\{\ket{\theta}\}_{\theta\in U(1)}$, given by $\hat{\Phi}:=\int_{0}^{2\pi}\theta\ket{\theta}\bra{\theta}d\,\theta$. In particular, $[X^{\boldsymbol{\lambda}},\hat{\Phi}]=0$, and therefore $M_{a}$ (and each $X^{\boldsymbol{\lambda}}$) will be diagonal in the asymptotic reference frame basis. Finally, we can write this as $X^{\lambda}=\int \bra{\theta}X^{\lambda}\ket{\theta} \ket{\theta}\bra{\theta}d\,\theta$.
	
	However, the operators $X^{\boldsymbol{\lambda}}$ transform in a particular way under the group action. Moreover, it follows directly from equation (\ref{classicalreference}) that the asymptotic reference frames satisfy $\ket{\theta}=U(\theta)\hc\ket{0}$. These two observations imply that $\bra{\theta}X^{\boldsymbol{\lambda}}\ket{\theta}=\alpha_{\boldsymbol{\lambda}}(\E_{0})e^{-i\lambda\theta}$ for some constant $\alpha_{\boldsymbol{\lambda}}(\E_{0})$ that depends only $\E_{0}$ the representative origin in the process orbit of the induced process $\E$. $\E_{0}$ corresponds to the process induced by the reference frame state $\sigma=\ket{0}\bra{0}$. Altogether, 
	\begin{equation}
	X^{\boldsymbol{\lambda}}=\alpha_{\boldsymbol{\lambda}}(\E_{0})\int e^{-i\lambda\theta}\ket{\theta}\bra{\theta}d\,\theta.
	\end{equation}

	However (see \cite{Busch}) the displacement operators can be written as $\Delta^{\lambda}=e^{i\lambda\hat{\Phi}}$. This implies that $X^{\lambda}=\Delta^{-\lambda}$ and the maps induced by the protocol must take the form of:
	\begin{equation}
		\E_{\sigma}(\rho)=\sum_{\lambda}\alpha_{\lambda}(\E_{0})\Tr(\Delta^{-\lambda}\sigma)\Phi^{\lambda}.
	\end{equation}
	It  can be shown that this admits a Kraus decomposition of the form (\ref{catalytic}), and thus the protocol is necessarily a catalytic coherence protocol as defined in equation (\ref{catalytic}), which completes the proof.
\end{proof}

The abelian structure of $U(1)$ allows us to understand the protocol in another way. While the coherence protocol appears in conflict with cloning intuitions, it should not be viewed as a cloning of reference frame data, but as the \emph{broadcasting of reference frame data} to multiple systems. Broadcasting is a mixed state version of cloning in which one wishes to copy unknown quantum states $\{ \rho_1, \dots, \rho_n\}$ to multiple other parties. In the single copy case a state $\rho_k$ is transformed to a bipartite $\sigma_{AB}$, such that the marginals are $\sigma_A = \rho_k$ and $\sigma_B = \rho_k$. It is known \cite{PhysRevLett.76.2818} that a set of quantum states $\{\rho_k\}$ may be broadcast perfectly if and only if $[\rho_i, \rho_j] = 0$ for all $i,j$. The relevance for us here is that the coherent properties of the environment $B$ are fully described by the expectation values $\<\Delta^k\> := \tr [ \Delta^k \sigma_B]$, and so we need only consider these degrees of freedom. However $[\Delta^k, \Delta^j] =0 $ for all $j,k$ and so a state of the form $\sigma_B = \frac{1}{d} (\I + \sum_k c_k \Delta^k + \mbox{other terms})$ can have the $\Delta^k$ components of the state broadcast in the sense described. 

\section{Application: How to gauge general quantum processes?}
\label{sec-gauge}
Gauge symmetries have played a deep and important role in modern quantum physics \cite{wienberg}. In the traditional sense they are statements about a redundancy in the system's dynamics. In what follows we shall again make use of the process mode formalism to provide an information-theoretic account of gauge symmetries that generalizes existing approaches. Importantly this account makes no requirement of a Lagrangian description, or that the dynamics is reversible and allows us to consider gauge symmetries in the absence of conserved charges.

In Section \ref{Section-bipartite} we analysed the structure of bipartite processes that are symmetric under the action of a group given by $\U_g \otimes \U_g$. As mentioned, implicit in this symmetry action is a relative alignment of the systems, which is encoded in the choice of tensor product $\otimes$ for states on $AB$. We have also shown that this gauge freedom corresponds to an arbitrary choice of origin for the process orbit $\M(G,\E)$. Moreover the structure of bipartite processes is naturally analysed in terms of diagrams $\theta = [(a, \tilde{a}) \stackrel{\lambda}{\longrightarrow} (b,\tilde{b})]$, which have a similar group-theoretic structure to Feynman diagrams for particles interacting via gauge bosons (e.g. electrons scattering via photons) \cite{baez2010algebra}. Given these aspects, it is therefore natural to ask if the freedom in choice of origin in $\M(G,\E)$ coincides in a way with the more traditional notion that arises in gauge theories. To analyse this, we describe how one gauges a general quantum process on a multipartite system from a global symmetry to a local symmetry.

\subsection{Gauging global symmetries for quantum processes -- The core recipe.}
We describe the gauging of quantum processes on multipartite systems. We do not address continuous quantum systems here, however one expects agreement once the system is approximated in a lattice formulation. Consider a multipartite system consisting of subsystems  $A_1, A_2, \dots A_n$ and each of them carry a group action of $G$ given by $U_{i}(g)$ for all $g\in G$ and $i\in\{1,2...,n\}$. Let $\E$ be a \emph{globally symmetric} quantum process acting on it -- this means $\E$ is invariant under the group action $U_{1}(g)\otimes...\otimes U_{n}(g)$ where the \emph{same} element is applied at each site. The aim is to transform the process into $\tilde{\E}$ acting on the system and some extra degree of freedom such that it becomes invariant under the \emph{local} group action $U_{1}(g_1)\otimes ... \otimes U_{n}(g_{n})$ where \emph{different} group elements are applied at each site.

An informal algorithm that describes our gauging of a globally symmetric quantum process to a local one is as follows:
\begin{enumerate}
\item (Background systems) We define an array of quantum reference frames that function to encode relational data.
\item (Background dynamics) We define a quantum process for the collection of reference frames that is symmetric under the global symmetry.
\item (Gauging of symmetry) We discard our access to the relational data between subsystems, via a uniform average over the local symmetry group.
\end{enumerate}
Our goal is to explicitly spell out the information-theoretic components involved in the gauging of general dynamics, and determine the structures required for generalization. Our analysis explicitly shows how the gauge systems encode quantum information about the relative alignment of subsystems, and that the gauge interactions generated under this prescription depend on the information-theoretic properties of the quantum reference frame states, as we shall describe below.

In the following we use the notation $\mathbf{g}:=(g_1,...g_n)$ for a group element in $G^{\times n}$, the local symmetry group for the composite system.
We write $U(\mathbf{g})$ for the action on $\H_{A_1 \dots A_n}$ and $\U_{\mathbf{g}}$ to denote the corresponding action on operators in $\B( \H_{A_1 \dots A_n})$. For compactness, we shall also use the notation $\mathfrak{U}_{\mathbf{g}} [\E] := \U_{\mathbf{g}} \circ \E \circ \U_{\mathbf{g}}^\dagger$ for the group action on processes.

For simplicity we now assume that the multipartite process has a particular structure.  Let $\Gamma = (V,E)$ be the graph obtained by associating each subsystem to a vertex $x \in V$, and $E$ denotes the set of all edges linking each subsystem. To each link $l=[xy]$ joining $x$ and $y$ we pick an arbitrary but fixed choice of orientation. 

We will restrict our analysis to a particular subset of globally symmetric processes $\E: \B( \H_{A_1\dots A_n}) \rightarrow \B(\H_{A'_1 \dots A'_m}) $ that can be written as
\begin{equation}
\E = \sum_{ \{(l_k,\theta_k)\}} c_{ \{(l_k,\theta_k)\}}\chi ^{(l_1,\theta_1)} \otimes \chi^{(l_2,\theta_2)} \otimes \cdots \otimes \chi^{(l_r,\theta_r)}
\end{equation}
with $ c_{ \{(l_k,\theta_k)\}} \in \mathbb{C}$, and where we range over all ordered links $l =[xy]\in E$, between $A_x$ and $A_x$, and where $\chi^{(l,\theta)}$ is a $\theta$-diagram term on $A_x$ and $A_y$. We refer to these as \emph{2-symmetric} processes.

This definition has a simple physical interpretation in that there exists a Kraus decomposition for $\E$ in which all Kraus operators are products on operators on pairs of subsystems. For example, a special case is if we have a spin lattice model with a Hamiltonian $H$ involving pairwise Heisenberg interactions, and expand the unitary $\exp[i tH]$ in powers of $H$, then we will have non-trivial terms acting on multiple systems, but they will take the form of being pairwise symmetric.  This clearly generalises in an obvious way to 3-symmetric (and beyond), where we would consider not just directed links but also oriented triangles (or simplices) on the total graph. Including this would obscure the core ideas and also require a generalisation of the core result on the structure of bipartite symmetric processes, so instead we focus on the case of 2-symmetric quantum processes.

\begin{figure}[t]
\includegraphics[width=9cm]{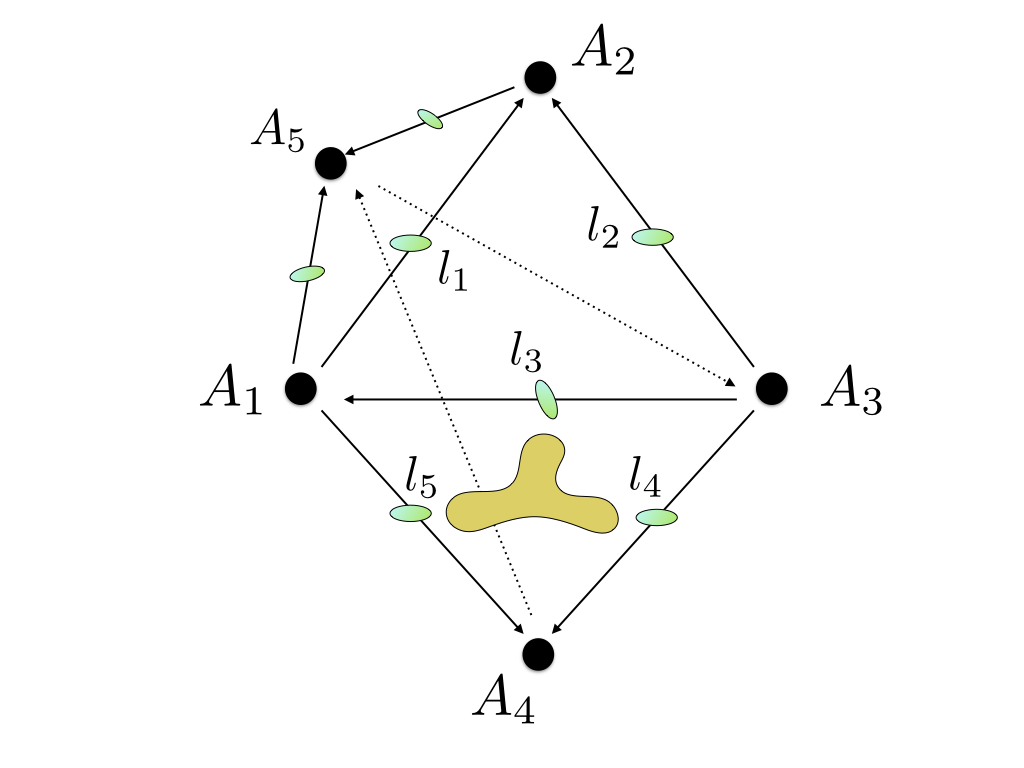}
\caption{\textbf{Gauging a quantum process.} Given a number of subsystems $A_1, \dots A_n$, we associate to each directed link $l_k$ a quantum reference frame (green ellipses) that encodes the relative orientation of the subsystems at its end-points. For the class of $2$-symmetric processes the array of systems $\{l_k\}$ suffices to gauge the dynamics. For $3$-symmetric processes, one must consider plaquette terms (yellow curved region), or equivalently the relative alignments of triples such as $(l_3,l_4,l_5)$. The properties of the quantum reference frames determine the interactions between subsystems.}
\end{figure}

\subsubsection{The inclusion of background reference frame systems} We first introduce an array of quantum reference frames that behave trivially under the global group action. Specifically to every link $l \in E$ we place a quantum system, with Hilbert space $\H_l$ whose principal role is to encode the relative alignment of the end-points of the link. This relative alignment is fully determined by a single group element in $h\in G$. More explicitly, since 
\begin{equation}\label{local-decomp}
(\U_h  \otimes id) \circ (\U_{g_x} \otimes \U_{g_y})=\U_{g_y} \otimes \U_{g_y},
\end{equation}
for  $h = g_y g_x^{-1}\in G$, the fully local action on $A_x \otimes A_y$ differs from a global one by a single group element degree of freedom (on either subsystem).

The reference frame on $l$ functions so as to record this relative alignment through an encoding $h \mapsto \sigma_h \in \B(\H_l)$. In order to be consistent with equation (\ref{local-decomp}) the action of the local symmetry group on a state $\sigma_h$ of $l$ is given by
\begin{equation}
\U_{\mathbf{g}} ( \sigma_h) = \sigma_{g_x h g_y^{-1}}.
\end{equation}
This defines the symmetry action on the reference system $l$. For an initialisation of $l$ in the state $\sigma_e$ we have that $\U_{\mathbf{g}} (\sigma_e) = \sigma_e$ for a global action $\mathbf{g} = (g,g,\dots ,g)$, while for a more general action the reference $l$ encodes the relative alignment via $e \rightarrow g_x g_y^{-1}$, as required.

We do not need to make any assumption as to how well such an encoding can be done, however, modulo technical aspects, there always exists a classical encoding in which one has a set of perfectly distinguishable of states $\{|h\>\}$ for a reference frame system, which carries a well-defined group action given on the basis via $U(\mathbf{g})| h\> := |g_x h g_y^{-1}\>$. 

\subsubsection{Specifying dynamics for the reference frame systems}
Now crucially the quantum reference frames on the links become dynamical objects, and themselves must be subject to a quantum process. However we require the the total quantum process, on subsystems and reference frames, to be invariant under the full local group action. Moreover, we wish that any changes in the relative alignments of systems be encoded in the reference frames. Therefore we must define interactions between subsystems and reference frames that act non-trivially so as to accomplish this.

One could introduce arbitrary couplings between subsystems and reference frames and deduce how well they perform, however the simplest construction is to define couplings that naturally mirror with the process modes that we have introduced. A \emph{process gauge coupling}, $\{\A^{(l,\lambda)}_{jk}\}$ for a quantum reference frame on a link $l$ is a set of superoperators $\A^{(l,\lambda)}_{j,k} : \B(\H_l) \rightarrow \B(\H_l)$ such that
\begin{equation}
\mathfrak{U}_{\mathbf{g}} [\A^{(l,\lambda)}_{j,k}] = \sum_{m,n} v^\lambda(g^{-1}_x)_{mj}v^\lambda(g_y)_{kn}\A^{(l,\lambda)}_{mn},
\end{equation}
under the local symmetry action, and where $x$ and $y$ are the endpoints of the directed link $l$.

These process gauge couplings are essential for gauging the global symmetry to a local one, and if one views process modes $\Phi^\lambda_k$ as comprising a vector  $\boldsymbol{\Phi^\lambda} = (\Phi^\lambda_1, \dots, \Phi^\lambda_d)^T$ of terms that transform irreducibly, then a process gauge coupling $\{\A_{ij}\}$ can be viewed as comprising a matrix of process terms
\begin{equation}
\A = \begin{bmatrix} 
\A_{11} & \A_{12}&\cdots & \A_{1d} \\
\A_{21} & \A_{22}&\cdots & \A_{2d} \\
\vdots & \vdots & \vdots & \vdots \\
\A_{d1} & \A_{d2}&\cdots & \A_{dd} \\
\end{bmatrix}
\end{equation}
for which a local symmetry transformation on subsystem $A_x$ corresponds to left multiplication by the $d\times d$ matrix $v(g_x^{-1})$ of irrep components (for the irrep $\lambda$ with $\rm{dim}(\lambda) = d$), a symmetry transformation on subsystem $A_y$ corresponds to right multiplication by $v(g_y)$, while the diagonal components of $\A$ are each invariant under global actions.

 Since we have restricted to processes that are 2-symmetric, it suffices to describe the construction for a general bipartite superoperator term $\chi^{(l,\theta)} =\sum_j \Phi^{\boldsymbol{\lambda}}_{x,j} \otimes \Phi^{\boldsymbol{\lambda^{*}}}_{y,j} $ with the link $l$ joining $A_x$ and $A_y$. The promotion of the globally symmetric process $\E$ to a locally symmetric one $\tilde{\E}$ is implemented by first making explicit the background process. Since for any fixed $j$ we have that $\mathfrak{U}_{\mathbf{g}} [ \A^{(l,\lambda)}_{j,j}] = \A^{(l,\lambda)}_{j,j}$ for $\mathbf{g}$ being a global symmetry action the superoperator term $\A^{(l,\lambda)}_{j,j}$ is a ``background scalar'' under the global action and so can be included into $\chi^{(l,\theta)}$ without affecting any symmetry properties.
\begin{align}
\chi^{(l,\theta)}&=\sum_j\Phi^{\boldsymbol{\lambda^{*}}}_{x,j} \otimes \Phi^{\boldsymbol{\lambda}}_{y,j} \nonumber\\
\chi^{(l,\theta)}&\longrightarrow \chi'^{(l,\theta )} =\sum_j \Phi^{\boldsymbol{\lambda^{*}}}_{x,j} \otimes \A^{(l,\lambda)}_{j,j}  \otimes \Phi^{\boldsymbol{\lambda}}_{y,j}
\end{align}
The superoperator $\chi'^{(l,\theta )}$ acts on the subsystems in exactly the same way as $\chi^{(l,\theta)}$ under the global group action -- we have simply made explicit the background degrees of freedom. 

While the above describes the process couplings to the system, we must also ensure that the full process is completely positive and trace-preserving. In particular it is insufficient to include only interaction terms on the reference frames -- there must be purely local terms on the reference frames so as to ensure trace-preservation. The details of this are not needed for our present analysis.

\subsubsection{Gauging the process to a local symmetry}
Having made explicit the background reference frame and process gauge couplings we promote the global symmetry to a local one by discarding relative alignments. This is done by averaging over all independent local group actions. The gauge-invariant process components are now obtained via G-twirling the superoperator $\Phi'_{(l,\theta)}$ over the full local group $G^{\times n}$, and are given by 
\begin{equation}
\chi^{(l,\theta)} \stackrel{\mbox{\tiny gauging}}{\longrightarrow} \tilde{\chi}^{(l,\theta )} := \G[ \chi'^{(l,\theta)}] =  \boldsymbol{\Phi^{\lambda^*}_x }^T\cdot \A^{(l,\lambda)} \cdot \boldsymbol{\Phi^\lambda_y}
\end{equation}
with the dot denoting summation over the adjacent indices of the vector-matrix form of the local process modes and gauge couplings.

The fully local process is then
\begin{equation}
\tilde{\E} = \sum_{ \{(l_k,\theta_k)\}}c_{ \{(l_k,\theta_k)\}} \tilde{\chi}^{(l_1,\theta_1)} \otimes \tilde{\chi}^{(l_2,\theta_2)} \otimes \cdots \otimes \tilde{\chi}^{(l_r,\theta_r)},
\end{equation}
and the invariance of each term implies we have gauged the globally symmetric dynamics to a process with local gauge symmetry.  

\subsection{Illustrative example: lattice gauge theory}
We highlight this alternative information-theoretic perspective within the traditional context of unitary dynamics for a lattice gauge theory as described in the Kogut-Susskind Hamiltonian approach \cite{PhysRevD.11.395}. We consider the total Hamiltonian on a two dimensional square lattice $\Gamma = (V,E)$ given by a nearest neighbour hopping $H= \sum_{\x} N(\x) + \sum_{\x, \boldsymbol{\epsilon}\sim \x} K(\x,\boldsymbol{\epsilon})$, where $\boldsymbol{\epsilon}\sim \x$ denotes summation over nearest neighbour points to $\x$. The local particle density observable at each site $\x$ is $N(\x):= \sum_k a^\dagger_k(\x) a_k(\x)$ and the kinetic term describing a hopping from site $\x$ to neighbouring site $\x+\boldsymbol{\epsilon}$ is given by the hermitian operator
\begin{equation}
K(\x,\boldsymbol{\epsilon}):=\sum_k a^\dagger_k(\x+\boldsymbol{\epsilon})a_k(\x) + a^\dagger_k(\x)a_k(\x+\boldsymbol{\epsilon}). 
\end{equation}
The unitary evolution that results from the above Hamiltonian $\E(\rho):=e^{-itH}\rho e^{itH}$ is symmetric under the global group action $\mathfrak{U}_{\mathbf{g}}[\E]=\E$ for all $\mathbf{g}=(g,...,g)$. However, while the local particle density $N(\x)$ is also invariant under the local group action, the hopping term in \emph{not}.
 
 Gauging this unitary process will allow one to make purely local dynamical statements. Since the dynamics are generated by a Hamiltonian, we can do the gauging on the level of the generator for simplicity -- in other words we gauge the superoperator $\mathcal{L}(\rho) = i[ H, \rho]$, which in turn generates the unitary dynamics under exponentiation. As discussed in the previous section the procedure involves adding to every link the lattice $l\in E$ a reference frame $\h_{l}$ which can perfectly encode the group element for the relative alignment of adjacent sites on the lattice. Specifically $\h_{l}$ is spanned by perfectly distinguishable set of pure states $\{\ket{h} :  h\in G\}$ and transforms under the local group action with elements $g_{\x}$ and $g_{\x+\boldsymbol{\epsilon}}$ on the vertices of the edge $l$ according to $\ket{h}\longrightarrow \ket{g_{\x}hg_{{\x+\boldsymbol{\epsilon}}}^{-1}}$.

\begin{figure}[h]
\includegraphics[width=9.5cm]{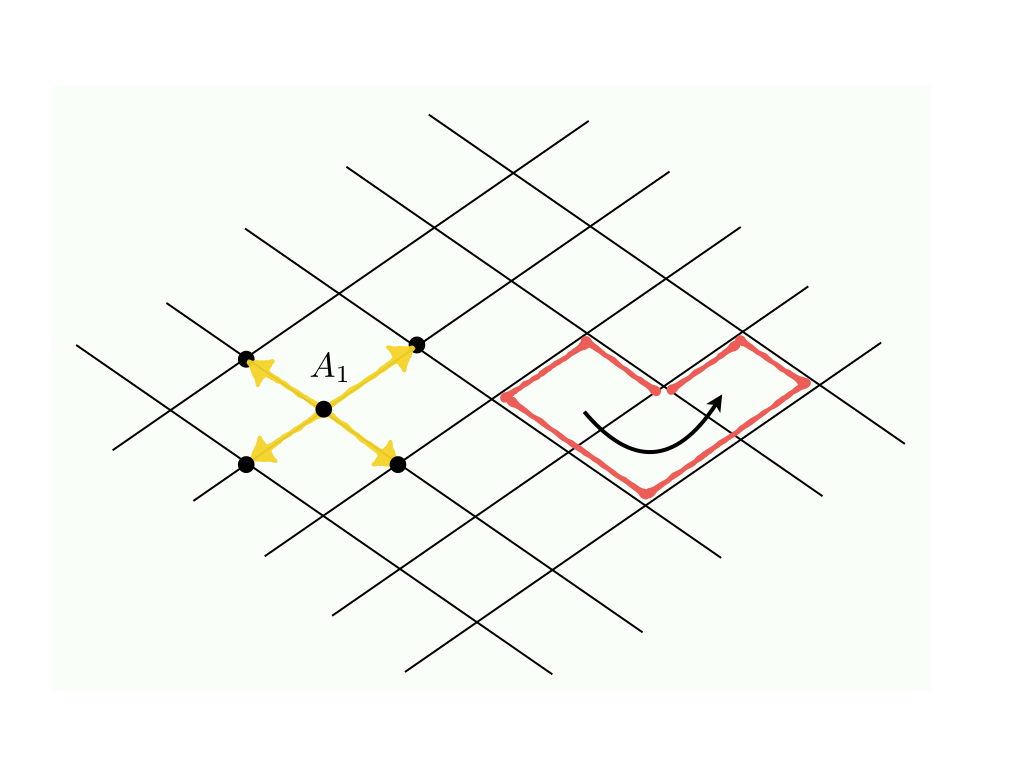}
\caption{\textbf{Lattice gauge theory.} Subsystems are located at the vertices of the lattice, while quantum reference frames on the links. The generators $\{J_a\}$ of the local group action at a subsystem $A_1$ act non-trivially on the vertex as well as the 4 directed links surrounding it (yellow radial arrows).  Gauss' law is satisfied if the state of the system is a symmetric state, $\rho = \G[ \rho]$, under the full local symmetry. Within a fully symmetric scenario, the only observables that can be measured are those that are invariant under the symmetry group. Wilson loops (e.g. the red-loop shown) are defined purely on the reference frames, and are examples of such measurable observables.   }
\end{figure}

The procedure at the Hamiltonian level amounts to gauging the hopping term $K \rightarrow \tilde{K} = \sum_{\x, \boldsymbol{\epsilon} \sim \x} K(\x, \boldsymbol{\epsilon})$, with the inclusion of the link operator such that 
\begin{equation}
\tilde{K}(\x ,\boldsymbol{\epsilon}) = \sum_{jk} a^\dagger_j (\x) \otimes L_{jk}(l) \otimes a_k(\x +\boldsymbol{\epsilon}) + h.c
\end{equation}
where $L_{jk}=\int u_{jk}(h)\ket{h}\bra{h}d\,h $. The process gauge couplings that encode the relative alignment of the subsystems into the reference frame $\h_{l}$ are given by $\mathcal{A}_{jk}(\rho) := [L_{jk},\rho]$.

This describes the dynamics of the systems at the vertices -- on top of this however one must include kinetic terms for the links. A full treatment of this would be beyond the aims of the present work, and so we refer the reader to \cite{gauge1,gauge2,gauge3,gauge4,zohar2015quantum}.

\subsection{A resource theory perspective on gauge dynamics and Gauss' Law}

Having described how the gauging procedure coincides with the traditional unitary dynamics on a lattice approach, we can briefly discuss how it looks from the perspective of quantum resource theories.

In the resource theory of asymmetry, and quantum reference frames, symmetry defines the freely preparable states (or `free states') of the theory. In particular, under the full local symmetry constraint we have the elementary information-theoretic result that any composite state $\rho$ cannot be distinguished from $\G[\rho]$ given by
\begin{equation}
\G [ \rho ] = \int_{G^{\times n}} \!\!\!\! d^n g \,\, \U_{\mathbf{g}} [\rho].
\end{equation}
This fact can be used within the resource-theoretic approach to determine the observables that can be measured within a purely symmetric context \cite{ahmadi2013wigner}. In the language of gauge theories these observables are called ``physical observables'' and the states for which $\G[\rho] = \rho$ are called the ``physical states'' of the theory.

Since $\tilde{\E}$ is symmetric under the local symmetry group we have that $\tilde{\E} (\G(\rho)) = \G(\tilde{\E}(\rho)) =\G(\tilde{\E}(\G(\rho))$. Therefore the dynamics preserve the set of all symmetric states, which is a minimal requirement for consistency. In the language of asymmetry resource theory, these gauge-invariant processes are the free operations of the theory.

Now, the states for which $\rho = \G[\rho]$ are convex mixtures of states with support in the eigenspaces of the generators $\{J_c\}$ of the action $U(\mathbf{g}) = e^{i \sum_c \theta_c J_c}$, where $\theta_c$ are group parameters. Thus, the free states in the theory are convex mixtures of states with sharp values of gauge-invariant observables. This condition is a generalized form of Gauss' Law.

We can outline that this is true for the lattice gauge system. The local group representation is $\mathbf{g} \mapsto U(\mathbf{g})$ and has independent group parameters $g_{\x}$ defined at each site $\x$ on the lattice. Therefore the representation can be written as
\begin{equation}
U(\mathbf{g}) = \exp \left [ i \sum_{\x, c}\theta_c(\x) J_c(\x) \right ],
\end{equation}
where, local to each site $\x$, we have $\theta(\x) \in \mathbb{R}$ and with the operators $\{J_c(\x)\}$ being the local generators of the group action. It is important to note the the operators $J_c(\x)$ act non-trivially both on the vertex $\x$ quantum system and also on the quantum systems residing on the four adjacent links around $\x$. The ``physical Hilbert space of states'' is defined as the span of the gauge invariant vectors $|\psi\>$ that obey $\mathcal{O}_k (\x) |\Psi\> = s_k(\x) |\Psi\>$ for all $\x$ and for a maximal commuting subset of observables $\{\mathcal{O}_k\}$ obtained from the generators \cite{gauge1,gauge2,gauge3,gauge4}. The eigenvalues $\{s_k (\x)\}$ are called ``static charges'', since they are constants of any gauge-invariant evolution. More typically it is demanded that there are no static charges and so the physical space of states is the null space the above observables, and is mapped into itself by all of the local generators. 

We can outline this for the case of $G=U(1)$ where we simply have a scalar number degree of freedom at each site, and a single generator $J(\x)$ at each site. Denoting the lattice vectors as $\boldsymbol{\epsilon}$ in the horizontal direction, and $\boldsymbol{\epsilon}'$ in the vertical direction for a 2-d square lattice. It turns out (see \cite{gauge1,gauge2,gauge3,gauge4} or the recent review \cite{zohar2015quantum}) that this decomposes into a term $q(\x)$ that is purely local to $\x$, and link operators $E(\x +\y)$ acting on the directed link joining $\x$ to $\y$. More explicitly, it takes the form
\begin{equation}
J(\x) = E(\x+\boldsymbol{\epsilon}) - E(\x-\boldsymbol{\epsilon}) + E(\x + \boldsymbol{\epsilon}') - E(\x - \boldsymbol{\epsilon}') - q(\x).
\end{equation}
Thus in the limit $| \boldsymbol{\epsilon}|, | \boldsymbol{\epsilon}'|\rightarrow 0$ the set of physical states are required to obey 
\begin{equation}
\left [ \nabla \cdot \mathbf{E}(\x) - q(\x)  \right ] |\Psi\> = 0,
\end{equation}
which is simply Gauss' Law for the electric field $\mathbf{E}(\x)$ at the point $\x$ in terms of the local charge density $q(\x)$. However from the resource-theoretic perspective, the Gauss law coincides with condition that we can only freely prepare states for which $\G[\rho] = \rho$. 

For convenience we summarize this resource-theoretic perspective: \emph{in the resource theory of asymmetry for a local gauge group $G$, the free states of the theory coincide with the set of all convex mixtures of pure quantum states $|\Psi\>\<\Psi|$ that obey a generalized Gauss' Law. The set of free operations within the resource theory coincide with the set of all locally gauge-invariant processes.}

We also note briefly that the symmetric observables on the reference frames correspond to Wilson loops, and which are also fully invariant under the local group action. The basic loops are around a single plaquette of the lattice and give rise to terms
\begin{equation}
W_p = \tr [ L(l_1) L(l_2) L(l_3) L(l_4)],
\end{equation}
where there is an implicit summing and trace over the $m,n$ indices of $L_{mn}(l)$. It is readily seen that $\U_{\mathbf{g}} [ W_p] = W_p$ for all $\mathbf{g} \in G^{\times n}$ in the local symmetry group. We leave a more detailed analysis to later work where viewing the gauge symmetry from a resource-theoretic perspective could provide a natural context in which to study entanglement in gauge theories.

\subsection{Fixing a gauge -- from local to global symmetry.}
In the context of a gauged process, we can also consider the opposite direction, namely how to go from a local gauge symmetry to a global one. We restrict our discussion to the case in which the reference frame can perfectly encode group elements in a basis $\{|g\>\}$.

  The way in which this gauge fixing can be done is simply by pre- and post-selecting the reference frames onto particular group elements. This breaks the the local symmetry $G^{\times n}$ down to a particular global $G$ symmetry.

Again, it suffices to consider gauging the two site case. The local symmetry is $\U_{(h,g)}$, which we wish to fix to a global action $\U_{(h(g),g)}$ where we assume $h(g) = wgw^{-1}$, for some $w\in G$, and which defines the way in which the action at $A_2$ is related to that at $A_1$. 

The gauge-fixing is achieved as a pre- and post-selecting of the form
\begin{align}
\tilde{\E} \rightarrow\tilde{\E}_{h_1,h_2}:=( id \otimes \Pi_{h_2}) \circ \tilde{\E} \circ (id \otimes \Pi_{h_1}),
\end{align}
where $\Pi_h (\sigma) = |h\>\<h| \sigma |h\>\<h|$, is the projection onto the pure state $|g\>$. The projection $id \otimes \Pi_h$ breaks the $G^{\times 2}$ symmetry action $\U_{(h,g)}$ to the global symmetry action $\U'_g :=\U_{(hgh^{-1}, g)}$, for any $g\in G$. Note that
\begin{equation}
 \U_{(hgh^{-1}, g)}= \U_{(h,e)} \circ \U_{(g,g)} \circ \U^\dagger_{(h,e)},
\end{equation}
and so the passage between global and local symmetry coincides with the degree of freedom discussed in Section \ref{sec-orbit} and \ref{sec-orbit-gauge-freedom} for the relative alignment of two subsystems. More details on gauge-fixing can be found in the Supplementary Material Section \ref{gaugefixing}

\section{Discussion}
The central feature of this work is a tool-kit with which to analyse general quantum processes.
It extends prior asymmetry analysis to a diagrammatic decomposition reflecting both the causal structure of processes and the underlying symmetry principle.
The construction stemmed from a simple and general motivating question on the structure of symmetric processes on many-body systems and it lead to a range of insights and applications.

We have provided an information-theoretic analysis of how a quantum process can be gauged to a local gauge symmetry.
The procedure coincides with traditional approaches: unitary reversible processes lead to lattice gauge theories and (although not discussed here) state preparation processes recover recent constructions in Tensor Networks \cite{Buyens:2014pga,PhysRevX.5.041034,1367-2630-18-4-043008,gauge} that involve gauging quantum states \cite{1367-2630-16-10-103015}.
Since unitary dynamics and state preparation are particular instances of quantum processes, our results can be viewed as generalizations that include both cases within a single unifying setting -- to this aim we use only primitive information-theoretic concepts such as quantum reference frames, and quantum processes on multipartite systems, and with as few assumptions as possible, and without any Lagrangian formulation.

However, one could ask how restrictive it was to use 2-symmetric processes and what our analysis tells us about gauging of symmetries more generally.
Since the bipartite covariance result in Theorem 2 is fully general, the set of 2-symmetric processes can be viewed as the most general form of CPTP maps for which the gauging occurs for pairwise Kraus interactions.
To go beyond this would require slightly more involved machinery for tripartite terms.
However, for sufficiently short timescales, it would be expected that an approximation to 2-body interactions is appropriate, and so falls under the analysis here. In a related direction, one can work solely at the level of generators for the dynamics -- and so perform the gauging on a Lindbladian operator. One direction this might be of use would be in recent work \cite{Poulin-Preskill} on information loss in quantum field systems, where the present techniques would allow gauging of quantum fields without having a global conservation present. We leave to this line of inquiry to future work.

Gauge theories exhibit highly non-local features that give rise to subtleties when one looks at entanglement in this context \cite{PhysRevLett.117.131602,PhysRevD.89.085012,Ghosh2015}. However entanglement theory is best described in terms of the resource theory of Local Operations and Classical Communications (LOCC) \cite{horodecki2009quantum}. This setting does not readily admit a Lagrangian description and so 
one might expect that the formalism that we have presented would be ideally suited for tackling such features in systems with gauge symmetry.

\section{Acknowledgements}
We would like to thank Iman Marvian, Matteo Lostaglio and Kamil Korzekwa for useful discussions on these topics.
CC is supported by EPSRC through the Quantum Controlled Dynamics Centre for Doctoral Training. DJ is supported by the Royal Society.

\bibliography{paper}

\newpage
\onecolumngrid
\appendix

\section{Background -- notations, definitions, basic results.}
We use $\H_A$ to denote the Hilbert space associated to a quantum system $A$, and $\B(\H_A)$ to denote the set of (bounded) linear operators on $\H_A$. A quantum process $\E:\B(\H_A) \rightarrow \B(\H_{A'})$ is a completely-positive trace-preserving superoperator taking states $\rho_A \in \B(\H_A)$ into states $\E(\rho_A) \in \B(\H_{A'})$ for an output system $A'$. We denote the space of superoperators $\Phi : \B(\H_A) \rightarrow \B(\H_{A'})$ by $\T(A,A')$.

By Wigner's theorem, a symmetry on a system $A$ is represented by either a unitary or anti-unitary action on $\H_A$. In this work we consider only unitary actions. Associated to a symmetry group $G$ we have a unitary representation $U:G \rightarrow \B(\H_A)$, with $U(g)$ being unitary on $\H_A$ for all $g\in G$ that respects the usual group composition rules. 

Since we work at the level of density operators and processes, it is convenient to use additional notation. For any $X \in \B(\H_A)$ we denote the adjoint action as $\U_g(X) := U(g) XU(g)^\dagger$. In a similar way we can define a group action on superoperators $\Phi \in \T(A,A')$ via $\Phi \mapsto \U'_g \circ \Phi \circ \U_g^\dagger$, where $\U_g^\dagger := \U_{g^{-1}}$ and $\U'_g$ is the unitary action of $G$ on the output system $A'$.

An operator $X\in \B(\H_A)$ is called \emph{symmetric} if $\U_g (X) = X$ for all $g\in G$, while a superoperator $\Phi \in \T(A,A')$ is called \emph{symmetric} if $\U'_g \circ \Phi \circ \U_g^\dagger = \Phi$. We also use the short-hand $\mathfrak{U}_g [\Phi] := \U'_g \circ \Phi \circ \U_g^\dagger$.

We make use of vectorization of linear operators extensively, and use a modified version of the notation in \cite{watrous2011theory}. Given a linear map $L: \H_A \rightarrow \H_B$ we can define its vectorization, denoted $|vec(L)\>$ which is a vector in $\H_B \otimes \H_A$, by the following method. For $L = |a\>\<b|$, with $\{|a\>\}$ and $\{|b\>\}$ being computational bases for the two spaces, we define
\begin{equation}
|vec(L)\> := |a\>\otimes |b\>.
\end{equation}
The vectorization of a more general linear map $L = \sum_{a,b} L_{ab} |a\>\<b|$ with $L_{a,b} \in \mathbb{C}$ is then fully specified by demanding linearity hold: $|vec(L_1 + L_2)\> = |vec(L_1)\> + |vec(L_2)\>$ for all linear maps $L_1, L_2$ from $\H_A$ to $\H_B$.

It is then easy to verify the following two central properties of vectorization:
\begin{align}
|vec(AXB)\> &= A\otimes B^T |vec(X)\> \\
\<vec(L)|vec(M)\>& = \tr (L^\dagger M),
\end{align}
for all linear maps between the appropriate spaces. The first relation is powerful in the context of entangled bipartite quantum systems, while the second simply says that the mapping $vec$ is an isometry between the Hilbert space $\H_B \otimes \H_A$ and the space of linear maps from $\H_A$ to $\H_B$ with the Hilbert Schmidt inner product $\<L,M\> := \tr (L^\dagger M)$.

The application of these relations make the following easy to establish
\begin{lemma} Given two quantum systems $A$ and $B$ that are isomorphic we have that 
	\begin{align}
	M \otimes \I | vec(\I)\> &= \I \otimes M^T |vec(\I)\> \\
	U\otimes U^* |vec(\I)\> &= |vec(\I) \> \\
	\tr (L^\dagger M) &= \tr ((L^\dagger \otimes M) \mathbb{F})
	\end{align}
	for all $L, M \in \B(\H_A)$ and for all unitaries $U \in \B(\H_A)$, and where $\mathbb{F} := |vec(\I)\>\<vec(\I)|^{T_B} = \sum |ab\>\<ba|$ is the swap operator on $AB$.
\end{lemma}
These relations generalise to the case where $A$ and $B$ are not isomorphic, and where we allow $M$ to map into a different space, by observing that the smaller system, $B$ say, has $\H_B$ isomorphic to a strict subspace of $\H_A$.
\subsection{Representations of superoperators}
Given a superoperator $\Phi \in T(A,A')$ we can represent it in a number of different ways.
The Choi representation $J(\Phi) \in \B(\H_{A'} \otimes \H_A)$ is provided by
\begin{equation}
J(\Phi) := \Phi \otimes id_A ( |vec(\I)\>\<vec(\I)|).
\end{equation}
with inverse relation given by 
\begin{equation}
\Phi(X) = \tr_{A'} (\I_A \otimes X^T J(\Phi)),
\end{equation}
for any $ X \in \B(\H_A)$. The Kraus decomposition of $\Phi$ is given by $\Phi(X) = \sum_k A_k X B_k^\dagger$, where $\{A_k\}_{k=1}^N$ and $\{B_k\}_{k=1}^N$ are the set of Kraus operators. This automatically implies that the corresponding Choi operator is given by
\begin{equation}
J(\Phi) = \sum_k |vec(A_k)\>\<vec(B_k)|.
\end{equation}

The vectorization map gives another representation $K(\Phi) \in \B(\H_{A'} \otimes \H_A)$ via the expression $K(\Phi): |vec(X)\> \mapsto |vec(\Phi(X))\>$ for all $X$. It is easy to verify that
\begin{equation}
K(\Phi) = \sum_k A_k \otimes B_k^*.
\end{equation}
We also have that $\Phi$ is a quantum process if and only if $A_k = B_k$ for all $k$ and $\sum_k A^\dagger_kA_k = \I$, and if and only if $J(\Phi)$ is a positive semi-definite operator with $\tr_A(J(\Phi)) = \I_{A'}$.

The Steinspring dilation $(V,\H_B, |\eta\>_B)$ provides a final representation for a quantum process $\Phi \in \T(A,A')$ given by
\begin{equation}
\Phi(\rho) = \tr_C V( \rho_A \otimes |\eta\>_B\<\eta|)V^\dagger, 
\end{equation}
where $V: \H_A \otimes \H_B \rightarrow \H_{A'} \otimes \H_C$ is an isometry ($V^\dagger V = \I$), and $\sigma_B$ is a fixed quantum state on an auxiliary system $B$, which can be taken to be pure. 

\section{Decomposition of quantum processes}
\subsection{Representations and tensor product representations}
Given a fixed group $G$ one can usually classify and construct every irreducible representation for that particular group. These are exactly those representations which do not have a proper subrepresentation and therefore they contain no subspace invariant under the action of all group elements. We will be dealing with compact Lie groups $G$ and for these types of groups all their irreducible representations are finite dimensional. We denote by $\hat{G}$ the set of all irreducible representations of  $G$. Each irreducible representation is uniquely determined in a canonical way by a distinguished vector which we generically denote by $\lambda\in \hat{G}$ and is called the \emph{heighest weight vector}. A $\lambda$-irrep acts on an $\rm{dim}(\lambda)$ vector space $V^{\lambda}$ with an irreducible representation $v^{\lambda}:G\longrightarrow GL(V^{\lambda})$ that has matrix coefficients $v^{(\lambda)}_{mm'}(g)$ determined by some fixed basis choice for $V^{\lambda}$. In particular they satisfy Schur's orthogonality relations (which are valid for any compact group) for any $\lambda,\mu\in \hat{G}$:
\begin{equation}
\int\limits_{G} v^{(\lambda)}_{mm'}(g) (v^{(\mu)}_{nn'}(g))^{*}d\,g=\frac{1}{\rm{dim}(\lambda)}\delta_{\lambda,\mu}\delta_{mn}\delta_{nn'}
\end{equation}
For any unitary representations $U_{A}:G\longrightarrow \B(\h_A)$ the Hilbert space $\h_{A}$ has a canonical decomposition into subspaces on which the group acts irreducibly. Formally we can write
\begin{equation}
\h_{A}=\bigoplus_{\lambda\in\hat{G};\alpha} V^{\lambda,\alpha}
\end{equation}
where $\alpha$ is a multiplicity label counting the number of times an irreducible representation appears in the decomposition of $\h_{A}$. The symmetry of the system $A$ which manifests itself through the unitary representation $U$ is the only property that dictates which irreps and corresponding multiplicities appear in the decomposition. 

Since we will be interested in bipartite systems $\h_{A}\otimes\h_{B}$ we want to know how one can decompose this space into irreducible components. Suppose that $U_{B}:G\longrightarrow \B(\h_B)$ is a unitary representation of $\h_{B}$ then there is a \emph{tensor product representation} acting on the composite system $U_{A}\otimes U_{B}:G\longrightarrow \B(\h_{A}\otimes \h_B)$ given by $U_{A}\otimes U_{B}(g)=U_{A}(g)\otimes U_{B}(g)$. For example in the case of SU(2) the irreps are labelled by positive half-integers $j\in\{0,1/2,1,3/2,...\}$ and have dimension $2j+1$. The tensor product representation of two irrep $j_1\otimes j_2$ decomposes into irreducible components according to the Clebsch-Gordan series $j_1\otimes j_2=|j_1-j_2|\oplus...\oplus j_1+j_2$. These correspond physically to the possible total angular momentum values that arise when coupling a particle with spin $j_1$ with another with spin $j_2$. Notice how there is only one configuration for each value of the total angular momentum meaning that each irrep in the decomposition appears with multiplicity one. While this is not necessarily the case for general compact groups $G$ similar techniques can be applied there to obtain the canonical decomposition of tensor product representations. We summarise below how these apply generally and refer to \cite{Reps} for a detailed analysis.
\subsubsection{Detour into generalised Clebsch-Gordan coefficients}
{\it{Let $U^{\mu}$ and $U^{\nu}$ be two irreducible representations of $G$ and assume these are realised on the vector spaces $V^{\mu}$ and $V^{\nu}$ respectively where $\mu,\nu\in\hat{G}$. Under the tensor product representation the space $V^{\mu}\otimes V^{\nu}$ decomposes into irreducible components:
		\begin{equation} V^{\mu}\otimes V^{\nu}\cong \bigoplus_{\lambda\in\hat{G}} m_{\lambda} V^{\lambda}
		\label{decomp}
		\end{equation}
		where $m_{\lambda}$ is the multiplicity of the $\lambda$-irrep. This implies that the product of representations $U^{\mu}\otimes U^{\nu}$ is unitarily equivalent to a block decomposition where each block is an irreducible representation of the group. One can write that for all $g\in G$ 
		\begin{equation}
		C(U^{\mu}(g)\otimes U^{\nu}(g))C^{\dagger}=\bigoplus m_{\lambda} U^{\lambda}(g)
		\end{equation}
		for some unitary matrix $C$ which represents nothing more than a change of basis in $V^{\mu}\otimes V^{\nu}$ from the tensor product basis to a basis that achieves the decomposition. The entries of this matrix are what we call the Clebsch Gordan coefficients (CGC) and provide a generalisation to arbitrary compact groups $G$ of the coefficients that appear when coupling angular momentum states.
		
		When $\{\ket{\mu,k}\}_{k=1}^{dim(\mu)}$ and $\{\ket{\nu,k}\}_{k=1}^{dim(\nu)}$ are basis for $V^{\mu}$ and $V^{\nu}$ respectively and $\{\ket{e^{\lambda,\alpha}_{k}}\}_{k=1}^{dim(\lambda)}$ a basis for the $\lambda$-irreducible component labelled by multiplicity $\alpha$ in the above decomposition into irreducible components then these are related through the Clebsch-Gordan coefficients:
		\begin{equation}
		\ket{e^{\lambda,\alpha}_{k}}=\sum\limits_{m,n} \clebsch{\mu}{m}{\nu}{n}{\lambda,\alpha}{k} \ket{\mu,m}\, \ket{e^{\nu}_{n}}
		\end{equation}
		where the coefficients $\clebsch{\mu}{m}{\nu}{n}{\lambda,\alpha}{k} $ represent entries for the unitary matrix $C$. The CGCs depend on the choice of orthonormal basis in the spaces $V^{\mu}$, $V^{\nu}$ and $V^{\lambda,\alpha}$.
		Beyond orthonormality relations inherited from the unitarity of $C$, the generalised Clebsch-Gordan coefficients posses many different types of permutation symmetries and they are non-zero when particular types of relations hold. Within quantum mechanics these relations are exactly the ones that give the selection rules.
		
		The problem of determining the multiplicity $m_{\lambda}$ of each irrep in \ref{decomp} for the general linear group of fixed dimension n is $sharpP$-complete and it can be approximated with a randomized polynomial time algorithm \cite{LRcomplexity}. The problem of determining CGCs is in the NP-complete class \cite{CGCcomplexity}.
		
}}

\subsection{Irreducible tensor operators} 
The structure of $\h_{A}$ provided by the symmetry carries over to higher-level Hilbert spaces such as $\B(\h_{A})$ and $\T(\h_{A},\h_{A'})$ in such a way that it respects their algebraic structure. The mathematical construction that will allow us to upgrade the decomposition of the Hilbert space $\h_{A}$ into irreducible components to the decomposition of $\B(\h_{A})$ are called \emph{irreducible tensor operators}. 
\begin{definition}
	Let $G$ be a compact group and $U$ a unitary representation of $G$ on the Hilbert space $\H_{A}$. Then for every irreducible representation $\lambda\in\hat{G}$ define the irreducible tensor operators (ITO) to be the set of operators $\{T_{k}^{(\lambda)}\}_{k=1}^{\rm{dim}(\lambda)}$ in $\B(\H_{A})$ such that for all $g\in G$:
	\begin{equation}
	\mathcal{U}_g(T_{k}^{\lambda})=\sum v_{kj}^{(\lambda)}(g)T_j^{\lambda}
	\label{ITOdef}
	\end{equation}
	where $v_{kj}^{(\lambda)}$ are matrix coefficients of the $\lambda$-irrep and ranges over all irreps in the decomposition of the representation $U\otimes U^{*}$.
\end{definition}
\noindent
The action of the group $G$ on the space of operators $\B(\h_{A})$ is given by the adjoint action $\mathcal{U}$. Therefore there is a canonical decomposition for $\B(\h_A)$ into irreducible components such that $\mathcal{U}$ acts like an irrep when restricted to each subspace. There is a natural isomorphism between $\B(\H_A)$ and $\h_{A}\otimes \h_{A}^{*}$ but since we can identify any Hilbert space with its dual we can identify the space of operators with two copies of $\H_A$ carrying the representation given by $U\otimes U^{*}$. This means that all irreps that appear when decomposing $\B(\h_A)$ into irreducible subspaces under $\mathcal{U}$ are exactly those that appear when decomposing $\h_{A}\otimes\h_{A}$ into irreducible subspaces under $U\otimes U^{*}$. The following lemma makes this point precise and shows that the set of all ITOs forms an orthonormal basis for $\B(\h_{A})$.\\ \\

\begin{lemma}
	Let $G$ be a compact group and $U$ a unitary representation of $G$ on the Hilbert space $\h$. Given a full set of irreducible tensor operators $\{T^{\lambda}_{k} : \lambda, k\}$ for $\B(\h_{A})$ then the set $\{\vc{T^{\lambda}_{k}}: k=1,...\rm{dim}(\lambda)\}$ forms an orthonormal basis for the $\lambda$-irrep in the decomposition of $\h_{A}\otimes\h_{A}$ under the action $U\otimes U^{*}$. Moreover the ITOs satisfy the orthonormality relation $\Tr((T^{\lambda}_{k})\hc T^{\mu}_{j})=\delta_{kj}\delta_{\lambda\mu}$ for all $\lambda,\mu$-irrep and all $k,j$.
	\label{ito}
\end{lemma}
\emph{Proof}: The result follows easily from orthonormality of matrix coefficients and properties of vectorisation.\qed\\

The basis that achieves the decomposition of $\B(\h)$ into irreducible components is given by the complete set of orthonormal ITOs $\{T^{\lambda}\}_{k=1}^{\rm{dim}(\lambda)}$ for $\lambda$ ranging over all irreps (including multiplicities) that appear in the representation $U\otimes U^{*}$. For every $\lambda$ the set of ITOs under the adjoint action transform irreducibly. Particularly $\U_{g}$ acts on the $\O_{\lambda,\alpha}:=\rm{span}\{ T^{\lambda,\alpha}_{k}:  1\leq k\leq \rm{dim}(\lambda)\}$ in the same way as does the irreducible representation of highest weight $\lambda$. This corresponds to the $\lambda$-irreducible component of multiplicity $\alpha$ in the decomposition of $\B(\h_{A})$. Then the space of operators splits into:
\begin{equation}
\B(\h_{A})\cong \bigoplus_{\lambda,\alpha} \mathcal{O}_{\lambda,\alpha} \    \ .
\end{equation}
Since there is clearly an underlying choice of basis for the irreducible tensor operators there is a sense in which the above decomposition is not entirely unique. However at the high level of the structure of the decomposition there is no freedom to mix operators belonging to different irreducible components.   Denote the $\lambda$-mode by $\mathcal{A}^{\lambda}=\bigoplus_{\alpha} \mathcal{O}_{\lambda,\alpha}$ the full $\lambda$-irreducible component where we have summed over all copies of the $\lambda$ irrep that appear in the decomposition of $\B(\h_{A})$. Therefore the space decomposes in a \emph{unique} way into subspaces: 
\begin{equation}
\B(\h_{A})\cong \bigoplus_{\lambda}\mathcal{A}^{\lambda}
\end{equation}Then given any density matrix $\rho\in \B(\h_{A})$ we can effectively decompose it into \emph{modes of asymmetry} according to: 
\begin{equation}
\rho=\sum\limits_{\lambda} \rho^{\lambda}
\end{equation}
where each of the $\rho^{\lambda}\in\B(\h_{A})$ represents the orthogonal projection of $\rho$ onto the $\lambda$-mode, that is onto the subspace $\mathcal{A}^{\lambda}$ of $\B(\h_{A})$. While indeed some of the projectors above can be zero,  the decomposition of $\rho$ into asymmetry modes will be unique because the coarse grained structure of $\B(\h_{A})$ given by the unitary representation $\U_{g}$ is rigid and always fixed by the symmetries of the underlying Hilbert space. 
\subsubsection{Uniqueness of the ITOs and asymmetry modes}
{\it{ In order to construct a fixed set of ITOs for the space of operators $\B(\h_{A})$ there are two underlying choice of basis: i) the basis for the Hilbert space $\h_{A}$ and ii) a basis for each irreducible representation resulting in a fixed set of matrix coefficients $v^{\lambda}_{kj}(g)$. More specifically the mode decomposition (the coarse grained structure) $\B(\h_{A})\cong \bigoplus_{\lambda} \mathcal{A}^{\lambda}$ is always unique and depends upon the symmetry of the Hilbert space so only on the unitary representation $U$. Once we have fixed a basis for the underlying Hilbert space then the finer-grained decomposition $\B(\h_{A})\cong \bigoplus_{\lambda,\alpha}\mathcal{O}_{\lambda,\alpha}$ becomes unique. Finally whenever we have fixed a basis for the irrep $v^{\lambda}_{kj}$ this implies that we fix the ITOs all together and particularly we fix the vector-component label associated to that particular irrep.}}

\subsection{Irreducible tensor superoperators}
\label{ITSALG}
A similar type of structure we find when dealing with the space of superoperators $\T(A,A')$ and we can further upgrade the irrep decomposition at the level of superoperators by defining the analogue of ITOs:
\begin{definition}
	Let G be a compact group and $U$, $U'$ unitary representations of $G$ on the Hilbert spaces $\h_A$ respectively $\h_{A'}$. For every irreducible representation $\lambda\in\hat{G}$ define the irreducible tensor superoperators (ITS) to be the set of $\{\Phi^{\lambda}_{k}\}_{k=1}^{\rm{dim}(\lambda)}$ in $\T(A,A')$ that transforms under the group action as:
	\begin{equation}
	\uu_{g}(\Phi^{\lambda}_{k}):=\U_{g}'\circ\Phi^{\lambda}_{k}\circ\U_{g}\hc=\sum\limits_{j}v^{(\lambda)}_{kj}\Phi^{\lambda}_{j}
	\end{equation}
	where $v^{\lambda}_{kj}$ are matrix coefficients of the $\lambda$-irrep and ranges over all irreps in the decomposition of $U\otimes U^{*}\otimes U'\otimes U'^{*}$.
\end{definition}
In the main text we use the term process modes for these ITS.

The space of superoperators $\T(A,A')$ has a similar structure in the sense that can be decomposed into irreducible components according to the underlying symmetries. In particular the set of all ITS forms a basis that achieves this decomposition. Therefore when we take into account the possibility of multiplicities to write:
\begin{equation}
\T(A,A')=\bigoplus_{\lambda,\alpha}\rm{span}\{\Phi^{\lambda,\alpha}_{k}: k=1,...,\rm{dim}(\lambda)\}
\end{equation}
where we sum over all $\lambda$-irreps and corresponding multiplicities $\alpha$. The following lemma gives a rigorous proof of this:
\begin{lemma}
	Given that $\{\Phi^{\lambda}_{k}\}_{k=1}^{\rm{dim}(\lambda)}$ forms an ITS for $\T(A,A')$ then for each $\lambda$-irrep in the decomposition of $U_A\otimes U_{A}^{*}\otimes U_A'\otimes (U_{A}')^{*}$ the corresponding Choi operators $J^{\lambda}_{k}:=J(\Phi^{\lambda}_{k})$ form a $\lambda$-irrep ITO under the action of $\U_{A}'\otimes \U_{A}^{*}$.
	\label{ITSChoi}
\end{lemma}

The above lemma allows us to construct ITSs form ITOs which in turn can be obtained by vectorising basis vectors for irreducible components appearing under the decomposition of the required representations. In light of Lemma \ref{ITSChoi} we can define an inner product on $\T(A,A')$ as 
$\<\E_1,\E_2\>:=\Tr(J[\E_1]\hc J[\E_2])\>$. This ensures that the complete set of ITS for $\T(A,A')$ forms an orthonormal basis.

Similarly to our previous discussion for ITOs the ITS also rely on particular basis choices: i) for the input and output Hilbert spaces $\h_{A}$ and $\h_{A}'$ ii) for the matrix coefficients/ irreducible representations and therefore are not uniquely determined by how they transform under the group action. However the coarse-grained decomposition of $\T(A,A')$ into irreducible components is unique in the sense that for any $\E\in \T(A,A')$ the orthogonal projection onto the $\lambda$-irrep isotypical component (i.e including multiplicities) given by $\bigoplus_{\alpha}\rm{span}\{\Phi^{\lambda,\alpha}_{k}:k=1,...\rm{dim}(\lambda)\}$ does not depend on the underlying choice of basis:
\begin{equation}
\E^{\lambda}=\rm{dim}(\lambda)\int \Tr(v^{\lambda}(g))\uu_{g}[\E] d\,g
\end{equation}
where $\Tr(v^{\lambda}(g))$ is just the character of the $\lambda$-irrep and it is independent on the choice of basis that give the matrix coefficients. Furthermore fixing a basis only for irreps and consequently their matrix coefficients then we obtain the asymmetry modes $\E^{\lambda}_{k}$ of the superoperator $\E$ and these are the projection onto the subspaces $\bigoplus_{\alpha}\rm{span}\{\Phi^{\lambda,\alpha}_{k}\}$ for any fixed $k$ and $\lambda$:
\begin{equation}
\E^{\lambda}_{k}=\rm{dim}(\lambda)\int v^{\lambda}_{kk}(g)\uu_{g}[\E]d\,g.
\end{equation}
Therefore any $\E\in\mathcal{S}(A,A')$ can be written as:
\begin{equation}
\E=\sum\limits_{\lambda,k}\E^{\lambda}_{k}=\sum\limits_{\lambda,\alpha,k} \Phi^{\lambda,\alpha}_{k}\, \<\E,\Phi^{\lambda,\alpha}_{k}\>
\end{equation}
where the projection on any $\lambda$-isotypical component $\E^{\lambda}=\sum_{\alpha,k} \Phi^{\lambda,\alpha}_{k}\, \<\E,\Phi^{\lambda,\alpha}_{k}\>$ is independent on the choice of basis that fixes the ITS and $\E^{\lambda}_{k}=\sum_{\alpha}\Phi^{\lambda,\alpha}_{k}\<\E,\Phi^{\lambda,\alpha}_{k}\>$ depends only on the choice of basis for the $v^{(\lambda)}$ irrep. 
\subsection{Constructing a natural basis of ITS for $\T(A,A')$}
\label{S-canonical-ITS}
By building upon notions of ITOs we provide a way to construct an orthogonal basis of ITSs for the space of superoperators. While the construction does not a priori assume any choice of basis for the underlying Hilbert spaces or irreps -- in practice for computational purposes these will be hidden behind the ITOs and determining these requires basis choices. 

The irreducible representations that appear in the decomposition of $\T(A,A')$ are exactly those found in the product representation $U\otimes U^{*}\otimes U'\otimes (U')^{*}$ and we have denoted the set of all such irreps by $\rm{Irrep}(A,A')$. Each $\lambda\in \rm{Irrep}(A,A')$ can be thought of as arising in the tensor product of a $a$-irrep in the decomposition of $U\otimes U^{*}$ with a $\tilde{a}$-irrep in the decomposition of $U'\otimes (U')^{*}$. This represents a useful way to keep track of multiplicities of each irrep in $\T(A,A')$ which takes into account the natural way in which superoperators act on the input and output systems. This is because $\tilde{a}$ and $a$ also labels irreps in the decomposition of the output space $\B(\h_{A}')$ respectively the input space $\B(\h_{A})$. We then say that the $\lambda$-irrep has an associated multiplicity label that we denote by $m_{\lambda}=(\tilde{a},a)$. While in practice $\tilde{a}$ and $a$-irrep themselves can carry multiplicities as well as giving rise to more than one $\lambda$-irrep in the tensor product $\tilde{a}\otimes a$ we will often suppress this for simplicity of notation. We only make it explicit when it becomes a relevant issue. For instance one should keep in mind that in the case of SU(2) the tensor product of two irreps does not contain irreps with multiplicity greater than one in their decomposition.

\begin{theorem}
	For any $\lambda\in\rm{Irrep}(A,A')$ and multiplicity label $m_{\lambda}=(\tilde{a},a)$ we have that the set of superoperators in $\T(A,A')$ given by:
	\begin{equation}
	\Phi^{\lambda,m_{\lambda}}_{k}(\rho):=\sum\limits_{m,n}\clebsch{\tilde{a}}{m}{a}{n}{\lambda}{k}T^{\tilde{a}}_{m}\Tr(T^{a}_{n}\rho)
	\label{ITS}
	\end{equation}
	forms an orthogonal basis of irreducible tensor superoperators for $\T(A,A')$. $\{T^{\tilde{a}}_{m}\}$ and $\{T^{a}_{n}\}$ denote sets of ITOs for $\B(\h_A)$ respectively $\B(\h_A')$.
	\label{ITSTheorem}
\end{theorem}
\begin{proof}
	We just need to check that under the group action $\uu_{g}[\Phi^{\lambda,m_{\lambda}}_k]\stackrel{?}{=}\sum_{k'}v^{(\lambda)}_{kk'}(g)\Phi^{\lambda,m_{\lambda}}_{k'}$ where the matrix coefficients for the $\lambda$-irrep need to be consistent with the choice of basis assumed by the CGCs. Since $\{T^{\tilde{a}}_{m}\}$ in $\B(\h_{A}')$ and $\{T^{a}_{n}\}$ in $\B(\h_A)$ are ITOs then they will transform as:
	\begin{equation}
	\begin{split}
	&\U_{g}(T^{a}_{n})=\sum_{nn'}v^{(a)}_{nn'}(g)T^{a}_{n'}\\
	&\U_{g}'(T^{\tilde{a}}_{m})=\sum_{m'}v^{(\tilde{a})}_{mm'}(g)T^{a}_{m'}
	\end{split}
	\end{equation}
	and therefore the set of superoperators defined in the Equation \ref{ITS} will transform as:
	\begin{equation}
	\uu_{g}[\Phi^{\lambda,m_{\lambda}}_k(\rho)]=\hspace{-0.5cm}\sum_{m,n,m',n'}\clebsch{\tilde{a}}{m}{a}{n}{\lambda}{k}v^{(\tilde{a})}_{mm'}v^{(a)}_{nn'}T^{\tilde{a}}_{m'}\Tr(T^{a}_{n'}\rho)
	\label{ITStransf}
	\end{equation}
	Taking into account the definition of CGC and the fact that they represent a unitary change of basis then they must satisfy: $\sum\limits_{k'}\clebsch{\tilde{a}}{m'}{a}{n'}{\lambda}{k'}v^{(\lambda)}_{kk'}=\sum\limits_{m,n}\clebsch{\tilde{a}}{m}{a}{n}{\lambda}{k}v^{(\tilde{a})}_{mm'}v^{(a)}_{nn'}$. For more details on the relations between CGCs and matrix coefficients we refer the reader to \cite{Reps}. Finally by substituting this into Equation \ref{ITStransf} we obtain
	\begin{align}
	\uu_{g}[\Phi^{\lambda,m_{\lambda}}_{k}(\rho)]&=\sum\limits_{k',m',n'}\clebsch{\tilde{a}}{m'}{a}{n'}{\lambda}{k'}v^{(\lambda)}_{kk'}(g) T^{\tilde{a}}_{m'}\Tr(T^{a}_{n'}\rho)\nonumber \\
	&=\sum_{k'}v^{(\lambda)}_{kk'}(g)\Phi^{\lambda,m_{\lambda}}_{k'}.
	\end{align}
\end{proof}
\begin{figure}
	\begin{center}
		\includegraphics[width=8cm]{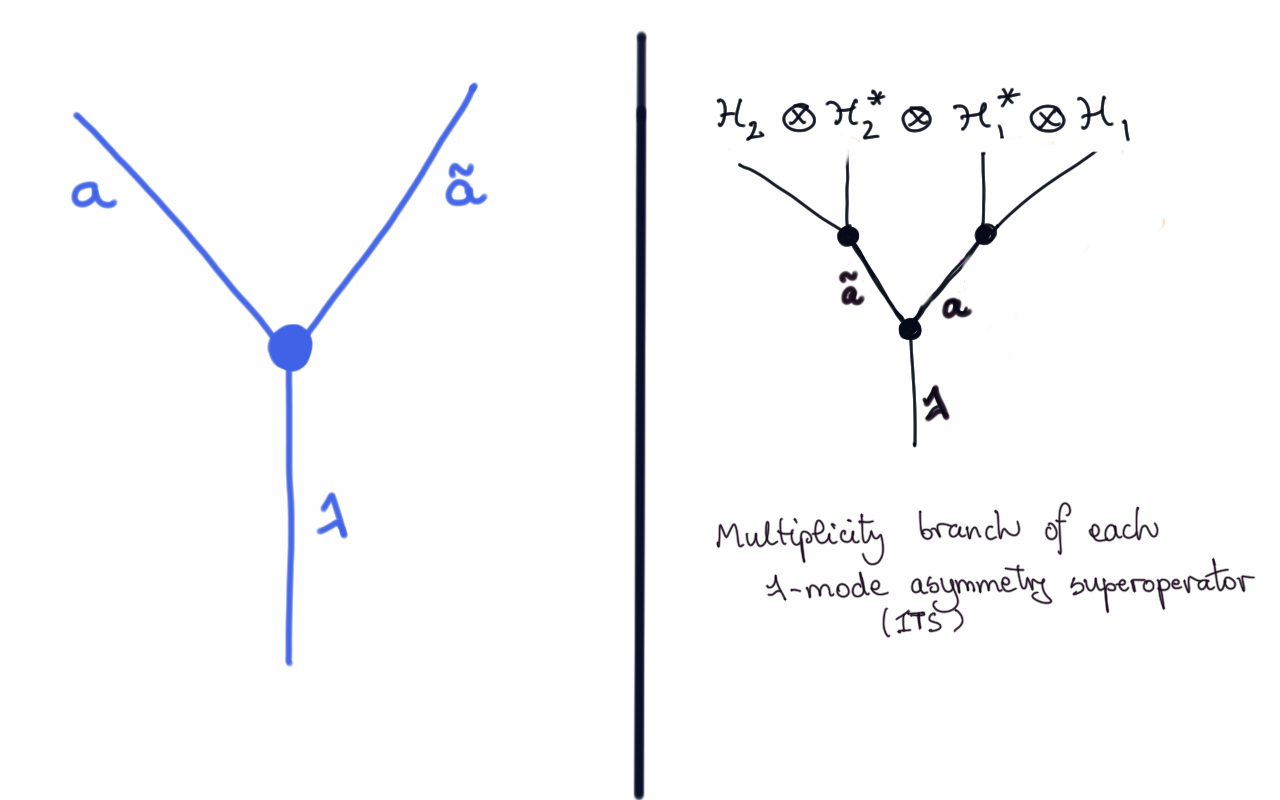}
		\caption{The edges label asymmetric resource species while the vertex represents the interaction of asymmetric resources. Mathematically it is a diagramatic representation of an intertwiner between couplings of irreducible representations. Physically - due to the directionality provided by the quantum operations - it has the interpretation that the input asymmetric mode interacts with some external asymmetry to produce the output asymmetric mode.  We can associate one such diagram to every irreducible tensor superoperator that decomposes the space $\T(\h_1,\h_2)$ of superoperators.
		}
	\end{center}
\end{figure}
In light of the above theorem we can identify a special type of ITSs that have support only on a single irreducible component in the input space and map it to a single irreducible component in the output space. 
\begin{definition}
	We define a complete set of \emph{canonical process modes} in $\T(A,A')$ to be the ITS set $\{\Phi^{\lambda,m_{\lambda}}_{k}\}_{k,\lambda}$ such that for any $\lambda\in\rm{Irrep}(A,A')$ and any corresponding multiplicity $m_{\lambda}$ there exists irreps in the input and output spaces labelled by $a$ and $\tilde{a}$ respectively such that for all $k$ the ITS $\Phi^{\lambda,m_{\lambda}}_{k}: V_{a}\longrightarrow V_{\tilde{a}}'$ is supported only on the a-irrep subspace of $\B(\h_A)$ and with range included only in the $\tilde{a}$-irrep subspace of $\B(\h_{A}')$
	\label{S-canonical-process-modes}
\end{definition}
While every canonical process mode will take the form given by Theorem \ref{ITSTheorem} for some choice of ITOs in the input and output spaces it is generally not true that any ITS can take this form. However the primitives are building blocks that allow us to construct any general ITS and therefore any basis that achieves the decomposition of $\T(A,A')$ into irreducible components. 
\begin{corollary}
	Any $\lambda$-irrep ITS in $\T(A,A')$ can be written as a linear combination of $\lambda$-irrep canonical process modes.
\end{corollary}
Moreover we draw attention on the fact that the construction of primitive ITS does not rely on a particular choice of basis for the input or output states since the only defining property is that it acts non-trivially only on a particular irreducible subspace and transforms it into another irreducible subspace.

\subsection{Homogeneous spaces}
\label{homogen}
Given a topological space $\M$ a \emph{group action} of $G$ on $\M$ is defined formally as $\cdot:G\times\M\longrightarrow\M$ such that it maps for every $g\in G$  a point $\mathbf{x}\in\M$ into $g\cdot\mathbf{x}\in \M$ and satisfies i) $g_1g_2\cdot\mathbf{x}=g_1\cdot(g_2\cdot\mathbf{x})$ and ii) $e\cdot\mathbf{x}=\mathbf{x}$ where $e$ is the identity element in $G$. A \emph{transitive} group action is one that allows to obtain every point on the manifold from an arbitrary initial point by acting with some group element: $\forall \mathbf{x},\mathbf{y} \  \exists \  \ g\in G$ such that $g\cdot \mathbf{x}=\mathbf{y}$. A $\emph{homogeneous space}$ is a space that has a transitive group action. Moreover all of these can be realised as quotient spaces $G/H=\{gH : g\in G\}$ for some closed subspace $H$ of $G$ carrying the subspace topology.
\subsubsection{Spherical Harmonics}
The spherical harmonics form a basis for the space $\mathcal{L}^2(S^2,\mathbb{C})$ of complex wavefunctions on the sphere. These are picked out through the decomposition of the space intro irreps. For spaces of functions this is conveniently expressed in terms of Lie derivatives. The Lie derivatives corresponding to the Lie algebra generators are given by $L_z$ and $L_\pm= L_x \pm i L_y$, and are obtained from the angular momentum operators $(L_x,L_y,L_z)$ defined in spatial coordinates as $L_i = \sum_{jk}\epsilon_{ijk} x_j \partial_j$. Going from cartesian to spherical coordinates $(\theta,\phi)$ it is easy to check that
\begin{align}
L_0 &= -2i \frac{\partial}{\partial \phi} \\
L_\pm &= e^{\pm i \phi}(\pm \frac{\partial}{\partial \theta} + i \cot \theta \frac{\partial}{\partial \phi}).
\end{align}
In terms of these differential operators we require that $L_0 Y_{j,m}(\theta,\phi) = 2m Y_{j,m}(\theta,\phi)$ and also $L_\pm Y_{j,m}(\theta, \phi) = \sqrt{(j\mp m)(j \pm m +1)} Y_{j,m\pm 1}(\theta,\phi)$.

It can be shown that imposing orthonormality implies that
\begin{equation}
Y_{jm} (\theta, \phi) = P_{jm} (\cos \theta) e^{im\phi},
\end{equation}
where $P_{jm}(x)$ are Legendre functions, with
\begin{align}
P_{jm}(x) &= (1-x^2)^{\frac{m}{2}} \frac{d^mP_j(x)}{dx^m }\\
P_j(x) &= \frac{1}{2^j j!}\frac{d^j (x^2-1)^j}{dx^j}.
\end{align}
\subsubsection{Harmonic analysis on homogeneous spaces}
\label{genSphHar}
We will be concerned with functions on homogeneous spaces and generalisations of spherical harmonics. We denote the space of square integrable functions on the compact\footnote{Unless stated otherwise we deal with compact groups only for many  reasons: i) the irreducible representations are finite dimensional ii) the left and right group actions coincide iii) there is a unique (up to scaling) Haar measure $(G,d\,g)$ iv) irreducibility and indecomposability are equivalent notions v) every finite representation has an equivalent unitary representation . Generally representation theory for compact groups is very well understood.} group $G$ by  by $\mathcal{L}^{2}(G)=\{f:G\longrightarrow \mathbb{C} : \int_{G} |f(g)|^2d\,g\leq \infty\}$. It carries a natural group action given by the \emph{left regular representation}
\begin{equation}
g\cdot f(h)=f(g^{-1}h)
\end{equation}
for all $g,h\in G$. From the perspective of group theory function spaces are useful objects because they encompass the representation theory structure of the respective group. In particular every $\lambda$-irrep of $G$ is realised in $\mathcal{L}^{2}(G)$ under the above group action with multiplicity equal to $\rm{dim}(\lambda)$. This is the content of Peter-Weyl theorem which importantly leads to the Fourier series decomposition of functions when the group $G$ is abelian. One of the consequences of this result is that any function $f\in \mathcal{L}^{2}(G)$ can be written as linear combinations of matrix coefficients ranging over all irreps
\begin{equation}
f(g)=\sum_{\lambda\in \hat{G}}\sum_{i,j=1}^{dim(\lambda)} \tilde{f}(\lambda,i,j)v^{\lambda}_{ij}(g)
\label{fourier}
\end{equation}
for some complex coefficients $\tilde{f}(\lambda,i,j)\in \mathbb{C}$ that can be recovered via the integral formula:$\tilde{f}(\lambda,i,j)=\int_G f(g)(v^{\lambda}_{i,j}(g))^{*}d\,g$ which follows directly from orthonormality of matrix coefficients. For abelian groups $G=U(1)$ the irreps are one-dimensional and the matrix coefficients are given by the usual exponentials  (which also correspond to the characters of the irreps) $e^{i\lambda g}$ for $g\in U(1)$. Then Equation (\ref{fourier}) becomes the Fourier series decomposition of $f\in\mathcal{L}^{2}(U(1))$.

Spherical harmonics arise naturally when we consider functions on homogeneous spaces $G/H$ and at the heart of it lies a generalisation of Peter-Weyl's theorem to such spaces. A function on a homogeneous space $G/H$ is an $\mathcal{L}^{2}(G)$ function that is constant on the left cosets. This means that for all $g\in G$ and $h\in H$ $f(gh)=f(g)$. For these types of function the decomposition into irreducible matrix coefficients as given in Equation (\ref{fourier}) simplifies. In particular only those matrix coefficients that are constant on the left cosets appear in the expansion. In defining spherical harmonics we require the subgroup $H$ to be a \emph{massive} subgroup. A subgroup $H$ is called massive whenever for any $\lambda$-irrep given by $v^{\lambda}$ if there exists a non-zero $H$-invariant vector $\ket{\mathbf{n}}$ in the $\lambda$-irrep carrier space such that $ v^{\lambda}(h)\ket{\mathbf{n}}=\ket{\mathbf{n}}$ for all $h\in H$ then it is unique (up to scalar multiples). This ensures that each $\lambda$-irrep in the decomposition of $\mathcal{L}^{2}(G/H)$ into irreducible components has multiplicity at most one.\footnote{ For example, on the sphere $S^1$ one can either consider the action of $SO(3)$ or $SU(2)$ which contain the massive subgroups $SO(2)$ and $U(1)$ respectively. While every irrep of $SO(3)$ appears in $\mathcal{L}^{2}(S^1)$ only the odd dimensional irreps of SU(2) appear. One can easily check for instance that the $1/2$-irrep fo SU(2) does not have a $U(1)$-invariant vector and therefore this irrep does not appear in the decomposition of functions on a sphere. In this sense we cannot define spherical harmonics on $S^1$ for the 1/2-irrep therefore reflecting the need for extra care on the assumptions for $H$ in order to define generalised spherical harmonics on homogeneous spaces.} In particular each such unique irreducible subspace is spanned by an orthonormal set of functions on $G/H$ which are exactly the \emph{spherical harmonics}. The \emph{associated spherical harmonics} give an explicit such basis with respect to some fixed orthonormal basis $\{\ket{e_{k}}\}$ for the $\lambda$-irrep carrier space where we identify $\ket{e_{1}}:=\ket{\mathbf{n}}$ with the unique $H$-invariant vector for the $\lambda$-irrep and are given by:
\begin{equation}
Y_{\lambda,k}(gH):=v^{\lambda}_{k1}(g)=\<\mathbf{n}|v^{\lambda}(g^{-1})|e_{k}\>
\end{equation}
for $k=1,..., \rm{dim}(\lambda)$ and $gH$ a coset representing an element in $G/H$. It is straightforward to check that $v^{\lambda}_{k1}(g)$ are invariant on left cosets and therefore the above are well defined (orthonormal) functions on $\mathcal{L}^{2}(G/H)$ that span the $\lambda$-irrep subspace.

Spherical harmonics have many useful properties and there are different equivalent viewpoints to the above: through the Lie derivatives, as eigenfunctions of invariant differential operators. However we will be mostly concerned with how they transform under the action of the left regular representation. From the above definition it is straightforward to establish that they transform similarly to irreducible tensor operators. Therefore for any $g\in G$, $\lambda$-irrep appearing in $\mathcal{L}^{2}(G/H)$ and $k=1,..., \rm{dim}(\lambda)$
\begin{align}
g\cdot Y_{\lambda,k}(g_{0}H)&=Y_{\lambda,k}(g^{-1}g_{0}H)=\<\mathbf{n}|v^{\lambda}(g_{0}^{-1}g)|e_{k}\> \nonumber \\
&=\sum_j\<\mathbf{n}|v^{\lambda}(g_{0}^{-1})|e_{j}\>\<e_{j}|v^{\lambda}(g)|e_{k}\>  \nonumber \\
&=\sum_j Y_{\lambda,j}(g_{0}H) v^{\lambda}_{kj}(g).\nonumber
\end{align}

A rigorous account of spherical harmonics on general homogeneous spaces can be found in Chapter 2 of \cite{spherical}.

\subsection{Axial quantum operations}
\label{S-axial-operations}
Before we prove the main theorem regarding axial operations we will take a closer look at the technical aspects involved in defining the process orbit. Given a fixed operation $\E\in\mathcal{S}(A,A')$ with input and output spaces carrying representations of the general compact group $G$ one can construct the stabiliser (or isotropy) group consisting of all those elements $h\in G$ that leave the operation unchanged: $Stab(\E)=\{h\in G | \ \uu_{h}[\E]=\E\}$.  This also defines an equivalence relation between the group elements by identifying any two that belong to $Stab(\E)$ and therefore we can construct the quotient space which for us corresponds to the process orbit $\M(G,\E):=G/Stab(\E)$. By the orbit-stabiliser theorem the process orbit is homeomorphic (for compact groups) to the orbit of $\E$ under the group action and therefore $\M(G,\E)$ is a homogeneous space.  The orbit of  $\E$ under the group action is given by:
$
Orb(\E)=\{\mathfrak{U}_{g}[\E]: g\in G)\}
$.

The space of all superoperators $\T(A,A')$ splits up into disjoint orbits as $\cup_{\E}Orb(\E)=\T(A,A')$. One can uniquely specify any operation by identifying a choice of origin in a particular orbit together with a set of coordinates on the process orbit. In general not all orbits will be isomorphic. This means that we might need two different types of process orbits to describe two quantum processes acting on the same spaces. In particular every fully symmetric operation lies in a single point orbit and the process orbit is simply a point. The invariant data of quantum processes can be directly associated with the choice of origin which gives a distinguished operation for each orbit.
\subsubsection{The use of process orbits}
\label{S-use-of-process-orbits}
A few points warrant mention regarding this perspective. Firstly, note that extreme cases are easy to state. For the case that $\E$ is symmetric the set $\M(G,\E)$ is a single point and so has no structure to it. In the case that the quantum process $\E$ lacks any residual symmetry, the space $\M(G,\E)$ is the full group manifold $G$ (quotiented by any phase symmetries  $e^{i \theta(g)}$ in the representation), and the realisation of $\E$ under a symmetry constraint can be interpreted in terms of encoding a group element $g$ in some token quantum system as $g \mapsto  \sigma_g \in \B(\H_B)$ followed by a symmetric state discrimination to access this group element. However the group element encoding perspective ignores the fact that not all group elements are equal in the realisation of $\E$ in $A$. In contrast, the homogeneous space $\M(G,\E)$ boils things down to the essentials and provides uniformity in the data required. 

Secondly, transformations of $B$ relative to $A$ correspond to displacements on $\M(G,\E)$, and therefore the choice of origin on $\M(G,\E)$ is an arbitrary one that simply encodes how the physics in $A$ is related to the physics in $B$. It is therefore a gauge choice in the sense that the physics local to $A$ does not care about how the symmetry action on $B$ is defined relative to $A$. A redundancy exists in our description of the composite system that leaves observed physics invariant.
We justify the use of this terminology more in the following subsections and the Section II B of the main text.

Finally, in Section \ref{S-Irreversibility} of this Supplementary Material we will be interested in saying when fundamental irreversibility occurs in the use of a symmetry-breaking resource to induce a quantum process $\E$ on a subsystem $A$. This statement must explicitly depend on both the symmetry group and the particular quantum process involved. We shall argue that the geometry induced on $\M(G,\E)$ provides a natural perspective on whether a resource state can be used in a repeatable manner without degrading it. From the perspective of asymmetry the realisation of any quantum process on a quantum system $A$, subject to a symmetry constraint, in the presence of a perfect ``classical reference'' can be viewed abstractly as a mapping $\M(G,\E) \rightarrow \T(A,A')$ with $x \mapsto \E$. The role of $\M(G,\E)$ is to provide an asymptotic, classical set of coordinates demanded under the symmetry constraint. However if one attempts to realise $\E$ using a bounded non-classical system, then we still have a mapping from $\M(G,\E)$ into $\T(A,A')$ however now the resolving power on $\M(G,\E)$ is ``blurred'' due to the non-classical resources used.

\subsubsection{Paradigmatic example: One qubit unitary axial operation under SU(2)}
Before giving the general statement we look at an illustrative example of the decomposition of unitary axial operation on a qubit into irreducible tensors. This will help to clarify all the core ingredients. Let $V=e^{i\mathbf{n}\cdot\sigma}$ with corresponding one qubit process $\V$ that takes the form $\mathcal{V}=\sum_{j,k}\alpha_{jk}\Phi^{j}_{k}$ where without loss of generality assume the ITSs are orthonormal. This is an axial map that leaves invariant any qubit with Bloch vector aligned along the $\mathbf{n}$ axis.

Denote by $J^{j}_{k}:=J[\Phi^{j}_{k}]$ the Choi representation of the $\Phi^{j}_{k}$ superoperator. The Choi operator of the unitary process $\mathcal{V}$ is given by $J[\mathcal{V}]=\vc{V}\vcb{V}$. Therefore we can write the alpha-coefficients in the form:
\begin{equation}
\alpha_{jk}(\mathbf{n})=\vcb{e^{i\mathbf{n}\cdot\sigma}}(J^{j}_{k})\hc\vc{e^{i\mathbf{n}\cdot\sigma}}.
\end{equation}
There is a natural group action on the space of functions on a sphere $\alpha\in \mathcal{L}^{2}(S^2)$ which for any $g\in SU(2)$ is given by:
\begin{equation}
g\cdot \alpha(\mathbf{n})=\alpha(g^{-1}\cdot \mathbf{n})
\end{equation}
where the action of the group element $g\in SU(2)$ on the vector in $S^2$ is defined by $(g\cdot\mathbf{n})\cdot\sigma:= g(\mathbf{n}\cdot\sigma)g\hc$. By a simple direct calculation we can check that the set of alpha-coefficients $\alpha_{jk}$ for fixed $j$ transform under this group action exactly like the spherical harmonics. We get that:
\begin{equation}
\alpha_{jk}(g^{-1}\cdot \mathbf{n})=\sum_{k'} (v^{j}_{kk'}(g))^{*}\alpha_{jk'}(\mathbf{n}).
\end{equation}
Since the group action of SU(2) on the sphere is transitive this means that by fixing a representative process in the $Orb(\mathcal{V})$ with some initial fixed Bloch vector then from this and the above relation we obtain the alpha-coefficients for any process in $Orb(\mathcal{V})$. Moreover we argue that they will be proportional to the dual spherical harmonics on the sphere.

The spherical harmonics $Y_{jk}\in \mathcal{L}^{2}(S^2)$ form a complete orthonormal basis for square integrable functions on a sphere. They achieve the decomposition of $\mathcal{L}^2(S^2)$ into irreducible components and are simultaneous eigenfunctions for all rotationally invariant differential operators. Under the group of rotations for any $R\in SO(3)$ the \emph{standard} spherical harmonics transform according to:
\begin{equation}
Y_{jk}(R^{-1}\mathbf{n})=\sum_{k'}\D^{j}_{kk'}(R)Y_{jk'}(\mathbf{n})
\end{equation}
where $D^{j}$ are the Wigner matrix for the $j$-irrep of SO(3). As a homogeneous space the sphere is diffeomorphic to the quotient spaces $S^{2}\cong SU(2)/U(1)\cong SO(3)/SO(2)$. This leads to a left regular representation of SU(2) and SO(3) on the space $\mathcal{L}^2(S^2)$ that decomposes the space into multiplicity free irreducible components. While for SO(3) the decomposition will be achieved by the spherical harmonics and every possible $j$-irrep appears, for SU(2) only the odd-dimensional irreps appear and these are exactly the ones isomorphic under the double cover to those of SO(3). 

Suppose that the double cover map is given by $H: SU(2)\longrightarrow SO(3)$ and associates an element $H(g)=R\in SO(3)$ for some $g\in SU(2)$. We choose the matrix coefficients $v^{j}_{kk'}$ for the $j$-irrep of $SU(2)$ to be the Wigner matrix i.e  $v^{j}_{kk'}(g)=\mathcal{D}^{j}_{kk'}(R)$ (this in turn means that we construct the ITSs with respect to this choice of basis as well). Then the alpha coefficients $\alpha_{jk}$ transform like the dual spherical harmonics $Y_{j^{*}k}=(-1)^{k}Y_{j,-k}$ in their standard form using the Condon-Shortly phase. Therefore they live in the irreducible subspace spanned by these and consequently can be written as linear combinations of $Y_{j^{*}k}$ for $k=-j,...,j$. Furthermore transitivity of the group action and some algebraic manipulation detailed in the main proof will result in the following form for the coefficients:
\begin{equation}
\alpha_{jk}(\mathbf{n})=a_{j} (-1)^{k} Y_{j,-k}(\mathbf{n})
\end{equation}
for some complex number $a_{j}$ independent on the vector component label $k$ and which depends on the choice of origin for the process orbit or equivalently a representative axial map in $Orb(\mathcal{V})$. 

In conclusion any one-qubit unitary axial map with fixed direction $\mathbf{n}$ can be decomposed as:
\begin{equation}
\mathcal{V}=\sum_{j,k} a_{j} (-1)^{k} Y_{j,-k}(\mathbf{n})\Phi^{j}_{k}
\end{equation}
which only assumes that the choice of basis for the ITS is aligned with the choice of basis for the spherical harmonics. However it remains a basis-independent statement since only the relative alignment is fixed and not the particular basis choice. The following sections give the general results. 

\subsubsection{SU-2 symmetry}
\label{sphericalharmonics}
Axial operations $\E\in\mathcal{S}(A,A')$ under SU(2) are those that are not invariant under the full group but have some residual symmetry given by a subgroup $H$ such that $\uu_{h}[\E]=\E$ for all elements of $h\in H$. By the orbit-stabilizer there is a natural homeomorphism between the homogeneous space $SU(2)/H$ and $Orb(\E)$ which we denote by $T:SU(2)/H\longrightarrow  Orb(\E)$ and explicitly given by $T(gH)=\mathfrak{U}_{g}[\E]$ for all $g\in G$ where $gH$ denotes a coset in $SU(2)/H$.\footnote{The proof is fairly standard. For any $h\in H$ and any $g\in SU(2)$ the following also holds: $\mathfrak{U}_{gh}[\E]=\mathfrak{U}_{g}\circ\mathfrak{U}_{h}[\E]=\mathfrak{U}_{g}[\E]$. This implies that whenever $\mathfrak{U}_{g}[\E]=\mathfrak{U}_{g'}[\E]$ since $\E$ is fixed under the action of H we must have equivalently that $g^{-1}g'\in H$ which means that the group elements generate the same coset $gH=g'H$. Finally SU(2) acts transitively on $Orb(\E)$ so any operation in the orbit will be reached by some $g\in G$} 
Whenever $H$ is isomorphic to a U(1) subgroup of SU(2) then the process orbit $\mathcal{M}(\E,G)$ will be diffeomorphic to a 2-sphere $S^{2}$. This means we can associate to any axial quantum process some direction $\mathbf{\hat{n}}$.

\begin{theorem}
	Let $\E:\B(\H_A)\rightarrow\B(\H_{A'})$ be an axial process associated to some direction $\hat{n}$, that takes states of a quantum system $A$ into states of a system $A'$. Then the orbit of $\E$ under the symmetry action is a $2$-sphere $S^2$ and there exists a basis of irreducible tensor superoperators $\{\Phi^j_k\}$ with $\E(\rho)=\sum_{j,k} \alpha_{j,k}\Phi^j_{k}(\rho)$ such that the components $\alpha_{jk}$ are un-normalised wavefunctions on the sphere given by
	\begin{equation}
	\alpha_{j,k}(\theta,\phi) =a_j Y_{jk}(\theta, \phi),
	\end{equation}
	where $\hat{n}=(\theta,\phi)$, $a_j$ is the norm of $\alpha_{jk}$ that is independent of $k$, and $Y_{jk}(\theta,\phi)$ are spherical harmonics.
	\label{spherical harmonics}
\end{theorem}
\begin{proof}

	Because $SU(2)/U(1)\cong S^{2}$ means that the process orbit is a sphere so the orbit will be homeomorphic to $S^2$. We denote this by $f:S^{2}\longrightarrow Orb(\E)$. 
	
	Any operation in the $Orb(\E)$ is an axial operation since it will be also fixed by $H$. Suppose we denote by $\mathbf{x}\in S^2$ such that $f(\mathbf{x})=\E$. The coefficients when decomposing $\E$ into ITSs will depend on $\E$ and therefore on $\mathbf{x}$ and take the form $\E=\sum_{j,k}\alpha_{j,k}(\mathbf{x})\Phi^{j}_{k}$. This means that any axial operation in $Orb(\E)$ will decompose in terms of coefficients that we can view as complex-valued functions on $S^{2}$ such that $\alpha_{j,k}:S^{2}\longrightarrow \mathbb{C}$ .
	Any axial operation $\mathfrak{U}_{g}[\E]\in Orb(\E)$ has a corresponding a unique vector in $S^{2}$ such that $f(\mathbf{x}')= \mathfrak{U}_{g}[\E]$ where $\mathbf{x}'$ depends on the fixed $g$. The transitive group action on the orbit induces a group action on $S^2$ given by $\mathbf{x'}:=g\cdot\mathbf{x}=f^{-1}[\uu_{g}(f(\mathbf{x}))] $. In turn since the coefficients $\alpha_{j,k}\in \mathcal{L}^{2}(SU(2)/U(1))$ we have an induced group action on the function space given by $g\cdot \alpha_{j,k}(\mathbf{x})=\alpha_{j,k}(g^{-1}\cdot\mathbf{x})$. Moreover this group action will be equivalent to the left regular representation since we can identify the orbit with the coset space. 
	
	As $\E$ transforms under the group action it generates only operations in $Orb(\E)$ so we can explicitly show how the coefficients transform under the induced group action by expressing $\alpha_{j,k}(\mathbf{x'})=g\cdot\alpha_{j,k}(\mathbf{x})$ in terms of only the set of coefficients evaluated at $\mathbf{x}$ i.e $\alpha_{j',k'}(\mathbf{x})$ for all $j'$ and $k'$. We show the following:
	
	{\bf{Claim 1}} \emph{The alpha-coefficients will transform under the induced group action as:
		\begin{equation}
		g\cdot \alpha_{jk}(\mathbf{x})=\sum_{k'}v^{j^{*}}_{kk'}(g)\alpha_{jk'}(\mathbf{x})
		\end{equation}
		for any $g\in SU(2)$ and any irrep $j\in\rm{Irrep}(A,A')$ with $v^{j^{*}}_{kk'}$ matrix coefficients of the dual irrep $j^{*}$.}
	\emph{Proof of claim:}\\
	We have that:
	\begin{equation}
	\mathfrak{U}_{g}[\E]=\sum\limits_{j,k} \alpha_{j,k}(\mathbf{x'})\Phi^{j}_{k}.
	\end{equation}
	We want to express $\alpha_{j,k}(\mathbf{x}')$ in terms of $\alpha_{j,k}(\mathbf{x})$ to explicitly show how the alpha coefficients transform under the group action. Since  $\Phi^{j}_{k}$ transform as irreducible tensor superoperators we have that:
	\begin{equation}
	\begin{split}
	\mathfrak{U}_{g}[\E]&=\sum\limits_{j,k} \alpha_{j,k}(\mathbf{x}) \mathfrak{U}_{g}[\Phi^{j}_{k}]\\
	&=\sum\limits_{j,k} v^{j}_{kk'}(g)\alpha_{j,k}(\mathbf{x})\Phi^{j}_{k'}.
	\end{split}
	\end{equation}
	The ITS sets $\Phi^{j}_{k}$ form a complete basis for the space of superoperators and in particular they are linearly independent. Then equating the two different expressions for $\uu_{g}[\E]$ and using orthogonality leads to
	\begin{equation}
	\alpha_{j,k}(\mathbf{x}')=\sum\limits_{j,k'}v^{j}_{k'k}(g)\alpha_{j,k'}(\mathbf{x}).
	\end{equation}
	Since $\mathbf{x}=g\cdot \mathbf{x}$ we can re-write the above as $g^{-1}\cdot\alpha_{j,k}(\mathbf{x})=\sum\limits_{j,k'}v^{j}_{k'k}(g)\alpha_{j,k'}(\mathbf{x})$. However there is nothing special in our choice of the particular element $g$ and corresponding point $\mathbf{x}'$ so it turns out that this occurs for all elements $g\in SU(2)$. Because we can use $(v^{j}_{kk'}(g))^{*}=v^{j}_{k'k}(g^{-1})$ and $(v^{j}_{kk'}(g))^{*}=v^{j^{*}}_{kk'}(g)$ where $j^{*}$ is the dual irrep this leads to the desired transformation
	\begin{equation}
	g\cdot \alpha_{j,k}(\mathbf{x})=\alpha_{j,k}(g^{-1}\cdot \mathbf{x})=\sum_{j,k'}v^{j^{*}}_{kk'}(g)\alpha_{j,k'}(\mathbf{x}).
	\label{transform}
	\end{equation}
	This ends our proof of Claim 1.
	
	It is immediate to check that $\alpha_{j,k}\in \mathcal{L}^{2}(S^2,\mathbb{C})$. However square integrable functions on the sphere decompose into a complete set of spherical harmonics. It is important to mention at this point that the irrep-decomposition of $\mathcal{L}^{2}(S^2,\mathbb{C})$ is multiplicity-free and each $j$-irrep component is spanned by a complete basis of orthonormal functions $Y_{jk}$ for $k=-j,...j$ which correspond to the spherical harmonics. The fact that the isotypical decomposition is multiplicity free is exactly what allows us to define spherical harmonics for this homogeneous space in the first place.
	Since $\alpha_{jk}$ are functions on the manifold $Orb(\E)$ and transform according to Equation (\ref{transform}) then they lie in the $j^{*}$-irrep component of $L^{2}(S^2,\mathbf{C})$ that is spanned by spherical harmonics $Y_{j^{*}k}$. Therefore we can write them in terms of a spherical harmonics basis such that:
	\begin{equation}
	\alpha_{jk}(\mathbf{x})=\sum\limits_{k'} a_{jk'}^{(k)}Y_{j^{*}k'}(\mathbf{x)}
	\end{equation}
	for some complex coefficients $a_{jk'}^{k}$ depending on some fixed j and k. We show that $a_{jk'}^{k}=\delta_{kk'} a_{j}$ for some complex number $a_{j}$ that is independent on the vector component $k$. To do so note that the above should hold for all $\mathbf{x}\in S^{2}$. So we have:
	\begin{equation}
	\begin{split}
	\alpha_{jk}(g^{-1}\cdot \mathbf{x})&=\sum\limits_{k'} a_{jk'}^{(k)}\left(Y_{jk'}((g^{-1}\cdot \mathbf{x)}\right)\\
	&=\sum\limits_{k',m}a_{jk'}^{(k)} v^{j*}_{k'm}(g)Y_{j^{*}m}(\mathbf{x})
	\end{split}
	\end{equation}
	and similarly we have that the alpha coefficients transform as:
	\begin{equation}
	\begin{split}
	\alpha_{jk}(g^{-1}\cdot\mathbf{x})&=\sum\limits_{k'} v^{j^{*}}_{kk'}(g)\alpha_{jk'}(\mathbf{x})\\
	&=\sum\limits_{m,k'} v^{j^{*}}_{kk'}(g) a^{(k')}_{jm} Y_{j^{*}m}(\mathbf{x}).
	\end{split}
	\end{equation}
	Using orthonormality of spherical harmonics we can equate the two different forms for the transformed alpha-coefficients to obtain that for all $m$ the following holds:
	\begin{equation}
	\sum\limits_{k'} v^{j^{*}}_{kk'}(g) a^{(k')}_{jm} =\sum\limits_{k'}a_{jk'}^{(k)} v^{j^{*}}_{k'm}(g).
	\end{equation}
	Now we can use orthonormality of matrix coefficients for the $j^{*}$ irrep to multiply both sides by $(v^{j^{*}}_{ls})^{*}$ and integrate over all group elements to get
	\begin{equation}
	\sum\limits_{k'}\delta_{kl}\delta_{k's}a^{(k')}_{jm}=\sum\limits_{k'} \delta_{k'l}\delta_{ms}a^{(k)}_{jk'}
	\end{equation}
	and therefore reduce to $a^{(s)}_{jm}\delta_{kl}=\delta_{ms}a^{(k)}_{jl}$ for all $m$. Alternatively this says that $a^{(s)}_{jm}=\delta_{ms}a^{(l)}_{jl}$ which means that the coefficients are all independent on the vector components and we have that $\alpha_{jk}=a_{j}Y_{j^{*}k}$ where $a_{j}=a^{(k)}_{jk}$ for any $k$. For the case of SU(2) the irreps sand their duals are isomorphic. Moreover under the choice of Condon-Shortly phase the spherical harmonics for the dual representation are related to the usual spherical harmonics by $Y_{j^{*}k}=(-1)^{k}Y_{j,-k}$. Therefore under some basis choice for the ITS/spherical harmonics the alpha-cofficients  will take the form:
	\begin{equation}
	\alpha_{j,k}(\mathbf{x})=a_{j}(-1)^{k}Y_{j,-k}(\mathbf{x})
	\end{equation}
	for any $\mathbf{x}\in S^2$ and some $a_{j}$ independent on the vector component labelled
	by $k$.
\end{proof}
\subsubsection{Axial operations under general symmetry}
\label{generalSph}
The results in the previous subsection can be generalised to arbitrary compact groups $G$.
\begin{theorem}
	Let $G$ be a compact group with unitary representations $U_{A}$ and $U_{A'}$ on the Hilbert spaces $\h_A$ and $\h_{A'}$. Consider operations $\E$ in $\T(A,A')$ that have an isotropy group $H$ subgroup of $G$. Then any such operation can be decomposed into irreducible tensor superoperators with corresponding coefficients proportional to harmonics functions on the homogeneous space $G/H$. That is:
	\begin{equation}
	\E=\sum_{\lambda,k} a_{\lambda}Y_{\lambda^{*},k}(\mathbf{x})\Phi^{\lambda}_{k}
	\end{equation}
	where $\mathbf{x}\in G/H$ correspond to the point in the process orbit associated with $\E$ and the complex coefficients $a_{\lambda}\in\mathbb{C}$ correspond to invariant data that are fixed for any operation in $Orb(\E)$. 
	\label{S-polar-decomposition}
\end{theorem}
\begin{proof}
	Any $\E$ can be decomposed in terms of irreducible superoperators as $\E=\sum_{\lambda,k}\alpha_{\lambda,k}\Phi^{\lambda}_{k}$.  However since $\E$ is also an axial operation with stabilizer (or isotropy) group $H=\{g\in G : \uu_{g}[\E]=\E\}$. Moreover any operation in $Orb(\E)$ will also be an axial operation with the same stabilizer group. By the orbit-stabilizer theorem there is an homeomorphism $Orb(\E)\cong G/H$. Clearly $G/H$ is a homogeneous space because $G$ acts transitively on any single orbit. Therefore we can associate to any point $\mathbf{x}\in G/H$ an axial operation in $Orb(\E)$, which also was previously referred to as the process orbit $\M(G,\E)$. This in turn implies that we can view the corresponding alpha-coefficients for each such axial operation as functions on the process orbit. In particular we can show that $\alpha_{\lambda,k}\in \mathcal{L}^{2}(G/H)$. From orthonormality of ITS we have that for any $\lambda,k$ $\alpha_{\lambda,k}(\mathbf{x})=\<\Phi^{\lambda}_{k},\E\>$ with $\mathbf{x}$ being the associated coordinate point of $\E$ on the process orbit. Because the group acts transitively for any $\mathbf{y}\in G/H$ there exists $g\in G$ such that $g\cdot\mathbf{x}=\mathbf{y}$. Therefore:
	\begin{equation}
	\int_{G/H} |\alpha_{\lambda,k}(\mathbf{y})|^2 d\,\mathbf{y}=\int_{G}|\alpha_{\lambda,k}(g\cdot\mathbf{x})|^2 d\,g.
	\end{equation}
	However $\alpha_{\lambda,k}(g\cdot \mathbf{x})=\<\Phi^{\lambda}_{k},\uu_{g}[\E]\> =\<\uu_{g}\hc[\Phi^{\lambda}_{k}],\E\>$ and hence we obtain that 
	\begin{equation}
	\alpha_{\lambda,k}(g\cdot\mathbf{x})=\sum_{j}v^{\lambda}_{jk}(g)\<\Phi^{\lambda}_{j},\E\>=\sum_{j}v^{\lambda}_{jk}(g)\alpha_{\lambda,j}(\mathbf{x})
	\end{equation}
	which by substituting in the above integral we get straight away from orthonormality of matrix coefficients that $\int_{G/H}|\alpha_{\lambda,k}(\mathbf{y})|^2 d\,\mathbf{y}<\infty$ and therefore $\alpha_{\lambda,k}\in \mathcal{L}^{2}(G/H)$. By a repeat of the argument for the normalization in for axial maps, we find that the norm of $\alpha_{\lambda,k}$ is independent of $k$ and so we can write $\alpha_{\lambda,k}(\x) = a_\lambda Y_{\lambda^*,k}(\x)$ as claimed.

\end{proof}

In the case that $H$ is a massive subgroup, the space $\mathcal{L}^{2}(G/H)$ decomposes under the left regular representation into irreducible subspaces all appearing with multiplicity one. Unlike the decomposition of $\mathcal{L}^2(G)$ not all irreps of $G$ appear in the decomposition of $\mathcal{L}^{2}(G/H)$. However each such $\lambda$ irreducible subspace has a basis of spherical harmonic functions $Y_{\lambda,k}$ with the vector label component ranging from $k=1,...\rm{dim}(\lambda)$. However the alpha-coefficients transform like the $\lambda^{*}$-irrep under the left regular representation and hence they can be written in terms of the corresponding spherical harmonics $Y^{\lambda^{*}}_{k}$.

\subsection{One qubit operations under SU(2) symmetry}
\label{S-single-qubit-process-modes}
Consider a one qubit system $\h\cong\mathbb{C}^{2}$ which carries the 1/2-irreducible representation of SU(2). This is the fundamental representation of SU(2) whose action is given by matrix multiplication. The space of operators $\B(\h)$ decomposes according to this symmetry into two orthogonal subspaces carrying the trivial 0-irrep respectively the 1-irrep. An analogous way of expressing this statement is through the Bloch-vector representation of a qubit state. Any valid density matrix $\rho\in \B(\h)$ can be written as $\rho=\frac{1}{2}\left(\iden+\mathbf{r}\cdot\boldsymbol{\sigma}\right)$ where $\boldsymbol{\sigma}=(\sigma_x,\sigma_y,\sigma_z)$ is the vector of Pauli matrices and $\mathbf{r}$ represents the Bloch vector associated with $\rho$. The Pauli matrices transform non-trivially under the adjoint action of any SU(2) matrix and they span the 1-irrep subspace whereas the identity is invariant under the group action and corresponds to the trivial 0-irrep subspace. 

According to this symmetry any one-qubit quantum process can be decomposed into ITSs -- each of which will act non-trivially only on one of the two modes $a\in\{0,1\}$ and transform them into a different mode $\tilde{a}\in\{0,1\}$ according to the $\lambda$-irrep label that captures how the ITS behaves under the group action. To each ITS we associate a particular diagram label $\theta=(a,\tilde{a})\stackrel{\lambda}{\longrightarrow}$ where $\lambda\in\{0,1,2\}$ with multiplicity two, three and one respectively. The diagram $(1,0)\stackrel{1}{\longrightarrow}$ is unphysical and therefore does not appear in the decomposition of any valid quantum process. Out of the total five remaining possible ITSs, two of them $(0,0)\stackrel{0}{\longrightarrow}$ and $(0,1)\stackrel{1}{\longrightarrow}$ give an output independent of the Bloch vector of the initial state $\rho\in\B(\h)$. The first corresponds to the superoperator $\Phi_{0}(\rho)=\frac{\iden}{2}$ and the second span a 3-dimensional superoperator subspace that displace the maximally mixed state and therefore appear only in non-unital quantum processes. The latter we denote by $\Phi^{1,\tilde{m}_{1}}_{k}(\rho)=\boldsymbol{\sigma}_{k}$ where the multiplicity label $\tilde{m}_{1}$ specifies that this particular 1-irrep ITS arrises from coupling $(0,1)\stackrel{1}{\longrightarrow}$ and $k$ corresponds to the vector component. The space of invariant one-qubit superoperators is two dimensional and the most general such quantum process is a depolarised process. It is spanned by $\Phi_{0}$ and the ITS given by diagram $(1,1)\stackrel{0}{\longrightarrow}$ which we explicitly denote by $\Phi^{0,m_0}(\rho)=\frac{\mathbf{r}\cdot\boldsymbol{\sigma}}{\sqrt{3}}=\rho-\frac{\iden}{2}$. This leaves only two non-trivial diagrams corresponding to both input and output mode carrying asymmetry. The remaining 1-irrep ITS corresponds to coupling $(1,1)\stackrel{1}{\longrightarrow}$ which we label with multiplicity $m_1$ will be given by:
\begin{equation}
\begin{split}
\Phi^{1,m_1}_{-1}&=-\frac{1}{2\sqrt{2}}\left(\sigma_{+}\Tr(\sigma_z\rho)-\sigma_z\Tr(\sigma_{+}\rho)\right)\\
\Phi^{1,m_1}_{0}&=\frac{i}{2\sqrt{2}}\left(\sigma_x\Tr(\sigma_y\rho)-\sigma_y\Tr(\sigma_x\rho)\right)\\
\Phi^{1,m_1}_{1}&=-\frac{1}{2\sqrt{2}}\left(\sigma_{-}\Tr(\sigma_z\rho)-\sigma_z\Tr(\sigma_{-}\rho)\right).
\end{split}
\end{equation}
Similarly for the 5-dimensional subspace spanned by superoperators corresponding to the diagram $(1,1)\stackrel{2}{\longrightarrow}$ will be described by the following (where we suppress the multiplicity label since there is a single component for the 2-irrep)
\begin{align}
\Phi^{2}_{-2}(\rho)&=\frac{1}{2}\left( (\sigma_{+})\Tr((\sigma_{+})\rho)\right)\nonumber\\
\Phi^{2}_{-1}(\rho)&=-\frac{1}{2\sqrt{2}}\left( (\sigma_{+})\Tr(\sigma_z\rho)+\sigma_z\Tr(\sigma_{+}\rho) \right)\nonumber\\
\Phi^{2}_{0}(\rho)&=-\frac{1}{2\sqrt{6}}\left(-\sigma_x\Tr(\sigma_x\rho)/2-\sigma_y\Tr(\sigma_y\rho)/2+\sigma_z\Tr(\sigma_z\rho)\right)\nonumber\\
\Phi^{2}_{1}(\rho)&=-(\Phi^{2}_{-1}(\rho))\hc\nonumber\\
\Phi^{2}_{2}(\rho)&=(\Phi^{2}_{-2}(\rho))\hc.
\end{align}
The most general axial quantum process on one qubit with fixed axis $\mathbf{n}=(x,y,z)$ takes the diagramatic form:
\begin{equation}
\begin{split}
\E&=\Phi_{0}+\alpha_{0}\Phi^{0,m_0}+\tilde{\alpha}_{1}\boldsymbol{Y}^{1^{*}}(\mathbf{n})\cdot\boldsymbol{\Phi}^{1,\tilde{m}_1}\\
&+\alpha_{1}\boldsymbol{Y}^{1^{*}}(\mathbf{n})\cdot\boldsymbol{\Phi}^{1,m_1}+\alpha_2\boldsymbol{Y}^{2^{*}}(\mathbf{n})\cdot\boldsymbol{\Phi}^{2}.
\end{split}
\end{equation}
\section{Biparitite symmetric operations}
\label{bipartiteproof}
\subsection{General structure theorem}
\label{S-bipartite-theorem}
In this section we will be concerned with fully characterising bipartite symmetric operations under general symmetry constraints. 
\begin{theorem}
	Let $G$ be a compact group and consider an input quantum system described by $\H_{\rm \tiny in}=\H_{A}\otimes\H_{B}$ and an output quantum system described by $\H_{\rm \tiny out}=\H_{A'}\otimes\H_{B'}$. These spaces carry unitary representations of $G$ given by $U_{A}\otimes U_{B}$ and $U_{A'}\otimes U_{B'}$ respectively. Then every symmetric process $\E$ can be written as
	\begin{align}
	\E &= \sum_{\lambda, m, m'}\Phi_{\lambda, m, m'} \nonumber \\
	\Phi_{\lambda, m, m'}&=\sum_{k=1}^{\rm{dim} \lambda} \mathcal{A}^{\lambda,m}_k\otimes \mathcal{B}^{\lambda^{*},m'}_{k}
	\end{align}
	where $\lambda \in \rm{Irrep}(A,A') \cap \rm{Irrep} (B, B')$, $m, m'$ are multiplicity labels for the irrep and where $\{\A^{\lambda_{A},m_{A}}_{k}\}$ forms a complete set of ITS for $\T(\H_{A},\H_{A'})$ and similarly $\{\B^{\lambda_{B},m_{B}}_{k}\}$ forms a complete basis set of ITS for $\T(\H_B,\H_B')$.
	\label{main}
\end{theorem}
\begin{proof}
	First we show that every $\Phi_{\lambda,m,m'}$ constructed this way are symmetric under the group action on space of bipartite superoperators $\T(A\otimes B,A'\otimes B)$:
	\begin{align}
	\uu_{g}(\Phi_{\lambda,m,m'})&=\sum\U_{A'}(g)\circ\A^{\lambda,m}_{k}\circ \U_{A}\hc(g)\otimes\nonumber\\
	& \otimes \U_{B'}(g)\circ\B^{\lambda^{*},m'}_{k}\circ \U_{B}\hc(g) \nonumber \\
	&=\sum_{k}v^{\lambda}_{kk'}(g)v^{\lambda^{*}}_{kj}(g)\A^{\lambda,m}_{k'}\otimes\B^{\lambda^{*},m'}_{j}\nonumber
	\end{align}
	where in the last equation we only used the fact that $\{A^{\lambda}_{k}\}$ and $\{B^{\lambda^{*}}_{k}\}$ form complete sets of ITSs. Observe that there is an underlying assumption that both sets of ITS transform relative to the matrix coefficients computed with respect to the same basis. Moreover the standard $\lambda$-irrep matrix is unitary. Therefore  $\sum_{k}v^{\lambda}_{kk'}(g)v^{\lambda^{*}}_{kj}(g)=\sum_{k}v^{\lambda}_{kk'}(g)v^{\lambda}_{jk}(g^{-1})=(v^{\lambda}v^{\lambda})\hc_{jk'}=\delta_{jk'}$ where the last equality comes from unitarity and have used the dual representation is given by $v^{\lambda^{*}}(g)=v^{\lambda}(g^{-1})^{T}$. This implies that for any $\lambda$ and multiplicities $m$ and $m'$ we have that\footnote{Remark: Dealing with matrix coefficients gives a false sense of simplicity. Orthonormality of matrix coefficients is actually a powerful result that can be linked to Peter-Weyl's theorem and Schur's lemma. These are technical statements which are central to many important results in group theory. The core underlying result of the main theorem we prove is a consequence of these: {\it{The tensor product of two irreducible representations $\lambda\otimes\mu$ contains at most one trivial irrep and this occurs if and only if $\mu\cong \lambda^{*}$ }}}:
	\begin{align}
	\uu_{g}(\Phi_{\lambda,m,m'})=\sum_{j}\A^{\lambda,m}_{j}\otimes\B^{\lambda^{*},m'}_{j}=\Phi_{\lambda,m,m'}\nonumber
	\end{align}
	This implies that each of the superoperators $\Phi^{\lambda,m,m'}$ are invariant under the group action. Finally all we have left to prove is they span the whole invariant subspace of superoperators on the bipartite system considered. To do so we need to characterise the trivial subspaces and their multiplicity in the decomposition of $\T(A\otimes B, A'\otimes B')$ into irreducible components. All these irreducible components are exactly the ones that appear in the decomposition of the tensor product of representations $U_{A'}\otimes U_{A'}^{*}\otimes U_{B'}\otimes U_{B'}^{*}\otimes U_{A}\otimes U_{A}^{*}\otimes U_{B}\otimes U_{B}^{*}$. The irreducible components remain the same under swaping the tensors in the previous representation. Therefore any irrep in $\T(A\otimes B, A'\otimes B')$ arises in decomposing the tensor product coupling of an irrep $\lambda_{A}\in \rm{Irrep}(A,A')$ and $\lambda_{B}\in\rm{Irrep}(B,B')$. This accounts for all irreps in the space of superoperators on the bipartite system. In particular a classical result in representation theory says that there is at most one trivial 0-irrep subspace in the decomposition of $\lambda_{A}\otimes\lambda_{B}$ and appears if and only if $\lambda_{B}\cong \lambda_{A}^{*}$. Therefore as we range over all $\lambda\in \rm{Irrep}(A,A')$ with multiplicity $m$ and $\lambda^{*}\in \rm{Irrep}(B,B')$ (if it exists) with multiplicity $m'$ we generate all possible trivial components. Each of them will be distinct because they arise in orthogonal subspaces. Most easily this can be seen by the orthogonality relations:
	\begin{align}
	\<\Phi_{\lambda,m,m'},\Phi_{\nu,n,n'}\>&=\sum_{k,j}\<\A^{\lambda,m}_{k},\A^{\nu,n}_{j} \>\<\B^{\lambda^{*},m'}_{k},\B^{\nu^{*},n'}_{j}\>\nonumber\\
	&=\rm{dim}(\lambda)\delta_{\lambda,\nu}\delta_{m,n}\delta_{m',n'}
	\end{align}
	where in the last equation we used the fact that ITS at $\T(A,A')$ and $\T(B,B')$ are orthogonal (and we also assume without loss of generality that they are normalised).
\end{proof}
While the above theorem holds for general local ITS at system $A$ and $B$ we will be interested in decomposing general symmetric operations in terms of primitive local ITS. This means that the complete basis of symmetric superoperators can be labelled by diagrams $\theta=[(a,\tilde{a})\stackrel{\lambda}{\longrightarrow}(b,\tilde{b})]$ which reflect what are the local ITS couplings that make up each basis element. In particular in terms of the allowed symmetry-breaking properties of the input and output states it has the following mathematical form:

	\begin{equation}
	\Phi_{\theta}(\rho_{AB})=\sum_{m,m',n,n',k}\clebsch{\tilde{a}}{m}{a}{n}{\lambda}{k}\clebsch{\tilde{b}}{m'}{b}{n'}{\lambda^{*}}{k} T^{\tilde{a}}_{m}\otimes T^{\tilde{b}}_{m'}\Tr(T^{a}_{n}\otimes T^{b}_{n'}\rho_{AB})
	\label{diagrambasis}
	\end{equation}

where the operators $\{T^{\tilde{a}}_{m}\},\{T^{\tilde{b}}_{m'}\}, \{T^{a}_{n}\}, \{T^{b}_{n'}\}$ form complete sets of orthonormal ITOs for $\B(\h_{A'}),\B(\h_{B'}),\B({\h_{A}})$ respectively $\B(\h_{B})$. By ranging over all possible diagrams Equation \ref{diagrambasis} gives an orthogonal complete basis of symmetric bipartite superoperators.  We make several remarks on the generality of such a construction:\\
i) There is no a priori choice of basis for the underlying Hilbert spaces and there is freedom in choosing a basis for the local input and output states. \\
ii) There is however an assumption in lifting the choice of basis from the local Hilbert space to basis of local (and global) superoperators. In particular the choice of basis for superoperators is such that it maps input basis states to output basis states, but generally there are valid orthonormal superoperator basis that do not act in this way. The reason for this construction is purely operational: it allows us to analyse both local symmetry-breaking properties of the bipartite states and of the global symmetric processes as well as how these two concepts interact with one another.\\
iii) The notation for each diagram $\theta$ and underlying irrep couplings does not explicitly include multiplicities. However it should be understood that when we identify a particular diagram it arises from particular couplings of input and output irreps having specific multiplicities.

\begin{lemma} Given  $\E\in \T(AB,A'B')$, if $\E$ has (a) non-trivial subsystems $A,B,A',B'$, (b) globally symmetry, and (c) complete-positivity, then $\E$ must take the form
	\begin{equation}
	\E = \sum_\theta ( x_\theta \Phi_\theta + x_\theta^* \Phi_{\theta^*}),
	\label{dualCP}
	\end{equation}
	where $x_\theta \in \mathbb{C}$ and where $\theta^*$ is the dual diagram to $\theta$. Moreover trace-preservation implies $x_\theta=\sqrt{\frac{d_Ad_B}{d_{A'} d_{B'}}}$ for $\theta = [ (0,0)\stackrel{0}{\longrightarrow}(0,0)]$ and $x_\theta=0$ for $\theta=[(\lambda,0)\stackrel{\lambda}{\longrightarrow}(\lambda,0)]$ for any non-trivial $\lambda$-irrep.
\end{lemma}
\begin{proof}
	Complete positivity is easiest to keep track in the Choi picture. $\E$ is CPTP map if and only if the corresponding Choi operator satisfies $J[\E]\geq 0$ and $\Tr_{A'B'}(J[\E])=\iden_{AB}$.
	In particular we also have that $J[\E]=J[\E]\hc$ and therefore since $J[\Phi_{\theta}]\hc=J[\Phi_{\theta^{*}}]$ we get that $\E$ takes the form of Equation (\ref{dualCP}). By construction of the primitive ITS the only diagrams contributing to the trace i.e for which $\Tr(\Phi_{\theta})\neq 0$ correspond to $\tilde{a}=\tilde{b}=0$ where $\theta=[(a,\tilde{a})\stackrel{\lambda}{\longrightarrow} (b,\tilde{b})]$. However for all these types of diagrams we have that $\Tr_{A'B'}(J[\Phi_{\theta}])=\sqrt{d_{A'}d_{B'}} \sum_{k}(A^{\lambda}_{k}\otimes B^{\lambda^{*}}_{k})^{T}$ where $A^{\lambda}_{k}$ and $B^{\lambda}_{k}$ are ITOs for the input system at $A$ and $B$ respectively. However due to orthonormality of ITOs in any linear combination of such diagrams for which $\sum_{\theta}x_{\theta}\Tr_{A'B'}(J[\Phi_{\theta}])=\iden_{AB}$ only the coefficient associated to $\theta=[(0,0)\stackrel{0}{\longrightarrow} (0,0)]$ is non-zero. Since the ITOs are also normalised this results in a fixed value for the associated non-zero coefficient given by $\sqrt{\frac{d_{A}d_{B}}{d_{A'}d_{B'}}}$.
	
\end{proof}
\subsection{Classes of diagrams -- local, injection and relational.}

We need to distinguish several types of irreps for what follows. For every quantum system the identity operator $\I \in\B(\H)$ is always a trivial state-mode ($\U_g(\I)= \I$ for all $g$) that is required for a state to be normalised. We denote with state mode via $a=0$. Any other trivial state mode is then $a' \cong 0$. For the process mode label we need only specify whether it is a trivial irrep or not, and therefore write $\lambda \cong 0$ and $\lambda\not \cong 0$ respectively. For general bipartite scenarios we may decompose the set of all diagrams into three natural types:

\begin{enumerate}
	\item The first class of diagrams are \emph{local diagrams} for which $\lambda\cong 0$ and so involve no symmetry-breaking interactions. These are therefore symmetric under the local symmetry action. 
	\item The second class of diagrams we distinguish are \emph{injection diagrams} for which $\lambda \not \cong 0$, and either $\tilde{a}=0$ or $\tilde{b}=0.$ These describe the transference of asymmetry degrees of freedom from one side to the other. For example, these diagrams are required to take a symmetric state at $A$ to a non-symmetric state via the injection of asymmetry from $B$.
	\item The final class of diagrams we distinguish are \emph{relational diagrams} for which $\lambda\not \cong 0 $, and $\tilde{a}$, $\tilde{b} \ne 0$. These evolve the relational asymmetry degrees of freedom between $A$ and $B$. As such, these diagrams only contribute when both the inputs at $A$ and $B$ carry symmetry-breaking degrees of freedom. In particular, they have no local affects at either $A$ or $B$.
\end{enumerate}
The first class is easy to justify, while the other two classes are best seen by example. In the next section we illustrate these features in the elementary case of processes on $2$ qubits.

\subsection{Illustrative examples -- The set of 2-qubit symmetric quantum processes}

To illustrate this we can consider the group $G=SU(2)$ and analyse the set of 2-qubit symmetric quantum processes (we assume the input and output spaces coincide).
Consider a bipartite system $AB$ consisting of two qubits each of which carry the spin-1/2 irrep of SU(2). It is easy to see that $\rm{Irrep}(A,A') = \rm{Irrep}(B,B') = \{0,1,2\}$ with multiplicities two for $\lambda=0$, three for $\lambda=1$ and one for $\lambda=2$.

\begin{figure}[h!]
	\includegraphics[width=10cm]{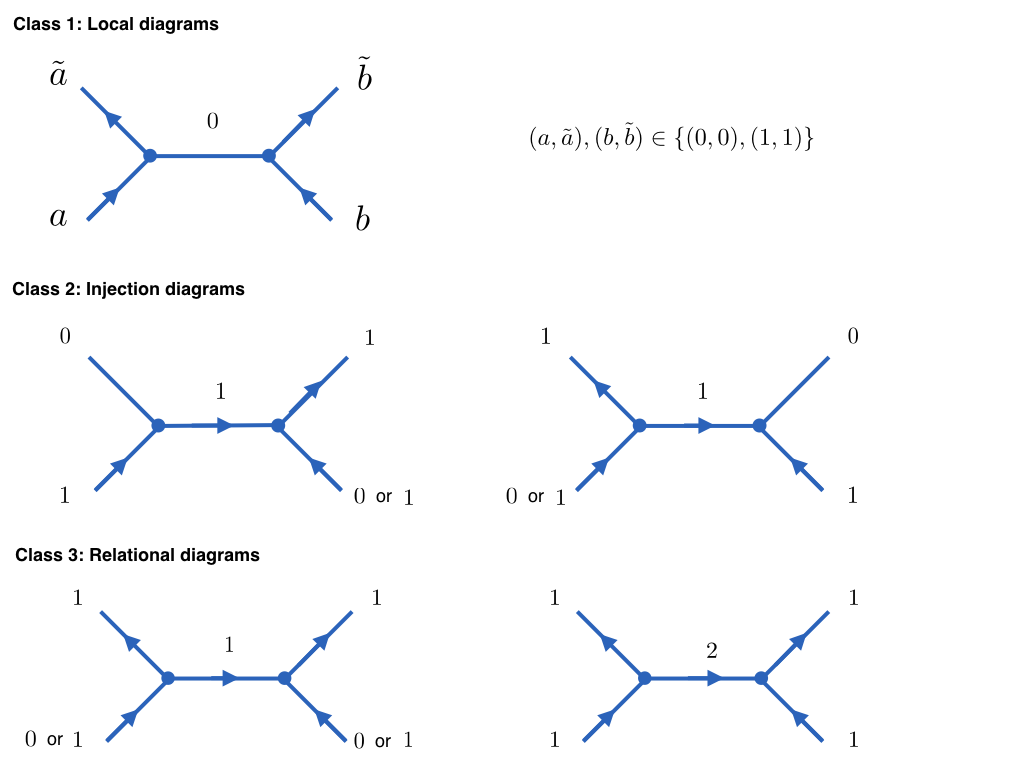}
	\caption{\textbf{The set of  $2$-qubit $SU(2)$-diagrams.} All possible diagrammatic terms allowed for a 2-qubit symmetric quantum process. The space of valid processes is 13-dimensional.}
\end{figure}\label{2qubitdiagrams}

The set of all diagrams from which all 2-qubit symmetric processes are built, are shown in Figure (\ref{2qubitdiagrams}). 

\subsubsection{Local processes on two qubits}
The simplest bipartite symmetric processes are those for which $\lambda \cong 0$ in all process diagrams. In other words, no asymmetry flows between $A$ and $B$. For these processes $\E$ is symmetric on both $A$ and $B$ separately, and form a two-parameter family of processes given by the product of partial depolarising processes $\E = \D_{p_A} \otimes \D_{p_B}$, where $p_A, p_B \in [-\frac{1}{3},1]$ and
\begin{equation}
\D_p (\rho) = p\rho + \frac{1}{4}(1-p) \I.
\end{equation}
These processes involve only two process modes locally, $id(\rho) = \rho$ and $\Phi^0 (\rho) = \sigma_x \rho \sigma_x + \sigma_y \rho \sigma_y + \sigma_z\rho \sigma_z$.

\subsubsection{Asymmetry injection processes on $AB$ and simulation of processes at $A$.}
The next simplest bipartite processes to consider are those globally symmetric quantum processes built from diagrams in the second class mentioned above. We can characterise the general set of processes that made use of asymmetry in the input state at system $B$ and transfer it to the subsystem $A$. These are significant in being the only terms of relevance in any protocol where we wish to simulate or induce a target quantum process at $A$ via asymmetry resources at $B$. Specifically, we would consider using a state $\sigma_B$ at $B$ under a bipartite symmetric process $\V_{AB}$ on $AB$ as
\begin{equation}
\E_A(\rho_A) := \tr_B [ \V(\rho_A \otimes \sigma_B)].
\end{equation}
It can be seen that only diagrams in class-1 and class-2 contribute to $\E_A$, and since class-1 diagrams are purely local, the class-2 are precisely the diagrams that contribute to the injection of asymmetry from $B$ into $A$ or vice versa. 

The general structure of these processes are
\begin{equation}
\E = \E_0 + \E_{in,A} + \E_{in,B},
\end{equation}
where $\E_0$ is built solely from class-1 diagrams, and for any $\rho$ we have $\tr_B[\E_{in, A}(\rho)]=0$, and $\tr_A [\E_{in,B}(\rho)]=0$ and so these components describe the injection into $A$ and $B$ respectively. In this we also include diagrams that maintain the local Bloch vectors (which are class-1 diagrams) since we are interested in inducing local processes.

Each of these two terms form a $3$-parameter family of maps given by
\begin{align}
\E_{in,A}&=x\Phi_{\theta_1}+y\Phi_{\theta_2}+z\Phi_{\theta_3}\\
\E_{in,B}&=x'\Phi_{\theta'_1}+y'\Phi_{\theta'_2}+z'\Phi_{\theta'_3}
\label{1E}
\end{align}
for coefficients $x,y,z,x',y',z'$ chosen such that $\E$ is a valid quantum process.

The diagrams involved are given by
\begin{align}
\theta_1&=[(1,1)\stackrel{0}{\longrightarrow}(0,0)] \nonumber \\
\theta_2&=[(0,1)\stackrel{1}{\longrightarrow}(1,0)]\nonumber\\
\theta_3&=[(1,1)\stackrel{1}{\longrightarrow}(1,0)]
\end{align}
for simulation at $A$, and by
\begin{align}
\theta'_1&=[(0,0)\stackrel{0}{\longrightarrow}(1,1)] \nonumber \\
\theta'_2&=[(1,0)\stackrel{1}{\longrightarrow}(0,1)]\nonumber\\
\theta'_3&=[(1,0)\stackrel{1}{\longrightarrow}(1,1)]
\end{align}
for simulation at $B$.

We can also describe the action of these processes via their action on a general $2$-qubit state. We use the canonical form:
\begin{equation}
\rho_{AB}=\frac{1}{4}\left(\I \otimes \I+\mathbf{a}\cdot\boldsymbol{\sigma}\otimes\iden+\iden\otimes\mathbf{b}\cdot\boldsymbol{\sigma}+\sum_{i,j}T_{ij}\sigma_{i}\otimes\sigma_{j} \right)
\label{qubit}
\end{equation}
for such a state, where local Bloch vectors $\mathbf{a}$ and $\mathbf{b}$ together with the correlation matrix $T_{ij}$ are chosen such that $\rho_{AB}$ is a positive matrix. Given the above parameterization of the general process $\E$, we have that
\begin{equation}
\E(\rho_{AB})=\frac{1}{4}\left(\iden\otimes\iden+\mathbf{\tilde{a}}\cdot\boldsymbol{\sigma}\otimes\iden+\iden\otimes\mathbf{\tilde{b}}\cdot\boldsymbol{\sigma}\right),
\end{equation}
where the local Bloch vectors of the output states are given by
\begin{align}
\mathbf{\tilde{a}}&=-\left(\frac{x\mathbf{a}}{\sqrt{3}}+y\mathbf{b}+\frac{z\mathbf{T}}{\sqrt{2}}\right) \\
\mathbf{\tilde{b}}&=-\left( \frac{x'\mathbf{b}}{\sqrt{3}}+y'\mathbf{a}+\frac{z' \mathbf{T}}{\sqrt{2}}\right).
\end{align}
Here $\mathbf{T}$ the vector with components $\mathbf{T}_{k}:=\sum_{i,j}\epsilon_{kij}T_{ij}$. The geometric significance of this can be seen for the case of initial product states $\rho_{AB}=\rho_{A}\otimes\rho_{B}$ with local Bloch vectors $\mathbf{a}$ and $\mathbf{b}$ then the correlation matrix takes the form of $T_{ij}=a_{i}b_{j}$ and therefore since $\epsilon_{kij}a_{i}b_j=(\mathbf{a}\times\mathbf{b})_{k}$ and so the vector $\mathbf{T}=\mathbf{a}\times\mathbf{b}$ is cross-product between the input Bloch vectors at each site. More generally $\mathbf{T}$ is a vector component that describes the joint asymmetry of $A$ and $B$, in contrast to $\mathbf{a}$ and $\mathbf{b}$, which are purely local terms.

We may now set $x'=y'=z'=0$, and study the set of all processes involved in simulation at $A$. This is a $3$-parameter family in $(x,y,z)$ and so we can plot the allowed region in 3-D. The set of all such quantum processes is given by the convex set bounded by the paraboloid
\begin{equation}
x^2+x\left(\frac{2-6y}{\sqrt{3}}\right)+2y+3y^2+6z^2=1
\end{equation}
and the plane $\sqrt{3}(1+y)+x=0$. 

We can make a change of coordinates for which the quartic boundary (paraboloid) reduces to one of the 17 standard forms. Let
\begin{align}
&X=(1+\sqrt{3}x-3y))/2\nonumber\\
&Y=1-3y\nonumber \\
&Z=3z/\sqrt{2}.
\end{align}

Then the region of parameters $(X,Y,Z)$ is given by the three dimensional convex set bounded by an elliptic paraboloid described by the equation,
\begin{equation}
X^2+Z^2=Y,
\end{equation}
and the plane $2+X-Y=0$. 
In Figure \ref{5} we show this parameter region while highlighting the points corresponding to distinguished extremal processes.
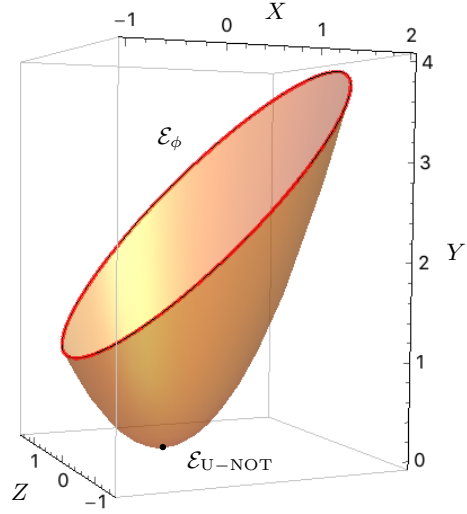
\begin{figure}[h!]
	\begin{tikzpicture}
	\node[anchor=south west,inner sep=0] at (0,0) {\includegraphics[width=5.5cm]{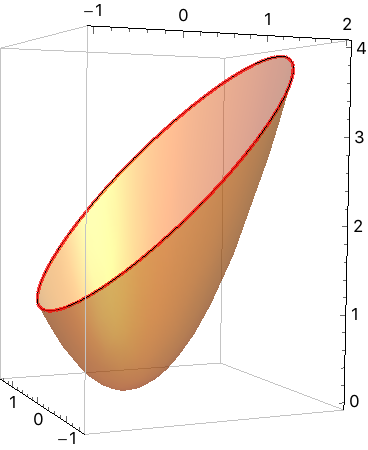}};
	\node at (0, 0.3) {$Z$};
	\node at (5.8,3.5) {$Y$};
	\node at (3.4,6.7) {$X$};
	\filldraw (1.9,0.9) circle (1pt);
	\node at (2.8,0.7)   {$\mathcal{E}_{\rm \tiny U-NOT}$};
	\node at (2, 5) {$\mathcal{E}_{\phi}$};
	\end{tikzpicture}
	\caption{\textbf{The set of induced processes on $A$ via a globally symmetric process on $AB$.} Shown is the allowed parameter region for the rotated coefficients corresponding to each class-2 diagram appearing in a general quantum process that injects asymmetry. The boundary is described by the intersection of a plane and an elliptic paraboloid.}
	\label{5}
\end{figure}
In particular, we find that the vertex of the paraboloid at $(X,Y,Z) =(0,0,0)$ corresponds to the quantum process
\begin{equation}
\E_{\mbox{\tiny U-NOT}}(\rho)=\frac{1}{4} (\I  - \frac{1}{3}\mathbf{b} \cdot \boldsymbol{\sigma}\otimes \I )
\end{equation}
which is the result of discarding the input at $A$, performing an approximate Universal NOT gate on system $B$, and then injecting this into the output system at $A$. This approximate U-NOT turns out to actually be the optimal ``spin inversion'' that is allowed by quantum mechanics \cite{PhysRevA.60.R2626}. 

The intersection of the two boundary regions is an ellipse that we parametrise by a single coordinate $\phi$. In terms of the parameters we have $Z(\phi)=\frac{3}{2}\cos{\phi}$, $Y(\phi)=\frac{3}{2}\sin{\phi}+\frac{5}{2}$ and $X(\phi)=\frac{3}{2}\sin{\phi}+\frac{1}{2}$. And gives rise to the one-parameter family of processes with
\begin{align}
\E_\phi (\rho)& =\frac{1}{4}\left(\iden\otimes\iden+\mathbf{\tilde{a}}\cdot\boldsymbol{\sigma}\otimes\iden\right) \\
\mathbf{\tilde{a}} &=\frac{1}{2}\left( (\mathbf{a} + \mathbf{b}) + (\mathbf{b}- \mathbf{a})\sin \phi  +\mathbf{T} \cos \phi \right).
\end{align}
For the case of product input states with local Bloch vectors $\mathbf{a}$ and $\mathbf{b}$ at $A$ and $B$ respectively, we see that the set of accessible points form an ellipse displaced from the origin by $\frac{1}{2}(\mathbf{a} + \mathbf{b})$ with orientation defined by $\frac{1}{2}(\mathbf{b} - \mathbf{a})$ and $ \frac{1}{2} \mathbf{a} \boldsymbol{\times} \mathbf{b}$. Also note that for $\phi = \frac{\pi}{2}$ we have that the output on $A$ is the input state on $B$. Therefore the line joining $\E_{\mbox{\tiny UNOT}}$ and $\E_{\phi=\frac{\pi}{2}}$ is the set of general depolarization processes on $B$ with output sent to $A$.

We therefore find that the set of all processes using a qubit $B$ to induce a non-symmetric process on $A$ is to a good approximation given by the convex hull of the optimal U-NOT gate and the set of processes $\E_\phi$ with $0 \le \phi \le 2 \pi$.

\subsubsection{Purely relational processes} 
There are 5 diagrams in total in the class-3 resulting in the most general quantum processes that involve these type of diagrams taking the form of:
\begin{equation}
\E=\Phi_{0}+x_4\Phi_{\theta_4}+x_5\Phi_{\theta_5}+x_6\Phi_{\theta_6}+x_{7}\Phi_{\theta_7}+x_{8}\Phi_{\theta_8}
\end{equation}
where
\begin{align}
\theta_{4}&=[(0,1)\stackrel{1}{\longrightarrow}(0,1)]\nonumber\\
\theta_5&=[(1,1)\stackrel{1}{\longrightarrow}(1,1)]\nonumber \\
\theta_6&=[(1,1)\stackrel{1}{\longrightarrow}(0,1)] \nonumber\\
\theta_7&=[(0,1)\stackrel{1}{\longrightarrow}(1,1)]\nonumber \\
\theta_8&=[(1,1)\stackrel{2}{\longrightarrow}(1,1)].
\end{align}

Any such quantum process which contains only class-3 diagrams has the property that the output state $\E(\rho)$ always has maximally mixed marginals for all initial states. More precisely, for any $\rho\in\B(\h_A\otimes\h_B)$ we have that $\Tr_{A}(\E(\rho))=\Tr_{B}(\E(\rho))=\frac{1}{2}\I$. This implies that we must have
\begin{equation}
\E(\rho_{AB})=\frac{1}{4}\left(\I\otimes \I+\sum_{i,j}R_{ij}\sigma_{i}\otimes\sigma_{j}\right)
\end{equation}
for some correlation matrix $R_{ij}$ that depends on both $\rho_{AB}$ and the particular relational process. We can make more precise the contribution of each diagram in class-3 to the tensor $R_{ij}$.
\begin{equation}
\begin{split}
&\theta_4: R^{\theta_4}=-\frac{\iden}{4}\\
&\theta_5: R^{\theta_5}=\frac{1}{8}(T^{T}+\Tr(T)\iden)\\
&\theta_6: R_{ij}^{\theta_6}=i\frac{\sqrt{2}}{2}(\epsilon_{ijk}a_{k})\\
&\theta_7: R_{ij}^{\theta_7}=i\frac{\sqrt{2}}{2}(-\epsilon_{ijk}b_{k} )\\
&\theta_8: R^{\theta_8}=\frac{1}{8}(T^{T}-\frac{2}{3}T+\Tr(T)\iden),
\end{split}
\end{equation}
where we have denoted by $R^{\theta_m}$ to be the correlation matrix of the output under applying the superoperator $\Phi_{\theta_{m}}$ to $\rho_{AB}$. In other words, $\Phi_{\theta_{m}}(\rho_{AB})=\sum_{i,j}R^{\theta_{m}}_{ij}\sigma_{i}\otimes\sigma_{j}$. 

To explore such processes, we restrict to those that are invariant under swapping $A$ and $B$. For this, the most general form is given by
\begin{equation}
\E=\Phi_{0}+x\Phi_{\theta_4}+y\Phi_{\theta_5}+z\Phi_{\theta_8}
\end{equation}
for real parameters $x,y,z$. Imposing that $\E$ is a valid quantum process implies the 3-d convex region of parameters with boundary surfaces given by the following quartics:
\begin{align*}
&\left(9x+3y+5z-3\right)^2=(5z+21y-12)^2-108(1-2y)^2\\
&\left(6y+3x\right)^2=6x+3+20z\\
&y^2=\left(\frac{1-x}{2}\right)^2 \  \text{for}  \   0<x\leq 1\\
&y^2=\frac{ (x+5/3)^2}{4}-\frac{4}{9}  \  \text{for}  \  -1/3\leq x\leq0
\end{align*}

In other words the boundary is the intersection of an elliptic cone, a parabolic cylinder, two intersecting planes and a hyperbolic cylinder respectively.
\begin{figure}[h!]
	\begin{tikzpicture}
	\node[anchor=south west,inner sep=0] at (0,0) {\includegraphics[width=7cm]{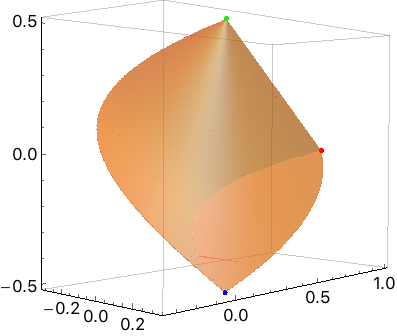}};
	\node at (4,6) {$\mathcal{E}_{2}$};
	\node at (3.5,0.9) {$\mathcal{E}_{1}$};
	\node at (6.3,3.5) {$\mathcal{E}_{\rm{singlet}}$};
	\node at (5, 0.3) {$x$};
	\node at (-0.1,3.2) {$y$};
	\node at (1.3, 0.1) {$z$};
	\end{tikzpicture}
	\caption{\textbf{$SU(2)$-symmetric, relational processes on $2$-qubits.} The allowed parameter region for the relational processes including only class-3 diagrams. The right-most red point is the singlet process $\E_{singlet}$, blue point is $\E_1$ and green point is $\E_2$ }
	\label{6}
\end{figure}
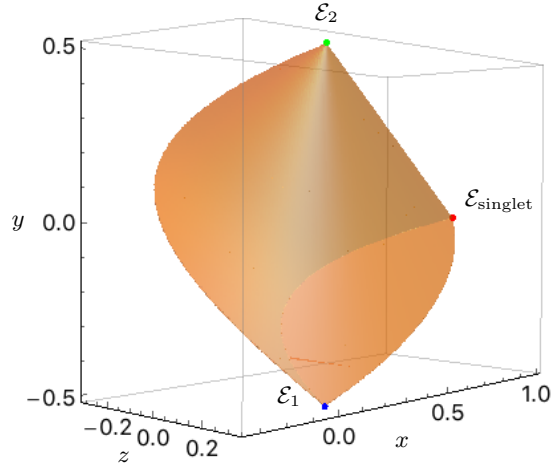

There are distinguished simple processes that correspond to points on the surface boundary. For instance the point $(1,0,0)$ is the unique intersection of the two intersecting planes, the parabolic cylinder, the elliptic cone. It corresponds to the singlet preparation process $\E_{singlet}(\rho)=\ket{\psi^{-}}\bra{\psi^{-}}$ for any $2$-qubit state $\rho$. 

In addition there are two points that lie at the intersection of the elliptic cone the parabolic cylinder and the hyperbolic cylinder and each of them are in one of the two planes and correspond to imposing that $x=0$.  The characterising feature of the processes for which $x=0$ is that they do not displace the maximally mixed state and are therefore unital processes. Corresponding to the two distinguished points we get the following quantum processes
\begin{align}
\E_{1}&=\Phi_{0}-\frac{1}{2}\Phi_{\theta_5}+\frac{3}{10}\Phi_{\theta_8}\\
\E_{2}&=\Phi_{0}+\frac{1}{2}\Phi_{\theta_5}+\frac{3}{10}\Phi_{\theta_8}
\end{align}
where $\E_1$ corresponds to the point lying in the plane $y=\frac{1}{2}(x-1)$ and all the other three surfaces and $\E_2$ corresponds to the point lying in the plane $y=\frac{1}{2}(x+1)$ and the non-degenerate three surfaces.

Both are only sensitive to the $T_{ij}$ components of the input state $\rho$. Without loss of generality, we look at how they act on states of the form: $\rho_{AB}=\frac{1}{4}\left(\I+\sum_{i,j}T_{ij}\sigma_{i}\otimes\sigma_{j}\right)$. Moreover, up to local unitaries any such state can be brought to a canonical form $\rho_{AB}=\frac{1}{4}\left(\I+\sum_{i}t_{i}\sigma_{i}\otimes\sigma_{i}\right)$ specified by a single vector $(t_1,t_2,t_3)$. The range of these parameters lie in a tetrahedron whose vertices correspond to the Bell states.

The action of these processes on the Bell states $\phi^\pm, \psi^\pm$ is given by:
\begin{align}
\E_1 (\phi^{\pm}) &=\frac{3}{20}\iden+\frac{3}{10}\phi^{\pm}+\frac{1}{10}\psi^{-}  \nonumber \\
\E_1 (\psi^+) &=\frac{3}{20}\iden+\frac{3}{10}\psi^{+}+\frac{1}{10}\psi^{-} \nonumber \\
\E_1 (\psi^-) &= \frac{1}{4}\iden
\end{align}
and
\begin{align}
\E_2 (\phi^{\pm}) &=\frac{2}{5}\iden-\frac{1}{5}\phi^{\pm} -\frac{2}{5}\psi^{-} \nonumber \\
\E_2 (\psi^+) &=\frac{2}{5}\iden-\frac{1}{5}\psi^{+} -\frac{2}{5}\psi^{-} \nonumber \\
\E_2 (\psi^-) &= \psi^{-}
\end{align}

Since the Bell states are extremal the convex hull of these images give the action in the more general case. The image of the tetrahedron of state is graphically displayed in Figure \ref{tetrahedron}.
\begin{figure}[h!]
	\begin{tikzpicture}
	\node[anchor=south west,inner sep=0] at (0,0) {\includegraphics[width=7cm]{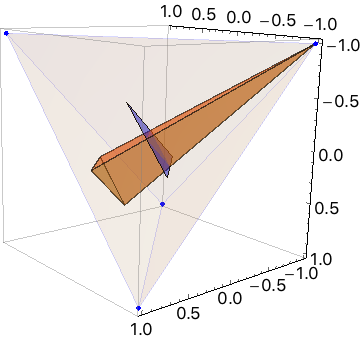}};
	\node at (6.0 ,5.4) {$\psi^{-}$};
	\node at (7.1, 3.5) {$t_1$};
	\node at (5, 0.5) {$t_2$};
	\node at (5, 6.8) {$t_3$};
	\node at (0,6.1) {$\phi^{-}$};  
	\node at ( 3,0.9) {$\psi^{+}$};
	\node at (3.6,2.7) {$\phi^{+}$};  
	\end{tikzpicture}
	\caption{\textbf{T-state transformations}. The set of $2$-qubit states with maximally mixed marginals modulo local choice of bases (or ``T-states'') have a tetrahedral state space with the four Bell states at the extremal points.  Under the extremal processes $\E_1$ and $\E_2$ the set of T-states is mapped into the (blue) triangle and inner (brown) tetrahedron respectively. }
	\label{tetrahedron}
\end{figure}
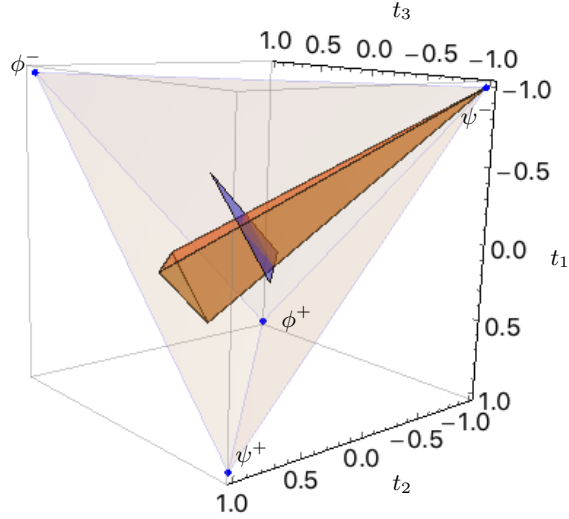
The preceding analysis can be used on more general bipartite quantum systems, where it allows a compact book-keeping for the analysis of quantum processes and simplifies the analysis. 

\subsubsection{Symmetric unitary processes on two qubits}

For the same of completeness, we briefly describe one more kind of symmetric process -- the $SU(2)$-symmetric unitaries on two qubits. Since $V = \exp [i tH]$ for some Hamiltonian $H$, the problem reduces to computing the allowed Hamiltonians. The symmetry of $V$ implies that $\U_g(H) =H$ for all $g\in G$ and so $H$ is an invariant operator under the group action. The space of invariant hermitian observables is spanned by $\I$ and $\sigma_x \otimes \sigma_x + \sigma_y \otimes \sigma_y + \sigma_z \otimes \sigma_z$, and therefore the symmetric unitaries on the system is a two-parameter family given by
\begin{equation}
\exp [ i (s \I + t (\sigma_x \otimes \sigma_x + \sigma_y \otimes \sigma_y + \sigma_z \otimes \sigma_z))].
\end{equation}
The first term is a phase term and so $V(t) = e^{it (\sigma_x \otimes \sigma_x + \sigma_y \otimes \sigma_y + \sigma_z \otimes \sigma_z)}$ is the only non-trivial unitary interaction present.

The quantum process $\E(\rho) = V(t) \rho V(t)^\dagger$  has a mode decomposition as shown in Figure \ref{2unitary}.
\begin{figure}[h!]
	\includegraphics[width=8cm]{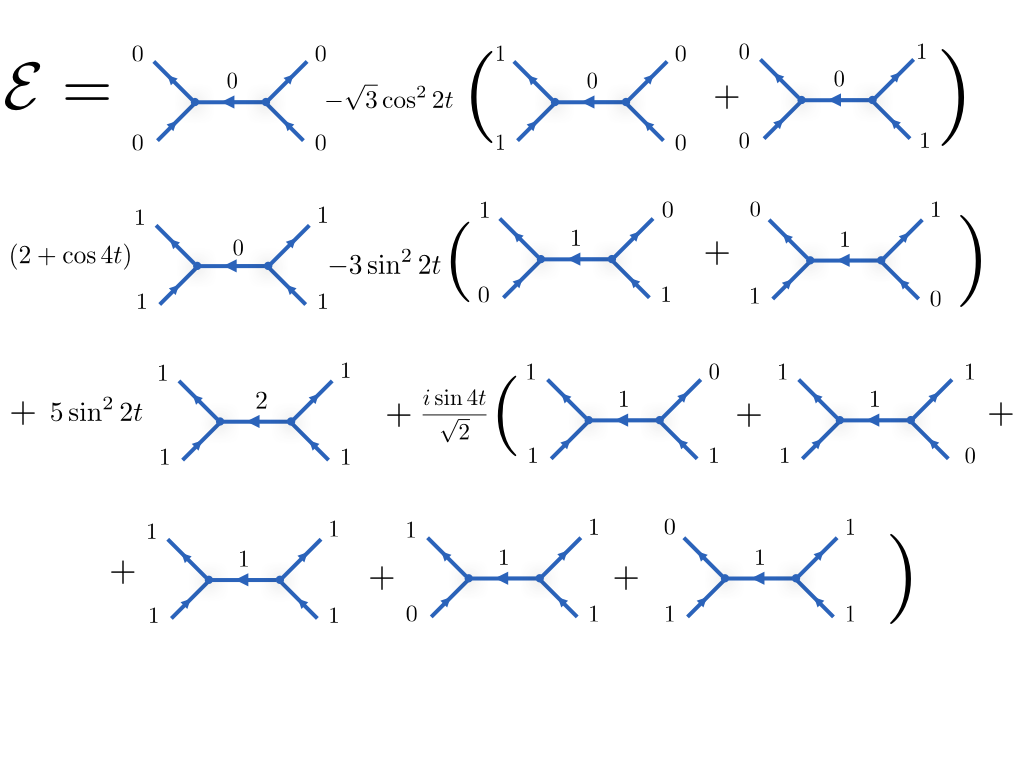}
	\caption{The decomposition of the symmetric unitary $V(t) = e^{i t (\sigma_x \otimes \sigma_x + \sigma_y \otimes \sigma_y + \sigma_z \otimes \sigma_z)}$ on two qubits.}
	\label{2unitary}
\end{figure}
Note that because $V$ is symmetric under swapping $A$ and $B$, we have this additional symmetry reflected in the diagram contributions. The expansion in terms of process modes shows that the global unitary has a non-trivial structure under the symmetry action, which is perhaps not surprising since the process is perfectly reversible and keeps all symmetry properties in the state constant.

\section{Gauging processes}
\label{S-gauge}
\subsection{From global to local symmetries}
In the following we illustrate the gauging procedure for 2-symmetric quantum processes by looking at the symmetric bipartite superoperator associated with diagram $\theta$ and link $l$ connecting local system $A_{x}$ with $A_{y}$. Our main Theorem \ref{main} together with linearity of the gauging map ensures that the same procedure works for any 2-symmetric quantum process. More precisely we denote the gauging map for superoperators by:
\begin{equation}
\hspace{-0.5cm}Gauge:\T(A_x A_y, A_x'A_{y}')\rightarrow \T(A_{x},A_{x}')\otimes \T(\h_{l},\h_{l})\otimes\mathcal{S}(A_{y},A_{y}')\nonumber
\end{equation}
where as in the main text $\h_{l}$ is the quantum reference frame corresponding to link $l$ and encodes the group elements. The action of the gauging map is such that it promotes a global symmetry to a local symmetry. Specifically any $\Phi_{\theta}\in \T(A_{x}A_{y},A_{x}'A_{y}')$ which is symmetric under the global representation (i.e  $\uu_{g}[\Phi_{\theta}]=\Phi_{\theta}$ for all $g\in G$) will be mapped to $Gauge(\Phi_{\theta})$ which is symmetric under the local representation. This means that $\uu_{g_{x}}\otimes\uu_{(g_{x},g_{y})}\otimes\uu_{g_{y}}[Gauge(\Phi_{\theta})]=Gauge(\Phi_{\theta})$ for all $g_{x},g_{y}\in G$. 

We provide an explicit gauging map that satisfies these requirements. First we explicitly add in the background degrees of freedom that are encoded in the quantum reference frame initialised with a gauge coupling that is symmetric under the global symmetry. The resulting process on the local systems and the link reference frame will be invariant under the global representation. Then we promote the global symmetry to a local symmetry by averaging over the local independent degrees of freedom and therefore removing all the relative alignments. 

Recall that the globally symmetric superoperator given by diagram $\theta$ decomposes into local ITS as: $\Phi_{(l,\theta)}=\sum_{j}\Phi^{\lambda}_{x,j}\otimes \Phi^{\lambda^{*}}_{y,j}$. We first map $\Phi_{(l,\theta)}\longrightarrow \Phi_{(l,\theta)}'=\sum_{j}\Phi^{\lambda}_{x,j}\otimes\mathcal{A}^{j,j}_{(l,\lambda)}\otimes\Phi^{\lambda^{*}}_{y,j}$ where $\mathcal{A}^{j,j}_{(l,\lambda)}$ are process gauge couplings such that $\uu_{(g,g)}[\mathcal{A}^{j,j}_{(l,\lambda)}]=\mathcal{A}^{j,j}_{(l,\lambda)}$ for all $g\in G$. Clearly $\Phi_{(l,\theta)}'$ is invariant under the global representation: $\uu_{g}\otimes\uu_{(g,g)}\otimes\uu_{g}[\Phi_{\theta}']=\Phi_{\theta}'$. We have that:
\begin{align}
&Gauge(\Phi_{(l,\theta)}):=\int_{G\times G}\uu_{g_{x}}\otimes\uu_{(g_{x},g_{y})}\otimes\uu_{g_{y}}[\Phi_{(l,\theta)}']d\,g_{x}d\,g_{y}\nonumber \\
&=\sum_{j}\int_{G\times G} \uu_{g_{x}}[\Phi^{\lambda}_{x,j}]\otimes\uu_{(g_x,g_y)}[\mathcal{A}^{j,j}_{(l,\lambda)}]\otimes \uu_{g_{y}}[\Phi^{\lambda^{*}}_{y,j}]d\,g_xd\,g_y\nonumber
\end{align}
Because the Haar measure for a compact group is both left and right invariant it is straightforward to check that the operation constructed above is symmetric under the local de-synchronised group action. Therefore we have that $\uu_{g_{x}}\otimes\uu_{(g_{x},g_{y})}\otimes\uu_{g_{y}}[Gauge(\Phi_{(l,\theta)})]=Gauge(\Phi_{(l,\theta)})$ for all $g_{x},g_{y}\in G$. However we know how the local ITS and the process gauge couplings transform under group actions on their respective systems and therefore can obtain a compact form for the gauging map. In particular we have that:
\begin{align}
&\uu_{g_x}[\Phi^{\lambda}_{x,j}]=\sum_{m} v^{\lambda}_{jm}(g_{x})\Phi^{\lambda}_{x,m}\nonumber \\
&\uu_{(g_x,g_y)}[\mathcal{A}^{j,j}_{(l,\lambda)}]=\sum_{m',n'}v^{\lambda}_{m'j}(g_{x}^{-1})v_{j,n'}^{\lambda}(g_{y})\mathcal{A}^{m'n'}_{(l,\lambda)}\nonumber \\
&\uu_{g_y}[\Phi^{\lambda^{*}}_{y,j}]=\sum_{n} (v^{\lambda}_{jn}(g_{y}))^{*}\Phi^{\lambda}_{y,n}.\nonumber 
\end{align}
We can combine all of these to obtain the form of the gauging map. The matrix coefficients satisfy $v^{\lambda}_{jm}(g_{x})v^{\lambda}_{m'j}(g_{x}^{-1})=\delta_{mm'}$ and also from unitarity we get that $v^{\lambda}_{jn'}(g_{y})(v^{\lambda}_{jn}(g_{y}))^{*}=v^{\lambda}_{jn'}(g_{y})(v^{\lambda}(g_{y})\hc)_{nj}=v^{\lambda}_{jn'}(g_{y})(v^{\lambda}(g_{y}^{-1}))_{nj}=\delta_{nn'}$. Therefore the action of the gauging map is given by:
\begin{equation}
Gauge(\Phi_{(l,\theta)})=\sum_{m,n} \Phi^{\lambda}_{x,m}\otimes \mathcal{A}^{mn}_{(l,\lambda)}\otimes\Phi^{\lambda^{*}}_{y,n}.
\end{equation}
\subsubsection{Example: U(1) symmetry}
We consider the situation when the symmetry group is $U(1)$ and suppose that the local ITS on system $A_{x}$ and $A_{y}$ are given by $\Phi^{\lambda}_{x}$ and $\Phi^{\lambda}_{y}$. Under the group action for each element $\phi\in U(1)$ they transform according to $\uu_{\phi}[\Phi^{\lambda}_{x}]=e^{i\lambda\phi}\Phi^{\lambda}_{x}$. For abelian groups all irreducible representations are one-dimensional and therefore each $\theta$-diagram corresponding to $\lambda$-irrep symmetry breaking carrier takes the form of: $\Phi_{(l,\theta)}=\Phi^{\lambda}_{x}\otimes \Phi^{\lambda^{*}}_{y}$. This is invariant under the global group action. We add in the quantum reference frame that encodes group elements -- in this case angles. Consider $\h_{l}$ to be the Hilbert space of an infinite ladder system with equally spaced energy eigenstates labelled by $\{\ket{m}\}_{m\in\mathbb{Z}}$. The coherent states on the circle are defined in $\h_{l}$ for each element $\phi\in U(1)$ as:
\begin{equation}
\ket{\phi}=\sum_{m}e^{-im\phi}\ket{m}.
\end{equation}
They form an orthonormal set of eigenvectors for the self-adjoint operator $\hat{\phi}=\int \phi\ket{\phi}\bra{\phi}d\,\phi$ the canonical conjugate of angular momentum in the z-direction. Therefore these states perfectly encode all group elements of $U(1)$. The group action on these states will be given by:
\begin{equation}
\U_{(g_x,g_y)}(\ket{\phi}\bra{\phi})=\ket{g_x+\phi-g_{y}}\bra{g_{x}+\phi-g_{y}}
\end{equation}
for all $g_x,g_y\in U(1)$. In particular $\ket{0}\bra{0}$ is invariant under the global representation i.e whenever $g_{x}=g_{y}$. Therefore adding in the globally symmetric background degrees of freedom:
$ \Phi_{(l,\theta)}\longrightarrow \Phi^{\lambda}_{x}\otimes \Phi^{\lambda^{*}}_{y}\otimes \ket{0}\bra{0}$ and therefore the gauging procedure maps:
\begin{equation}
\hspace{-0.5cm}Gauge(\Phi_{(l,\theta)})=\int e^{i\lambda(g_{x}-g_{y})}\Phi^{\lambda}_{x}\otimes\Phi^{\lambda^{*}}_{y}\otimes\ket{g_{x}-g_{y}}\bra{g_{x}-g_{y}}d\,g_{x}d\,g_{y}.
\end{equation}
However invariance of Haar measure means that:
\begin{equation}
\int e^{i\lambda(g_x-g_{y})}\ket{g_{x}-g_{y}}\bra{g_{x}-g_{y}}d\,g_xd\,g_y=\int e^{i\lambda\phi}\ket{\phi}\bra{\phi}d\,\phi.
\end{equation}
Using the defining decomposition of $\ket{\phi}$ in terms of the energy eigenstates of the ladder then we obtain the gauged operation that has a local symmetry:
\begin{equation}
\hspace{-0.5cm}\Phi_{(l,\theta)}\longrightarrow Gauge(\Phi_{(l,\theta)})=\Phi^{\lambda}_{x}\otimes\Phi^{\lambda^{*}}_{y}\otimes\sum_{m\in \mathbb{Z}}\ket{m}\bra{m+\lambda}.
\end{equation}
Remark: In the above we assume that the gauging processes are given by states on the reference frame system $\h_{l}$. This is a particular simplified scenario to illustrate the gauging procedure.
\subsection{Fixing a gauge: connections with pre and post selection with a group element}
\label{gaugefixing}
In here we demonstrate a particular way in which the gauge fixing for general quantum processes is achieved via pre and post selection with group elements. In doing so we also underline the physical interpretation of the polar decomposition and the role of the process orbit in providing the relative alignment between subsystem and environment.

Suppose that the quantum process $\tilde{\E}\in \T(\h_{A_{x}}\otimes\h_{A_{y}}\otimes\h_{l})$ that acts on systems $A_{x}$, $A_{y}$ and the reference frame $\h_{l}$ situated on the link between $x$ and $y$ is invariant under the local group action. This means that $\uu_{(g_{x},g_{y})}[\tilde{\E}]=\tilde{\E}$ for all group elements $g_{x}, g_{y}\in G$. Suppose that we post select with group element $h_2\in G$ and pre select with $h_1\in G$ then the resulting operation will be given by:
\begin{align}
\tilde{\E}_{h_1,h_2}(\tau):=( id \otimes \Pi_{h_2}) \circ \tilde{\E} \circ (id \otimes \Pi_{h_1}),
\end{align}
where $\Pi_h (\sigma) = |h\>\<h| \sigma |h\>\<h|$, is the projection onto the pure state $|h\>$. In general, the process $\tilde{\E}_{h_1,h_2}$ will not remain locally invariant since the measurements with group elements will break that symmetry. One can check that $\tilde{\E}_{h_1,h_2} $ now transforms under the local group action according to
\begin{align}
\U_{(g_{x},g_{y})}(\tilde{\E}_{h_1,h_2})=\tilde{\E}_{g_{x}h_{1}g_{y}^{-1},g_x h_2 g_{y}^{-1}}.
\label{localgroupT}
\end{align}
This means that the process resulting after pre and post-selection with a group elements $h_1$ and $h_2$ is transformed under the de-synchronised local group action with elements $(g_x,g_y)$ into the process corresponding to pre and post selection with group elements $g_x h_1 g_{y}^{-1}$ and $g_x h_2 g_y^{-1}$ respectively. We show this result by checking directly:

	\begin{align}
	\uu_{(g_{x},g_{y})}[\tilde{\E}_{h_1,h_2}]&=\mathcal{U}_{(g_x,g_y)}(\ket{h_2}\bra{h_2})\,\U_{g_x}\otimes\U_{g_{y}}\left(\bra{h_2}\tilde{\E}(\ket{h_1}\bra{h_1}\bra{h_1}\U_{(g_x,g_y)}\hc(\tau)\ket{h_1})\ket{h_2}\right)\\ \nonumber
	&=\ket{g_{x}h_2g_{y}^{-1}}\bra{g_{x}h_2g_{y}^{-1}}\U_{g_{x}}\otimes \U_{g_{y}}\left(\bra{h_2}\tilde{\E}\left(\ket{h_1}\bra{h_1} \U_{g_{x}}\hc\otimes \U_{g_{y}}\hc(\bra{g_{x}h_1g_{y}^{-1}}\tau\ket{g_{x}h_1g_{y}^{-1}})\right)\ket{h_2} \right) \\ \nonumber
	&=\ket{g_{x}h_2g_{y}^{-1}}\bra{g_{x}h_2g_{y}^{-1}}\U_{g_{x}}\otimes \U_{g_{y}}\left(\bra{h_2}\tilde{\E}\left(\U_{(g_x, g_{y})}\hc\left(\ket{g_{x}h_1g_{y}^{-1}}\bra{g_{x}h_1g_{y}^{-1}}\otimes\bra{g_{x}h_1g_{y}^{-1}}\tau\ket{g_{x}h_1g_{y}^{-1}} \right)\right)
	\ket{h_2}\right)\\ \nonumber
	&=\ket{g_{x}h_2g_{y}^{-1}}\bra{g_{x}h_2g_{y}^{-1}}\otimes \left(\bra{g_x h_2g_{y}^{-1}}\uu_{(g_{x},g_{y})}[\tilde{\E}]\left(\ket{g_{x}h_1g_{y}^{-1}}\bra{g_{x}h_1g_{y}^{-1}}\otimes\bra{g_{x}h_1g_{y}^{-1}}\tau\ket{g_{x}h_1g_{y}^{-1}} \right)\ket{g_{x}h_2g_{y}^{-1}}\right)\\ \nonumber
	&=\ket{g_{x}h_2g_{y}^{-1}}\bra{g_{x}h_2g_{y}^{-1}}\otimes \left(\bra{g_x h_2g_{y}^{-1}}\tilde{\E}\left(\ket{g_{x}h_1g_{y}^{-1}}\bra{g_{x}h_1g_{y}^{-1}}\otimes\bra{g_{x}h_1g_{y}^{-1}}\tau\ket{g_{x}h_1g_{y}^{-1}} \right)\ket{g_{x}h_2g_{y}^{-1}}\right)\\ \nonumber
	&=\tilde{\E}_{g_{x}h_1g_{y}^{-1},g_{x}h_2g_{y}^{-1}}(\tau).
	\end{align}

In the previous calculation we have only used that $\tilde{\E}$ is invariant under the local group action $\U_{(g_{x},g_{y})}$.

To establish how much the local symmetry has been broken by the pre and post selection we look at the set of all group elements $(g_{x},g_{y})$ under which $\tilde{\E}_{h_1,h_2}$ remains an invariant process. First note that the the reference frame perfectly encodes group elements such that the pure states $\ket{g}$ are orthonormal (and hence perfectly distinguishable). Therefore using the above result we have that $\uu_{(g_x,g_y)}[\tilde{\E}_{h_1,h_2}]=\tilde{\E}_{h_1,h_2}$ holds if and only if $g_x h_1 g_{y}^{-1}=h_1$ and $g_x h_2 g_y^{-1}=h_2$. Equivalently $\uu_{(h_1g_{y} h_1^{-1},g_{y}})[\tilde{\E}_{h_1,h_2}]=\tilde{\E}_{h_1,h_2}$ where $g_{y} h_2 h_1^{-1}=h_2 h_1^{-1} g_{y}$. 

In particular this is clearly satisfied whenever $h_1=h_2=h$ for some $h\in G$ . In this case the pre and post selection with group element $h$ breaks the local de-synchronised invariance resulting in a process that is invariant under the global action $\uu_{g}':=\uu_{(hgh^{-1},g)}=\uu_{(h,e)}\circ\uu_{(g,g)}\circ\uu_{(h^{-1},e)}$. Therefore we can view $\uu_{(h^{-1},e)}$ as a local change of basis that aligns system $A_x$ and $A_y$. Perfect alignment means that the process on the bipartite system is globally symmetric i.e invariant under the global representation. In other words the reference frame on the link encodes the group element $h$ required to align the two systems.

\section{Irreversibility in symmetry-breaking degrees of freedom}
\label{S-Irreversibility}
A symmetry principle in a quantum system need not correspond to a conservation law \cite{Noether}. In the case of symmetric unitary dynamics we do have that conservation of charges (corresponding to hermitian observables) hold, and that any symmetry-breaking degrees of freedom in quantum states (which may include a property not described by a hermitian observable) is also conserved.

However, more generally there is a disconnect between symmetry principles and conservation laws \cite{Noether}. For general symmetric quantum processes the expectation values of the generators of the symmetry can both increase and decrease, and a proper account must be supplemented with information-theoretic measures. In such cases quantum incompatibility \cite{heinosaari2016invitation} is expected to give rise to irreversibility in the symmetry-breaking degrees of freedom of a quantum system. For example, a quantum system that acts as a clock functions to break time-translation symmetry. However its use in say quantum thermodynamics may result in a back-action that distorts its subsequent ability to function as a clock. 

One might generally expect globally symmetric quantum processes $\rho_A\otimes \sigma_B \mapsto \E_{AB} (\rho_A \otimes \sigma_B) $ such that $\sigma_B  \mapsto \sigma'_B=\E_B(\sigma_B):= \tr_A  \left [ \E_{AB} (\rho_A \otimes \sigma_B) \right ]$, such that the state $\sigma'_B$ breaks the symmetry in a much weaker form than the original state $\sigma_B$ and is therefore less useful as a result. This constitutes an irreversibility under the symmetry constraint, however it could arise due to the particular interactions used -- might it be possible to use the state more wisely and not suffer such irreversibility?

There is a range of notions related to reversibility and irreversibility. In the simplest case an isolated symmetric, unitary evolution preserves all symmetry-breaking properties and conserves charges. There is also the notion of a \emph{catalytic} use of a symmetry-breaking resource $\sigma_B$ in which a quantum process is performed $\rho_A \otimes \sigma_B \mapsto \E_{AB} (\rho_A \otimes \sigma_B) = \E_A (\rho_A) \otimes \sigma_B$. In general theories of quantum resources (for example entanglement theory) such use of catalysts can have non-trivial effects and enlarge the set of accessible transformations on $A$.

However the above notions do not exhaust the possibilities. In \cite{aberg} a phenomenon called catalytic coherence was discovered by Johan {\AA}berg in which quantum coherence resources can be re-used in such a way that the state of the resource constantly changes, however its ability as a resource for inducing processes on multiple independent systems remains unchanged. The core setting involves a $U(1)$ symmetry constraint associated to a `number' operator $N$, with integer eigenvalues $n\in \mathbb{Z}$, on quantum systems \cite{LK}. A coherent state $\sigma_B$ on a `ladder' system, with Hilbert space $\H_{\rm ladder}$ spanned by the eigenstates $\{|n\>\}_{n \in \mathbb{Z}}$ of $N$, is present and used to induce some target map $\E$ on a system $A$ through the interaction
\begin{equation}
\rho_A \rightarrow \tilde{\E}(\rho_A) := \tr_{B}V( \rho_A \otimes \sigma_B) V^\dagger.
\end{equation}
Here $\tilde{\E}$ is an approximation of some target map $\E$ on the primary system. $V$ is a bipartite unitary that respects the global $U(1)$ symmetry constraint and takes the general form
\begin{equation}
V(U)=\sum U_{mn}\ket{\phi_{m}}\bra{\phi_{n}}\otimes \Delta^{n-m}
\end{equation}
where $\{\ket{\phi_{m}}\}_{m=1}^{\rm{dim}(A)}$ forms an orthonormal basis for system $A$ such that it transforms under the U(1) action as $U(\theta)\ket{\phi_{m}}=e^{im\phi}\ket{\phi_{m}}$, the operators $\Delta^{n-m}$ are displacement operators on the ladder system $\Delta^{n}=\sum_{j\in\mathbb{Z}} \ket{j+n}\bra{j}$ and $U_{mn}$ denotes the matrix entries of some arbitrary unitary with dimension $\rm{dim}(A)$. More generally, we will refer to any protocol that implements $\tilde{\E}$ of the form 
\begin{equation}
\tilde{\E}(\rho) =  \sum_{n,i}  \tr(\Delta^n \sigma ) K_{n,i} \rho K_{n,i}^\dagger
\end{equation}
for operators $\{K_{n,i}\}$ on $A$,  as simply \emph{a catalytic coherence protocol}, without any further qualifications.

The state on $B$ evolves non-trivially under the above $V$ as $\sigma_B \rightarrow \sigma'_B$, however if  one reuses $B$ on another quantum system under precisely the same protocol $V$ it was found that its ability to lift the symmetry constraint is undiminished (see \cite{aberg} for more details). In what follows we shall use the term `repeatability' to cover the above three distinct concepts, and which will be defined in the next subsection.

With these subtleties in mind we now study, in general terms, the use of symmetry-breaking resources and when they may be repeatedly used without degrading. In doing so we  make use of the process mode framework, and demonstrate its utility for the analysis of such questions.

\subsection{The use of symmetry-breaking resources and local simulation of quantum processes}
\label{S-protocol}
We assume the simulation of a quantum process $\E:\B(\H_A) \rightarrow \B(\H_{A'})$ locally on $A$ using a state $\sigma_B$ on the environment $B$, takes the form 
\begin{equation}
\E (\rho) = \tr_{B'}\V_{AB}(\rho_A \otimes \sigma_B),
\end{equation}
where $\V:\B(\H_A \otimes \H_B) \rightarrow \B(\H_{A'} \otimes \H_{B'})$ is symmetric under the group action of $G$, and which for simplicity can be assumed to be an isometry. We are interested in the use of some resource $\sigma_B$ to induce $\E$, independent of the state $\rho$. 

More generally we might wish to simulate a set of local processes $\{\E_k\}$ on the system $A$, which for simplicity we assume is a discrete set labelled by $k$. The general task is to devise a protocol that tries to achieve any one of these target maps when presented with an arbitrary quantum system $B$ that is prepared in an unknown state $\sigma_B$. Abstractly, given $(\E_k, B)$, a protocol must specify a symmetric process $\V_k:\B(\H_A \otimes \H_B) \rightarrow \B(\H_{A'} \otimes \H_{B'})$ such that the approximate process induced using $\sigma_B$, denoted $\tilde{\E}_k = \P(k,B, \sigma_B)$, is given by 
\begin{equation}
\tilde{\E}_k(\rho) =\tr_{B'}\V_k(\rho_A \otimes \sigma_{B'}).
\end{equation}
Generally, the performance of the protocol $\P$ is then determined by how close $\tilde{\E}_k$ is to the target process $\E_k$ (using e.g. the diamond norm $||\tilde{\E}_k - \E_k||_\diamond$) for given $k$ and $\sigma_B$. It is also natural to assume a general protocol has a perfect classical limit, in the sense that as the environment $B$ becomes sufficiently large, and we are provided a state $\sigma_B$ that encodes group elements (asymptotically) perfectly, then the protocol $\P$ provides $\P(k,B,\sigma_B) = \E_k$ exactly for all $k$.

Since we are interested in studying irreversibility in the use of $B$, we define the back-action on the environment, given by $\R_k (\sigma_A) = \tr_{A'} \V_k (\rho_A \otimes \sigma_B)$. Given these details we state a precise a notion of repeatability as follows.

\begin{definition}
	Let $\E:\B(\H_{A_1}) \rightarrow \B(\H_{A'_1})$ be a quantum process on $A_1$, and let $B$ be any other quantum system. We say that a protocol $\mathcal{P}$ for $\E$ is 2-repeatable if given any system $A_2$ isomorphic to $A_1$ it specifies symmetric processes $\mathcal{V}_1:\B(\h_{A_1}\otimes\h_{B})\rightarrow \B(\h_{A'_1}\otimes\h_{B'})$ and $\mathcal{V}_{2}:\B(\h_{A_2}\otimes\h_{B'})\rightarrow \B(\h_{A_2'}\otimes\h_{B''})$  such that for all states $\sigma_B \in \B(\h_{B})$
	\begin{align}
	&\Tr_{A_1,B''}(\mathcal{V}(\rho_{1}\otimes\rho_2\otimes\sigma_B))=\tilde{\E}(\rho_2)\\ \nonumber
	&\Tr_{A_2,B''}(\mathcal{V}(\rho_{1}\otimes\rho_2\otimes\sigma_B))=\tilde{\E}(\rho_1).
	\end{align}
	with $\mathcal{V}:=(\iden_{A}\otimes\V_2)\circ(\V_1\otimes\iden_{A_2})$, and $\tilde{\E}$ being the approximation to $\E$ using $\sigma_B$.
\end{definition}

An elementary aspect of the repeatable use of some symmetry-breaking resource $\sigma_B$ to induce a map $\E$ on a system $A$ is that a subsequent use will also result in exactly the same quantum process. One can easily extend the above definition to $n$-repeatability where the same process $\E$ is induced on $n$ identical systems using the same initial system $B$ and we expand upon this in the following definition.
\begin{definition}
	Let $A_1,... A_{n}$ be $n$ isomorphic systems $A_{i}\cong A$ and $\E:\B(\h_{A})\longrightarrow \B(\h_{A'})$ a target process. We say that the protocol $\mathcal{P}$ for $\E$ using system $B$ is $n$-repeatable if it specifies a circuit of symmetric operations $\mathcal{W}=\mathcal{V}_n\circ\mathcal{V}_{n-1}\circ...\circ \mathcal{V}_{1}$ with $\mathcal{V}_{i}:\B(\h_{A_{i}}\otimes\h_{B_{i-1}})\longrightarrow \B(\h_{A'_i}\otimes \h_{B_{i}})$ for all $i$ initially acting on $B_0=B$ such that for any $\sigma\in\h_{B}$ and any $k$:
	\begin{equation}
	\Tr_{\smallsetminus k,B_{n}}(\mathcal{W}(\rho_1\otimes\rho_2\otimes...\otimes\rho_{n}\otimes \sigma))=\tilde{\E}(\rho_{k}) 
	\label{induced}
	\end{equation}
	the induced process $\tilde{\E}$ is the same on all subsystems $A_{k}$ and is an approximation of $\E$ using $\sigma$. In particular $\E$ will depend on $\sigma$ but not on $k$.
\end{definition}

This scenario is also depicted in the main text figure in which a sequence or ``circuit'' of symmetric interactions $\V_1, \V_2, \dots \V_n$ are performed on $B$ so as to induce local processes on subsystems $A_1, A_2, \dots A_n$. 

We wish to study when the state $\sigma_B$ can be used in an arbitrarily repeatable way, namely $n$-repeatable for any $n \in \mathbb{N}$. To determine this we can consider for any fixed total state on the systems $A_1, \dots A_n$, the induced process $\F_k$ under the protocol map $\V = \V_n \circ \cdots \circ \V_1$ from the system $B$ into any system $A_k$, which describes the transfer of reference frame data needed to induce the local process on $A_k$.

The induced process for fixed input state $\rho_{1 \dots n}$ on $A_1, \dots, A_n$ is given by
\begin{equation}
\F_k(\sigma_B) = \tr_{\smallsetminus k} \V( \rho_{1 \dots n} \otimes \sigma_B)
\end{equation}
where $\Tr_{\smallsetminus k}$ denotes discarding all systems except $A_k$.

However for the particular case of $ \rho_{1\dots n}= \rho^{\otimes n}$, the $n$-repeatability implies that $\F_k = \F$ for all $k = 1, \dots n$, and so in this case the protocol results in the same process from the reference frame $B$ into each of the subsystems.

In entanglement theory one has the notion of an $n$-extendible state, which gives a simple measure of the entanglement in the state. However this notion can be generalised to quantum processes, and relates directly to our present discussion.
\begin{definition}
	A quantum process $\F:\B(\H_B)\rightarrow \B(\H_A)$ is said to be \emph{$n$-extendible} if there exists a quantum process $\Lambda:\B(\H_B)\rightarrow \B(\H_{A_1}\otimes\H_{A_2}...\otimes \H_{A_n})$ with $\H_{A_i}\cong\H_{A}$ for all $i$, such that  for all $X \in \B(\H_B)$ $\Tr_{\smallsetminus i}\Lambda(X)=\F(X)$. 
\end{definition}
Therefore we see that $n$-repeatability of a protocol involving $B$ to simulate $\E$ implies the induced process $\F$ from $B$ to the output system $A'$ must be $n$-extendible. 
\begin{lemma}
	Suppose that system $B$ admits an n-repeatable use on systems $A_1,...A_{n}$ then for any fixed $\rho$ the resulting process on $B$ given by $\F_{\rho}:\B(\h_{B})\longrightarrow \B(\h_{A'_{k}})$:
	\begin{equation}
	\F_{\rho}(\sigma):=\Tr_{\smallsetminus k,B_{n}}(\mathcal{W}(\rho^{\otimes n}\otimes \sigma))
	\end{equation}
	is an $n$-extendible map.
	\label{nextendible}
\end{lemma}
The following lemma says that an operation that is $n$-extendible for all finite $n$ must be a measure and prepare.

\begin{lemma}
	A quantum operation $\mathcal{F}:\B(\h_B)\longrightarrow \B(\h_{A'})$ is n-extendible for all finite $n$ if and only if it is a measure and prepare process. Equivalently there is a POVM set $\{M_{a}\}$ and quantum states $\rho_{a}\in\B(\h_{A'})$ such that:
	\begin{equation}
	\mathcal{F}(\sigma_{B})=\sum_{a} \Tr(M_{a}\sigma_{B})\rho_{a}
	\end{equation}
	for all $\sigma\in \B(\h_A)$.
	\label{extendibility}
\end{lemma}
\begin{proof}
	A bipartite state is $n$-extendible for all finite $n$ it must be a separable state. Whenever $\mathcal{F}$ is n-extendible for all $n$ then its corresponding Choi operator $J[\mathcal{F}]$ is also n-extendible for all finite $n$ and therefore is separable. A Choi operator is separable if and only if the the corresponding process is entanglement breaking. Moreover an entanglement-breaking process has the form of a measure and prepare and therefore $\F$ takes the form stated above. 
\end{proof}

Given this, we can establish the following general constraint on any protocol that admits arbitrarily repeatable use of a resource $\sigma_B$.
\begin{theorem}\label{Theorem-MP}
	Let $B$ be a quantum system with Hilbert space $\H_B$. For any fixed $\sigma\in \B(\H_B)$ used as a reference frame, suppose that for every finite $n$ there is an n-repeatable circuit of global symmetric process that induces $\E:\B(\H_A)\rightarrow \B(\H_{A'})$ on a quantum system $A$. Then there exists a POVM $\{M_a\}$ on $\H_B$ and completely positive maps $\Phi_{a}:\B(\H_A)\rightarrow \B(\H_{A'})$ such that:
	\begin{equation}
	\E(\rho)=\sum_{a} \tr(M_{a}\sigma)\Phi_{a}(\rho).
	\end{equation}
\end{theorem}
\begin{proof}
	The previous Lemma \ref{nextendible} implies that the effective process on the environment $\mathcal{F}_{\rho}:\B(\h_{B}) \longrightarrow\B(\h_{A}')$ given for any fixed $\rho$ by $\mathcal{F}_{\rho}(\sigma_{B})=\Tr_{\smallsetminus k,B_{n}}(\mathcal{W}(\rho\otimes\sigma_{B})$ is n-extendible for all finite $n$. This statement is independent on the initial state on system $A$ (although the process $\mathcal{F}_{\rho}$ generally may not be so). Moreover $\mathcal{F}_\rho$ is a valid CPTP map as it arises as a composition of CPTP maps.
	By lemma \ref{extendibility} every n-extendible quantum operation for all finite $n$ must take the form of a measure and prepare process. Therefore for each fixed $\rho\in\B(\h_{A})$ there exists a POVM $\{N_i\}$ on $\h_{B}$ and quantum states $\rho_i\in \B(\h_{A}'$) such that:
	\begin{equation}
	\mathcal{F}_{\rho}(\sigma)=\sum\limits_{a} \Tr(N_i\sigma)\rho_i.
	\end{equation}
	Note that there could be dependence on $\rho$ in either $\rho_i$ or $N_i$, however this can be simplified by noting that one can decompose any POVM into a convex combination of extremal POVMs. We can write
	\begin{equation}
	N_i = \sum_k p_k M_{k,i}
	\end{equation}
	where $\mathcal{M}_k = (M_{k,i})$ is an extremal POVM for each $k$, and $p_k$ is a probability distribution. This implies that
	\begin{align}
	\Tr_{\smallsetminus k,B_{n}}(\mathcal{W}(\rho^{\otimes n}\otimes \sigma))&=\sum_{i, k}p_k  \Tr(M_{k,i}\sigma)\rho_i \\
	&= \sum_{i,k}\Tr(M_{k,i}\sigma) \Phi_{k,i}( \rho)
	\end{align}
	where $\Phi_{k,i}(\rho) := p_k \rho_i$ is a completely-positive linear map on $\rho$ (since it always returns up to normalization a valid quantum state), and which implies that the POVM acting on $B$ is independent of the input state $\rho$ on $A$. Introducing the single index $a=(k,i)$ completes the proof.
\end{proof}

This places a strong constraint on the repeatable use of an environment $B$ to induce maps on other systems. The content is easy to understand -- if the environment acts as a reference system for an arbitrary number of systems then the only information that can be used must be classical information \cite{Brandao:2015aa,Brandao:2017aa}. 

It is important to emphasize that this result applies to a circuit that induces  a \emph{single} quantum process $\E$ on the system $A$ -- simply interrogate the system $B$ once, copy the measurement information and propagate it to an unbounded number of systems to induce $\E$. This says nothing about whether the system $B$ (which might be finite dimensional) suffers irreversibility in the process. In order to determine this we must consider protocols that use $B$ for a second, independent quantum process $\E'$. 

In the next subsection we use our earlier results on the decomposition of quantum processes to analyse coherence protocols and provide an account of how in the case of a non-commutative symmetry $G$ that fundamental incompatibility in the use of symmetry-breaking resources is expected to arise.

\subsection{Process mode picture: The repeatable use of quantum coherence.}
\label{S-catalytic-theorem}
We can now apply the process mode formalism to the question of the repeatable use of symmetry-breaking resources. We first look at a quantum subsystem $A$, and the task of inducing a target quantum process $\E$ on that system under a $U(1)$ quantum coherence constraint. We shall determine precisely when we can induce $\E$ in an arbitrarily repeatable protocol using a coherent environment $B$. As discussed, the orbit $\M(G,\E)$ of $\E$ under the group action encodes the reference frame data required from the rest of the global system.  For the case of coherence $G=U(1)$ it is clear that $\M(G,\E)$ is either a point, for $\E$ being a symmetric process on $A$, or is a circle, when $\U_g \circ \E \circ \U_g^\dagger \ne \E$. The set of processes decompose under $U(1)$ into a basis of process modes $\{\Phi^\lambda\}$, which are one-dimensional since the group is abelian. Again we absorb any multiplicities into the $\lambda$ label for clarity of the exposition.

We are free to pick any point on $\M(G,\E)$ which corresponds to a reference quantum process $\E_0$, and decompose this map in terms of $\E_0=\sum_{\lambda} \alpha_{\lambda}(\E_0)\Phi^{\lambda}$, with data $\{\alpha_\lambda(\E_0) \in \mathbb{C}\}$ for the reference operation $\E_0$. With respect to $\E_0$ the target map $\E$ is obtained from $\E_0$ via some group transformation $\theta \in U(1)$, as $\E = \U_\theta \circ \E_0 \circ \U_\theta^\dagger$, and so
\begin{equation}
\E(\rho)=\sum_{\lambda \in \rm{Irrep}(A,A')} \alpha_{\lambda}(\E_0)e^{i\lambda \theta} \Phi^{\lambda}(\rho).
\label{dec}
\end{equation}

To implement any such target map one needs to specify the group element $\theta$ and the constant invariant data $\alpha_{\lambda}(\E_{0})$. Note that a perfect classical reference frame for the group U(1) is provided by the infinite dimensional space of wavefunctions on a circle $\mathcal{L}^{2}(S^{1})$. Using such a system for the environment $B$ means that an arbitrarily repeatable protocol can perfectly achieve the target $\E$ by performing a measurement that estimates the group element and transmitting the classical result to arbitrary many systems. The non-trivial issue is to determine what form such protocols can take under a U(1) global symmetry constraint. This is our aim for the current section. 

In particular we restrict to the case when $B$ is an infinite-dimensional ladder system $\h_{\rm{ladder}}$ for two reasons. First there is a fundamental connection between $\mathcal{L}^{2}(S^1)$ and $\h_{\rm ladder}$ such that the latter can also perfectly encode group elements $\theta\in U(1)$ into quantum states $\ket{\theta}\in\h_{\rm{ladder}}$ that can be perfectly discriminated. This has technical subtleties for the continuous group (since we work in a separable Hilbert space), but in terms of the eigenstates $\{\ket{n}\}_{n\in\mathbb{Z}}$ we can consider the following set of orthonormal `states' that encode any $\theta\in U(1)$
\begin{equation}
\ket{\theta}:=(2\pi)^{-1/2}\sum_{n\in \mathbb{Z}} e^{-in\theta}\ket{n},
\label{classicalreference}
\end{equation}
which should be understood as being meaningful in a distributional sense. In terms of $\mathcal{L}^{2}(S^1)$ this amounts to viewing the Dirac delta distribution $\delta(x - \theta)$ to be a normalized `wavefunction' on $S^1$.

We will refer to the set of all $\{\ket{\theta}\}_{\theta\in U(1)}$ as \emph{asymptotic reference frames}. For an in depth analysis of this connection and further useful properties of these states we refer the reader to the excellent book \cite{Busch}. The second reason why we restrict $B$ to be $\h_{\rm{ladder}}$ is that we want to have an asymptotic classical limit in the environment such that an initial state $\sigma_{B}=\ket{\theta}\bra{\theta}$ induces the target operation $\E=\U_{\theta}\circ\E_{0}\circ\U_{\theta}\hc$ that is associated with the point $\theta$ on the process orbit $\M(\E_{0},G)$.

We can now state the main result of this section.
\begin{theorem} 
	A protocol $\mathcal{P}$ that induces a local process $\E$ on $A$ using a ladder system $B$ satisfies:\\
	(a) Global U(1) symmetry.\\
	(b) Arbitrary repeatability.\\
	(c) Asymptotic reference frames on $B$ are not disturbed.\\
	(d) Asymptotic reference frames on $B$ yield perfect simulations.\\
	if and only if $\P$ is a catalytic coherence protocol.
\end{theorem}
This provides a clear physical interpretation of the repeatable use of quantum coherence in simple physical terms. Note it does \emph{not} imply that the system $B$ is in some perfectly coherent state, or that the state of $B$ stays the same -- the repeatability holds irrespective of the state on $B$. The proof of this result is straightforward using process modes, and is given as follows.

For the case where $B$ is a finite-dimensional system, a similar result can be established, but with additional qualifications. The back-action $\R$ on $B$ must (in general) map into a slightly larger system $C$ in order to maintain the repeatability condition (or otherwise it can only be $n$-repeatable for some finite $n$). This is discussed in section \ref{finitedimrep} of the Supplementary Material.
For finite $d$-dimensional subsystems $A$, we also see that the POVM required only involves modes no larger than $d$. For example for the case of $d=2$ we have that
\begin{equation}
M_a = x_a \I + y_a  \cos \hat{\phi} + z_a \sin \hat{\phi},
\end{equation}
where $x_a,y_a,z_a \in \mathbb{R}$ and we have $\Delta^{\pm 1} = \cos \hat{\phi} \pm i \sin \hat{\phi}$ being the only interaction term required between $A$ and the environment. 

More generally, from the perspective of $\M(G,\E)$ this corresponds to the fact that the target map $\E$ only requires a resolution of the target point on $\M(G,\E)$ to an angular scale at worst $\delta \theta \sim \frac{2\pi}{d}$, and so is as efficient as possible in the use of the reference. We can therefore view the protocol as involving a coarse-grained measurement of $\hat{\phi}$, which can be made without disturbing a subsequent measurement. 

In the case of coherence there is a single coordinate $\hat{\phi}$ that must be extracted from the environment up to some resolution, and therefore the simulations of different coherent maps $\{\E_1, \E_2, \dots \}$ on $A$ are essentially equivalent. This is no longer true for more general symmetry groups, as we discuss shortly.

\subsection{Interpretation of coherent protocol as broadcasting of reference frame data}
One can understand this result from another informal perspective, which perhaps helps clear up some confusion that might exist on catalytic coherence. The system $B$ can be continually reused, and its state will change continually under the protocol. Despite this, its ability to function as a coherence reference remains the same. One might feel that this clashes with cloning intuition -- namely quantum resources cannot be copied in general. However the coherence protocol should not be viewed as a cloning of reference frame data, but as the \emph{broadcasting of reference frame data} to multiple systems. Broadcasting is a mixed state version of cloning in which one wishes to copy unknown quantum states $\{ \rho_1, \dots, \rho_n\}$ to multiple other parties. In the single copy case a state $\rho_k$ is transformed to a bipartite $\sigma_{AB}$, such that the marginals are $\sigma_A = \rho_k$ and $\sigma_B = \rho_k$. It is known \cite{PhysRevLett.76.2818} that a set of quantum states $\{\rho_k\}$ may be broadcast perfectly if and only if $[\rho_i, \rho_j] = 0$ for all $i,j$. 

The relevance for us here is that the coherent properties of the environment $B$ are fully described by the expectation values $\<\Delta^k\> := \tr [ \Delta^k \sigma_B]$, and so we need only consider these degrees of freedom. However $[\Delta^k, \Delta^j] =0 $ for all $j,k$ and so a state of the form $\sigma_B = \frac{1}{d} (\I + \sum_k c_k \Delta^k + \mbox{other terms})$ can have the $\Delta^k$ components of the state broadcast in the sense described. What is non-trivial to establish, is that this can be done under a global symmetry constraint. The classicality of the underlying data is the key point, and explains how the protocol works from a different perspective. This is also consistent with our analysis in the next section concerning the structure of symmetric bipartite processes.

\begin{figure}
	\includegraphics[width=9cm]{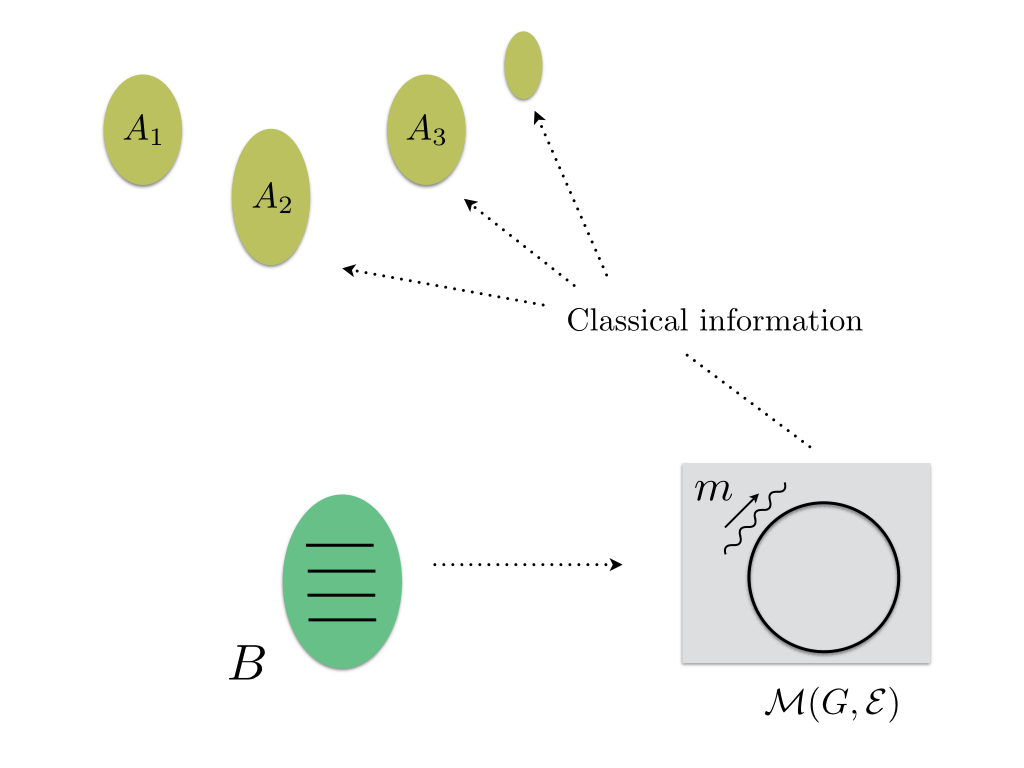}
	\caption{\textbf{Repeatable use of quantum coherence as the broadcasting of classical data.} For a non-symmetric $\E$ and quantum coherence, the process orbit $\M(U(1),\E)$ is a circle. All relevant operators are built from the shift operator $\Delta = e^{i\hat{\phi}}$, corresponding to the position observable $\hat{\phi}$ on the process orbit.  The only relevant parameters in a state $\rho$ of a reference system $B$ are the expectation values $\<\Delta^m\>$. Since $[\Delta^m, \Delta^n] = 0$ for all $n,m$ it is possible to perfectly broadcast the classical data $\{\< \Delta^m\>\}$ to any number of systems $A_1 , A_2 \dots A_n$. }
\end{figure}

\subsection{Irreversibility under general symmetry constraints and a geometric perspective.}
Given the analysis we have provided for quantum coherence, we might wonder if a similar construction applies for more general groups. For such cases, there is one simple way in which an environment $B$ can be used in a repeatable way -- namely we can embed the system's Hilbert space into the space of wavefunctions on $G$ and perform the measurement that estimates groups elements $\{|g\>\<g|\}$ on this infinite dimensional space. Since this extracts all the reference data from $B$ into a classical form it can be copied and repeatedly used. However this assumes a very particular interaction, and that $B$ can physically be embedded in the required infinite-dimensional system (which is a non-trivial assumption). 

We can therefore ask if repeatability can occur for a general group $G$ and a finite dimensional system $B$? For simplicity we can restrict to $G=SU(2)$ and consider just the set of all axial processes as our target quantum processes. As already described, the process orbit $\M(G,\E)$ for these quantum processes is the $2$-sphere $S^2$, with coordinates $(\theta, \phi)$. Now if arbitrary repeatability is present then we have by the same analysis that $\E = \sum_a \tr (M_a \sigma) \E_a$ where the POVM elements on $B$ must supply the coordinates on $\M$ via the condition
\begin{equation}
\sum_a c_{j,m, a} \tr(M_a \sigma_B)= a_j (\E) Y_{jm}(\theta, \phi),
\end{equation}
where $\{a_j(\E)\}$ are the invariant data for the process orbit, and $c_{j,a}$ are the coefficients of $\E_a$ in the $SU(2)$ process mode basis. 

However, now an important distinction is made with the $U(1)$ coherence case. The POVM that extracts the reference data from $B$ must estimate a point on a sphere. In the classical limit one can have perfect resolution of any point $(\theta, \phi) \in S^2$, however for finite dimensional $B$ it is impossible to provide a perfect encoding on the point. Moreover, we know that quantum mechanics on $S^2$ is a phase space and so in the case that $B$ is finite dimensional there will be a non-trivial uncertainty relation present. If, for example, $B$ is a $d$-dimensional spin, then one has operators $\hat{X}_i:= \frac{r}{d^2-1}J_i$ for $B$ that constitute non-commuting coordinates such that $\hat{X}_1^2+\hat{X}_2^2+\hat{X}_3^2 = r^2$. This defines the so-called ``fuzzy sphere'' \cite{fuzzy} in non-commutative geometry where one has a discrete representation of spherical geometry. In the $d\rightarrow \infty$ limit this coincides with classical geometry, however for finite $d$ has a fundamental lower bound on resolution and complementarity in measurements.

Therefore if one is using the system $B$ within some globally symmetric process to represent $\M(G,\E) \cong S^2$, then the complementary in measurements on this phase space will imply incompatibility in the use of symmetry-breaking resources. This incompatibility is not present for coherence, since essentially only $\hat{\phi}$ is needed to supply the reference data. 

More generally, the process orbit perspective suggests a form of quantum-mechanical irreversibility in the use of symmetry breaking resources that depends on whether the geometry that can be induced on $\M(G,\E)$ by $B$ is non-commuting or not. This is consistent with the asymptotic limit of classical reference frames in quantum theory for an arbitrary group $G$, and also with the case of quantum coherence, however we must leave any further analysis to later work.

\end{document}